\documentclass{sig-alternate}
\usepackage{color}
\usepackage{subfig}
\usepackage{array}
\usepackage{algorithm}
\usepackage{algorithmic}

\ifodd 0
\newcommand{\rev}[1]{{{\color{blue} #1}}} 
\else
\newcommand{\rev}[1]{#1}
\fi
\begin{document}
\makeatletter
\let\@copyrightspace\relax
\makeatother
\title{Online Energy Generation Scheduling for Microgrids with Intermittent Energy Sources and Co-Generation}

\numberofauthors{4} 
%
\author{
%
%
\alignauthor
Lian Lu\titlenote{The first two authors are in alphabetical order.}\\
       \affaddr{The Department of Information Engineering}\\
       \affaddr{The Chinese University of Hong Kong}\\
\alignauthor
Jinlong Tu\raisebox{10pt}{$\ast$}\\
      \affaddr{The Department of Information Engineering }\\
       \affaddr{The Chinese University of Hong Kong}\\
\alignauthor
Chi-Kin Chau\\
         \affaddr{Masdar Institute of Science and Technology}\\
\and  
\alignauthor
Minghua Chen\titlenote{Corresponding author.}\\
       \affaddr{The Department of Information Engineering}\\
       \affaddr{The Chinese University of Hong Kong}\\
\alignauthor
Xiaojun Lin\\
       \affaddr{School of Electrical and Computer Engineering}\\
       \affaddr{Purdue University}\\
}


\maketitle

\begin{abstract}
Microgrids represent an emerging paradigm of future electric power
systems that can utilize both distributed and centralized generations.  Two
recent trends in microgrids are the integration of local renewable
energy sources (such as wind farms) and the use of co-generation ({\em i.e.,}
to supply both electricity and heat).
%
However, these trends also bring
%
%
unprecedented challenges to the design of intelligent control
strategies for microgrids.
Traditional generation scheduling paradigms \rev{rely on} perfect prediction of
future electricity supply and demand. \rev{They} are no longer applicable to microgrids
with unpredictable renewable energy supply and \rev{with} co-generation (that
needs to consider both electricity and heat demand).
In this paper, we study online
algorithms for the microgrid generation scheduling problem with
intermittent renewable energy sources and co-generation, \rev{with the goal of
maximizing} the cost-savings with local generation.
Based on the insights from the structure of the offline optimal solution, we
propose a class of competitive online algorithms, called \textsf{CHASE}
(Competitive Heuristic Algorithm for Scheduling Energy-generation),
that track the
offline optimal in an online fashion. Under typical settings, we show
that \textsf{CHASE} achieves the best competitive ratio among all deterministic
online algorithms, and the ratio is no larger than a small constant 3.
%
We also extend our algorithms to intelligently leverage on {\em limited
prediction} of the future, such as near-term demand or wind forecast. By
extensive empirical evaluations using real-world traces, we show that our
proposed algorithms can achieve near offline-optimal performance. In a
representative scenario, \textsf{CHASE} leads to around 20\% cost reduction with no
future look-ahead, and the cost reduction increases with
the future look-ahead window.


\end{abstract}

\category{C.4}{PERFORMANCE OF SYSTEMS}Modeling techniques; Design studies
\category{F.1.2}{Modes of Computation}Online computation
\category{I.2.8}{Problem Solving, Control Methods, and Search}Scheduling

\terms{Algorithms, Performance}


\section{Introduction} \label{sec:intro}
Microgrid is a distributed \rev{electric power
system}
that can autonomously co-ordinate local generations and demands in a
dynamic manner
\cite{lasseter2004microgrid}.
Illustrated in Fig.~\ref{fig:microgrid}, modern microgrids often
consist of distributed renewable energy generations ({\em e.g.,} wind farms) and
co-generation technology ({\em e.g.,} supplying both electricity and heat
locally).
Microgrids can operate in either grid-connected mode or islanded mode. 
There have been worldwide deployments of
pilot microgrids, such as the US, Japan, Greece and Germany \cite{barnes2007real}.

\begin{figure}[htp!]
 	\begin{center}
	\includegraphics[scale=0.1]{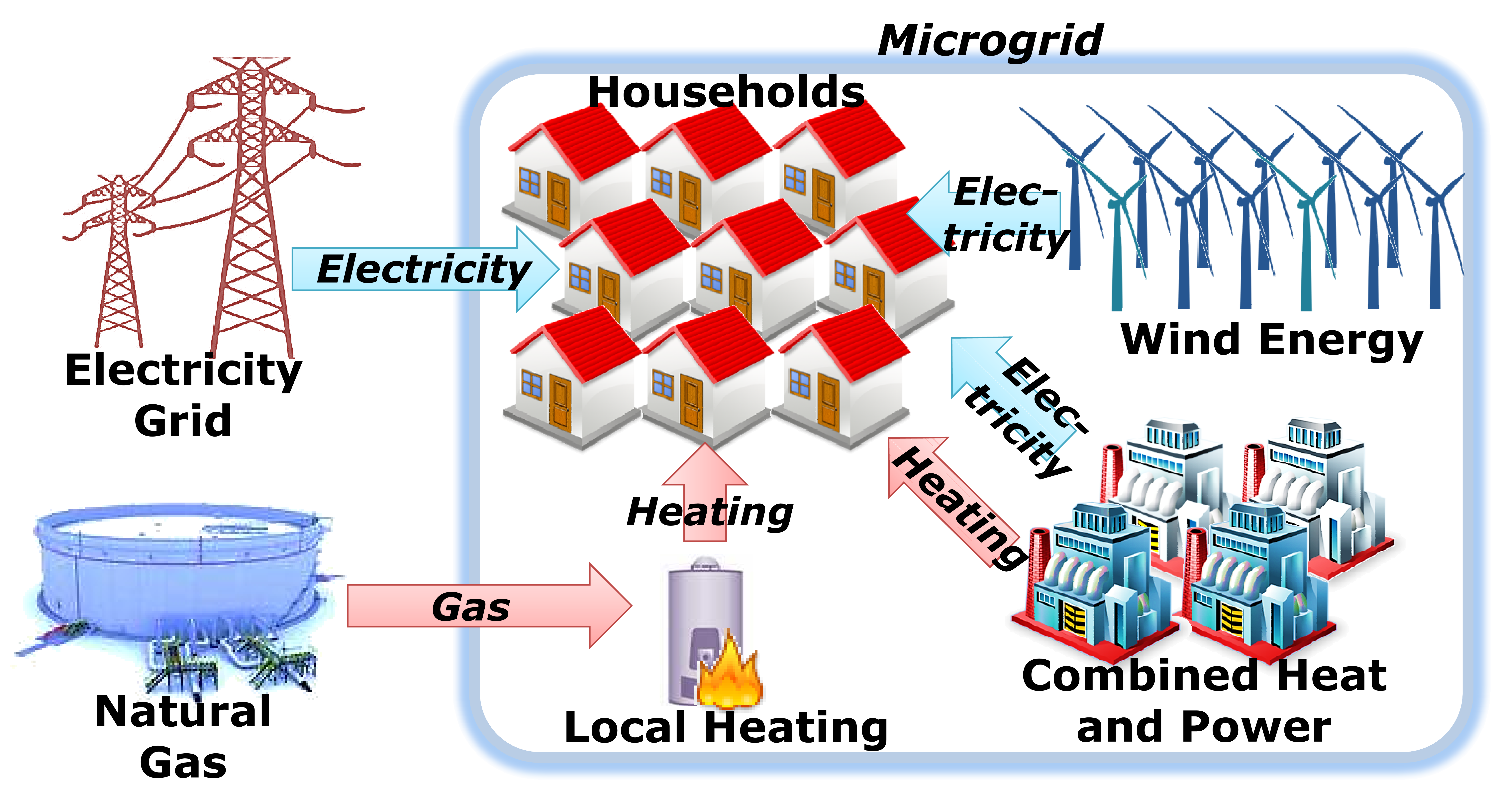}	
	\caption{An illustration of a typical microgrid.} \label{fig:microgrid} 
	\end{center}
\end{figure}

Microgrids are more robust and cost-effective than traditional
approach of centralized grids. They represent an emerging paradigm of future electric power
systems~\cite{marnay2007microgrids} that address the following two critical challenges.

\emph{Power Reliability}. Providing reliable and quality power is
critical both socially and economically. In the US alone, while the \rev{electric power
system} is 99.97\% reliable,
each year the economic loss due to power outages is at least \$150 billion~\cite{DoEReport}. However,  enhancing
power reliability across a large-scale power grid is very challenging
\cite{chowdhury2003reliability}. With 
local generation, microgrids can supply energy
\emph{locally} as needed, effectively alleviating the negative effects of
power outages.

\emph{Integration with Renewable Energy}. The growing environmental
awareness and government directives lead to the increasing
penetration of renewable energy. For example, the US aims at 20\% wind
energy penetration by 2030 to ``de-carbonize'' the power system.
Denmark targets at 50\% wind generation by 2025. However,
incorporating a significant portion of intermittent renewable energy
poses great challenges to grid stability, which requires
a new thinking of
how the grid should operate~\cite{varaiya2011smart}.
In traditional centralized grids, the actual locations of conventional energy generation,
renewable energy generation ({\em e.g.,} wind farms), and energy
consumption are usually distant from each other. Thus,
the need to coordinate conventional energy generation and
consumption based on the instantaneous variations of renewable energy
generation  leads to challenging stability problems. In
contrast,
in microgrids
renewable energy is generated and consumed in the \emph{local} distributed
network. Thus, the uncertainty of renewable energy is absorbed locally,
minimizing its negative impact on
the stability of the central transmission networks.


Furthermore,
microgrids bring significant economic benefits, especially with the
augmentation of combined heat and power
(CHP) generation technology. In traditional
grids, a substantial amount of residual energy after electricity
generation is often wasted. In contrast, in microgrids this residual
energy can be used to supply heat domestically.
By simultaneously satisfying electricity and heat demand using CHP generators, microgrids can often be much more economical
than
\rev{using external electricity supply and separate heat supply}
\cite{hawkes2009modelling}.

However, to realize the maximum benefits of microgrids, intelligent scheduling of both local generation and demand
must be established.
Dynamic demand scheduling in response to supply condition, also called
\emph{demand response} \cite{DoEReport,chiuelectric}, is one of the
useful approaches. But, demand response alone may be
insufficient to compensate the highly volatile fluctuations of wind
generation.
Hence, intelligent generation scheduling, \rev{which orchestrates both local and external generations to satisfy the time-varying energy demand,} is indispensable for the
viability of microgrids. Such generation-side scheduling must
simultaneously meet two goals. (1) To maintain grid stability,
the aggregate supply from CHP
generation, renewable energy generation, the centralized grid, and a separate heating system must meet the
aggregate electricity and heat demand. {(We do not consider the option of using energy storage in the paper, {\em e.g.,} to charge at low-price periods and to discharge at high-price periods. This is because for the typical size of microgrids, {\em e.g.,} a college campus,  energy storage systems with comparable sizes are very expensive and not widely available.)} 
(2) It is highly desirable that the microgrid can coordinate local generation and external energy procurement to minimize the overall cost of meeting the energy demand.

We note that a related generation scheduling problem has been
extensively studied for the traditional grids, involving both Unit
Commitment \cite{padhy2004unit} and Economic
Dispatch \cite{gaing2003particle}, which we will review in Sec.~\ref{sec:related} as related work. In a typical
power plant, the generators are often subject to several
operational constraints. For example, steam turbines have a slow ramp-up speed.
In order to perform generation scheduling, the utility company usually needs to
forecast the demand first. Based on this forecast, the utility company then
solves an offline problem to schedule different types of generation
sources in order to minimize cost subject to the operational
constraints. 

Unfortunately, this classical strategy does not work well for the microgrids
due to the following unique challenges introduced by the renewal energy sources
and co-generation.
%
%
The first challenge is that microgrids powered by intermittent renewable energy generations will face a significant
uncertainty in energy supply. Because of its smaller scale, abrupt
changes in local weather condition may have a dramatic impact that
cannot be amortized as in the wider national scale. 
In Fig.~\ref{fig:intro_trace_electricity}, we examine one-week traces of electricity demand for a college in
San Francisco~\cite{CEUS} and power output of a nearby wind station~\cite{NREL}.
We observe that although the electricity demand has a relative regular pattern
for prediction, the net electricity demand inherits
a large degree of variability from the wind generation, casting a
challenge for accurate prediction.
\begin{figure}[htp!] %
{\centering
\subfloat[ Elec. demand, wind gen.]{
\includegraphics[width=0.5\columnwidth]{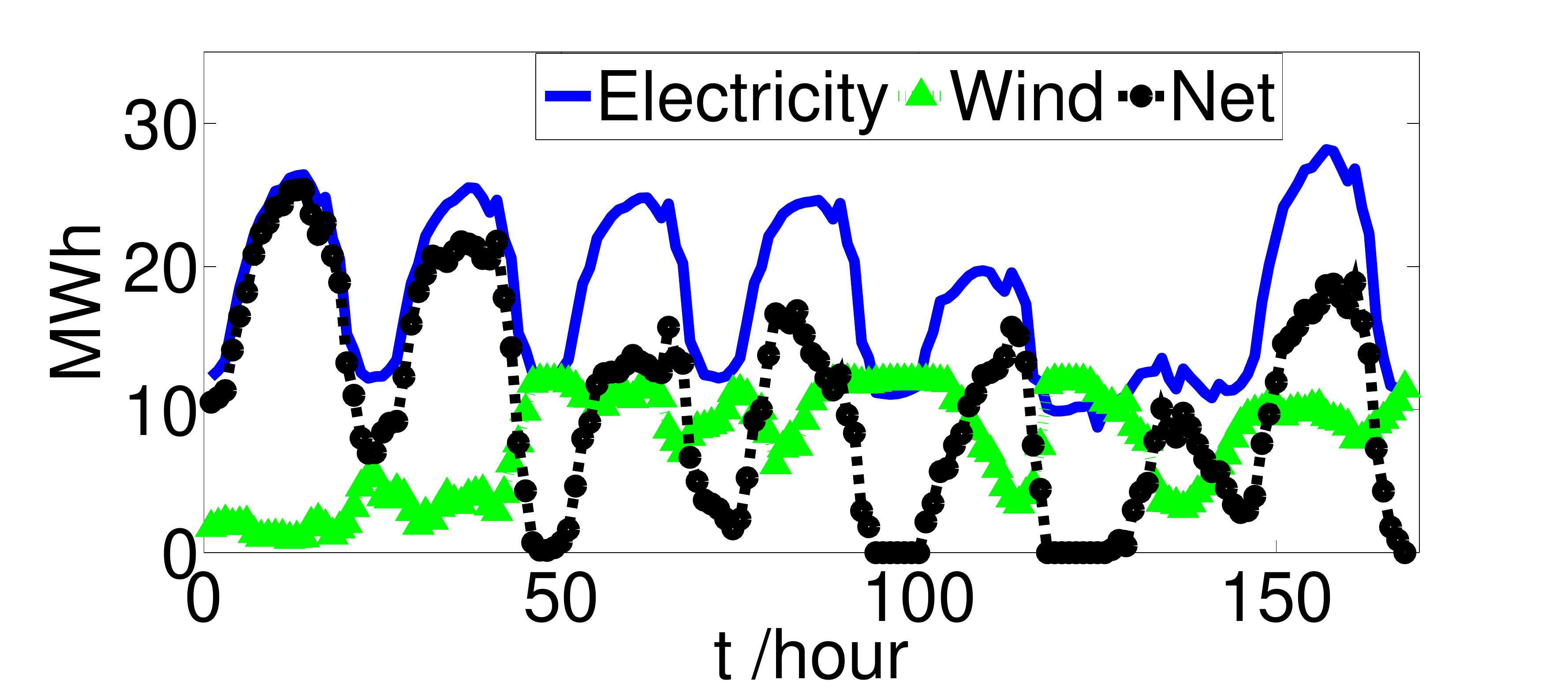}
\label{fig:intro_trace_electricity}}\subfloat[ Heat demand]{
\includegraphics[width=0.5\columnwidth]{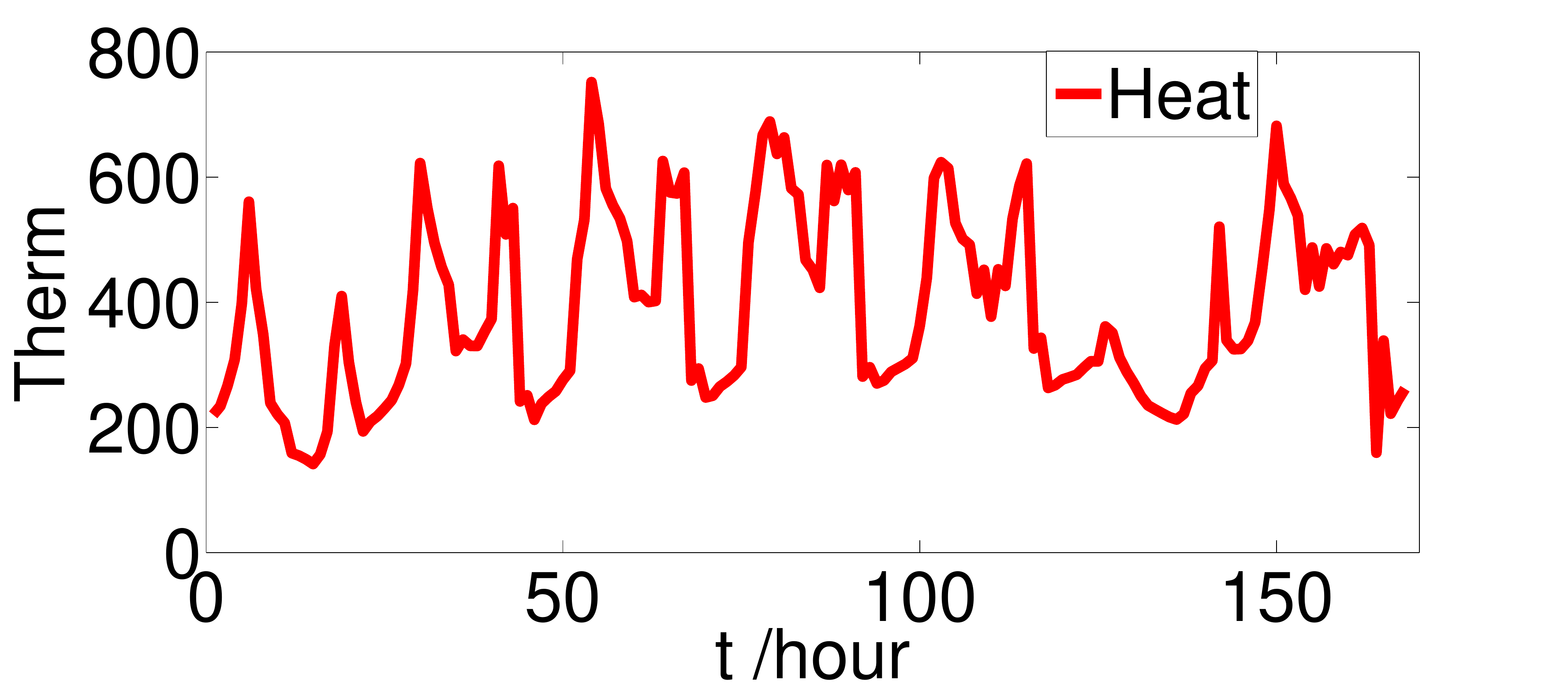} 
\label{fig:intro_trace_heat}}\caption{ Electricity demand, heat demand and wind generation in a week. In (a), the net demand is computed by subtracting the wind generation from the electricity demand.}
} 
\end{figure}

Secondly, co-generation brings a new dimension of uncertainty
in scheduling decisions. Observed from Fig.~\ref{fig:intro_trace_heat}, the heat demand exhibits a different stochastic
pattern that adds difficulty to the prediction of overall
energy demand.

Due to the above additional variability, traditional energy
generation scheduling
based on \emph{offline}
optimization assuming \emph{accurate} prediction
of future supplies and demands cannot be applied to the microgrid
scenarios. On the other hand,
%
%
there are also new opportunities. In microgrids there are usually only 1-2 types of
small reciprocate generators from tens of kilowatts to several megawatts.
These generators are typically gas or diesel powered and can be fired up
with large ramping-up/down level in the order of minutes. For example, a diesel-based engine can be powered up in 1-5 minutes and has a maximum ramp up/down rate of 40\% of its capacity per minute \cite{vuorinen2007planning}.
The ''fast responding'' nature of these local generators opens up opportunities to increase
the frequency of generator on/off scheduling that substantially changes the design space for energy generation scheduling.

Because of these unique challenges and opportunities, it remains an open
question of how to design effective strategies for scheduling energy
generation for microgrids.
\\
\\
\\
\subsection{Our Contributions}

In this paper, we formulate a general problem of energy generation
scheduling for microgrids.
Since both the future demands and future renewable energy
generation are difficult to predict, we use competitive analysis and
study online algorithms that can perform provably well under arbitrarily
time-varying (and even adversarial) future trajectories of demand and
renewable energy generation. Towards this end,
we design a class of simple and effective strategies for
energy generation scheduling named $\textsf{CHASE}$ (in short for Competitive Heuristic Algorithm
for Scheduling Energy-generation). Compared to traditional prediction-based and offline optimization
approaches, our online solution has the following salient benefits. \rev{First, $\textsf{CHASE}$ gives an absolute performance guarantee without the knowledge of supply and demand behaviors.} This minimizes the
impact of inaccurate modeling and the need for expensive data gathering, and
hence improves robustness in microgrid operations. Second, $\textsf{CHASE}$ works
without any assumption on gas/electricity prices and policy regulations.
This provides the grid operators flexibility for operations
and policy design without affecting the energy generation strategies for microgrids.

We summarize the key contributions as follows:
\begin{enumerate}
\item In Sec.~\ref{sec:offline}, we devise an
offline optimal algorithm for a generic formulation of \rev{the} energy generation
scheduling \rev{problem} that models most microgrid scenarios with intermittent energy
sources and fast-responding gas-/diesel-based CHP generators. \rev{Note that the offline problem is challenging by itself because it is a mixed integer problem and the objective function values across different slots are correlated via the startup cost.} We first reveal an elegant
structure of the single-generator problem and exploit it to construct
the optimal offline solution. The structural insights are further
generalized in
Sec.~\ref{sec:ngens} to the case with $N$ homogeneous generators.
The optimal offline solution employs a simple load-dispatching
strategy where each generator separately solves a partial scheduling
problem.

\item In Secs.~\ref{sec:online}-\ref{sec:ngens}, we build upon the
structural insights from the offline solution to design
$\textsf{CHASE}$, a deterministic online algorithm for scheduling energy
generations in microgrids. We name our algorithm $\textsf{CHASE}$
because it tracks the offline optimal solution
in an online fashion. We show that $\textsf{CHASE}$ achieves a competitive ratio
of $\min\left(3-2\alpha,1/\alpha\right)\leq 3$.
In other words, no matter how the demand, renewable energy
generation and grid price vary, the cost of $\textsf{CHASE}$ without
any future information is guaranteed to be no greater than
$\min\left(3-2\alpha,1/\alpha\right)$ times the offline optimal assuming
complete future information. Here the constant $\alpha=(c_o + c_m / L )/(P_{\max}+\eta\cdot c_g)\in (0,1]$
captures the maximum price discrepancy between using local
generation and external sources to supply energy.
We also prove that the above competitive ratio is the best possible for
any deterministic online algorithm.

\item
The above competitive ratio is attained without any future information
of demand and supply.
In Sec.~\ref{sec:lookahead}, we then extend $\textsf{CHASE}$ to
intelligently leverage limited look-ahead information, such as near-term
demand or wind forecast, to further improve its
performance. In particular, $\textsf{CHASE}$ achieves an improved competitive
ratio of $\min\left(3-2\cdot g(\alpha,\omega),1/\alpha\right)$ when
it can look into a future window of size $\omega$. Here,
the function $g(\alpha,\omega)\in [\alpha, 1]$ captures the benefit of looking-ahead and
monotonically increases from $\alpha$ to $1$ as $\omega$ increases. Hence, the larger the
look-ahead window, the better the performance.
In Sec.~\ref{sec:slowrep}, we also extend $\textsf{CHASE}$ to the case
where generators are governed by several additional operational
constraints ({\em e.g.,} ramping up/down rates and minimum on/off periods),
and derive an upper bound for the corresponding competitive ratio.

\item In Sec.~\ref{sec:empirical}, by extensive evaluations using real-world traces, we
show that our algorithm $\textsf{CHASE}$ can achieve satisfactory empirical performance and is robust to look-ahead error. In particular, a small look-ahead window is sufficient to achieve near offline-optimal performance. Our offline (resp., online) algorithm
achieves a cost reduction of 22\% (resp., 17\%) with CHP technology. 
The cost reduction is computed in comparison with the baseline cost achieved by using only the wind generation, the central
grid, and a separate heating system. The substantial cost reductions show the economic benefit of microgrids in addition to its potential in improving energy reliability. Furthermore, interestingly, deploying a partial local generation capacity that provides 50\% of the peak local demands can achieve 90\% of the cost reduction. This provides strong motivation for microgrids to deploy at least a partial local generation capability to save costs.
\end{enumerate}


\section{Problem Formulation}\label{sec:pf}
\begin{table}[htp!] 
{
\begin{tabular*}{1\columnwidth}{@{\extracolsep{\fill}}l|>{\raggedright}p{0.78\columnwidth}}
\hline
\textbf{Notation}{{} } & \textbf{Definition}
\tabularnewline
\hline
$T$  & The total number of intervals (unit: min)\tabularnewline
$N$  & The total number of local generators \tabularnewline
$\beta$  & The startup cost of local generator (\$)\tabularnewline
$c_{m}$  & The sunk cost per interval of running local generator (\$)\tabularnewline
$c_{o}$  & The incremental operational cost per interval of running local generator
to output an additional unit of power
(\$/Watt)\tabularnewline
$c_{g}$ & The price per unit of heat obtained externally using natural gas
(\$/Watt)\tabularnewline
${\sf T}_{{\rm on}}$  & The minimum on-time of generator, once it is turned on \tabularnewline
${\sf T}_{{\rm off}}$  & The minimum off-time of generator, once it is turned off \tabularnewline
${\sf R}_{{\rm up}}$  & The maximum ramping-up rate (Watt/min)\tabularnewline
${\sf R}_{{\rm dw}}$  & The maximum ramping-down rate (Watt/min)\tabularnewline
$L$  & The maximum power output of generator (Watt)\tabularnewline
$\eta$ & The heat recovery efficiency of co-generation \tabularnewline
$a(t)$  & The net power demand (Watt)\tabularnewline
$h(t)$ & The space heating demand (Watt)\tabularnewline
$p(t)$  & The spot price per unit of power obtained from the electricity grid
($P_{\min}\leq p(t)\leq P_{\max}$) (\$/Watt)\tabularnewline
$\sigma(t)$ & The joint input at time $t$: $\sigma(t) \triangleq (a(t), h(t), p(t))$ \tabularnewline
$y_n(t)$  & The on/off status of the $n$-th local generator (on as ``1'' and off
 as ``0''), $1\leq n\leq N$ \tabularnewline
$u_n(t)$  & The power output level when the $n$-th generator is on (Watt), $1\leq n\leq N$ \tabularnewline
$s(t)$ & The heat level obtained externally by natural gas (Watt)\tabularnewline
$v(t)$  & The power level obtained from electricity grid (Watt)\tabularnewline
\hline
\end{tabular*} 
Note: we use bold symbols to denote vectors, {\em e.g.,} $a\triangleq(a(t))_{t=1}^{T}$.} { Brackets indicate the units.}
\caption{\rev{Key notations. } }
\label{tab:notations}
\end{table}

We consider a typical scenario where a microgrid orchestrates different energy generation sources to minimize
cost for satisfying both local electricity and heat demands simultaneously, while meeting
operational constraints of \rev{electric power
system}.  We will formulate a microgrid cost minimization problem (\textbf{MCMP}) that incorporates intermittent energy demands, time-varying electricity prices, local generation capabilities and co-generation.


{We define the notations in Table~\ref{tab:notations}. We also define the acronyms for our problems and algorithms in Table~\ref{tab:acronyms}.}
\begin{table}[htp!] %
{
\begin{tabular*}{0.95\columnwidth}{@{\extracolsep{\fill}}l|>{\raggedright}p{0.75\columnwidth}}
\hline
\textbf{ Acronym}{{} } & \textbf{ Meaning}
\tabularnewline
\hline
$\textbf{MCMP}$  & Microgrid Cost Minimization Problem\tabularnewline
$\textbf{fMCMP}$ & $\textbf{MCMP}$ for \rev{fast-responding generators} \tabularnewline
$\textbf{fMCMP}_{\rm s}$ & $\textbf{fMCMP}$ \rev{with single fast-responding generator} \tabularnewline
$\textbf{SP}$ & A \rev{simplified version} of $\textbf{fMCMP}_{\rm s}$  \tabularnewline
$\textsf{CHASE}_{s}$ & The \rev{baseline version} of $\textsf{CHASE}$ for $\textbf{fMCMP}_{\rm s}$  \tabularnewline
$\textsf{CHASE}_{s+}$ & $\textsf{CHASE}$ for $\textbf{fMCMP}_{\rm s}$ \tabularnewline
$\textsf{CHASE}_s^{{\rm lk}(\omega)}$ & The \rev{baseline version} of $\textsf{CHASE}$ for $\textbf{fMCMP}_{\rm s}$ with look-ahead \tabularnewline
$\textsf{CHASE}_{s+}^{{\rm lk}(\omega)}$ & $\textsf{CHASE}$ for $\textbf{fMCMP}_{\rm s}$ with look-ahead  \tabularnewline
$\textsf{CHASE}^{\rm lk(\omega)}$ & $\textsf{CHASE}$ for $\textbf{fMCMP}$ with look-ahead  \tabularnewline
$\textsf{CHASE}^{\rm lk(\omega)}_{\sf gen}$ & $\textsf{CHASE}$ for $\textbf{MCMP}$ with look-ahead\tabularnewline
\hline
\end{tabular*} 
}
\caption{Acronyms for problems and algorithms.} 
\label{tab:acronyms}
\end{table}
\subsection{\label{sub:Model-Assumptions}Model}
\textbf{Intermittent Energy Demands}: We consider arbitrary renewable energy supply ({\em e.g.,} wind farms). Let the net demand ({\em i.e.,} the residual electricity demand not balanced by wind generation) at time $t$ be $a(t)$.
Note that we do not rely on any specific stochastic model of $a(t)$.

\textbf{External Power from Electricity Grid}: The microgrid can obtain external
 electricity supply from the central grid for unbalanced electricity demand in an on-demand manner. We let the spot
price at time $t$ from electricity grid be $p(t)$. We assume that $P_{\min}\leq p(t)\leq P_{\max}$.
Again, we do not rely on any specific stochastic model on $p(t)$.

\textbf{Local Generators}: The microgrid has $N$ units of homogeneous local generators, each having an maximum power output capacity $L$. Based on a common generator model~\cite{kazarlis1996genetic}, we denote $\beta$ as the startup cost of turning on a generator. Startup cost $\beta$ typically involves the heating up cost (in order to produce high pressure \rev{gas or} steam to drive the engine) and the \rev{time-amortized} additional maintenance costs resulted from each startup ({\em e.g.,} fatigue and possible permanent damage resulted by stresses during startups)\footnote{\rev{It is commonly understood that power generators incur startup costs and hence the generator on/off scheduling problem is inherently a dynamic programming problem. However, the detailed data of generator startup costs are
often not revealed to the public. According to \cite{cullen2011dynamic} and the references therein, startup costs of gas generators vary from several hundreds to thousands of US dollars. Startup costs at such level are comparable to running generators at their full capacities for several hours.}}. We denote $c_{m}$
as the sunk cost of maintaining a generator in its active state
per unit time, and $c_{o}$ as the operational cost per unit time
for an active generator to output an additional unit of energy.
Furthermore, a more realistic model of generators considers advanced {\em operational constraints}:
\begin{enumerate}

\item {\em Minimum On/Off Periods}: If one generator
has been committed (resp., uncommitted) at time $t$, it must remain committed
(resp., uncommitted) until time $t+{\sf T}_{{\rm on}}$ (resp., $t+{\sf T}_{{\rm off}}$).

\item {\em Ramping-up/down Rates}: The incremental power output in two consecutive time intervals is
limited by the ramping-up and ramping-down constraints.
\end{enumerate}
Most microgrids today employ generators powered by gas turbines or diesel engines. These generators are ``fast-responding'' in the sense that they can be powered up in several minutes, and have small minimum on/off periods as well as large ramping-up/down rates. Meanwhile, there are also generators based on steam engine, and are ``slow-responding'' with non-negligible ${\sf T}_{{\rm on}}$, ${\sf T}_{{\rm off}}$, and small ramping-up/down rates.

\textbf{Co-generation and Heat Demand}: The local CHP generators can
simultaneously generate electricity and useful heat. Let the heat
recovery efficiency for co-generation be $\eta$, {\em i.e.,} for each
unit of electricity generated, $\eta$ unit of useful heat can be
supplied for free.
Alternatively, without co-generation, heating can be generated separately
using external natural gas, which costs $c_{g}$ per unit time.
Thus, $\eta c_{g}$ is the saving due to using co-generation to supply
heat, provided that there is sufficient heat demand.
We assume $c_o \geq \eta \cdot c_g $.
In other words,
it is cheaper to generate heat by natural gas than purely by generators
(if not considering the benefit of co-generation).
Note that a system with no co-generation can be viewed as a special
case of our model by setting $\eta =0$. Let the heat demand at time $t$ be $h(t)$.

To keep the problem interesting, we assume that $c_{o}+\frac{c_{m}}{L}<P_{\max}+\eta\cdot c_{g}$.
This \rev{assumption} ensures that the minimum co-generation energy cost is cheaper
than the maximum external energy price. \rev{If this was not the case, it would have been optimal to always obtain power and heat externally and separately.}

\subsection{Problem Definition} \label{ssec:problem_definition}

We divide a finite time horizon into $T$ discrete time slots,
each is assumed to have a unit length without loss of generality.
The microgrid operational cost in $[1,T]$ is given by
\begin{align} 
&{\rm Cost}(y,u,v,s) \triangleq\textstyle{\sum}_{t=1}^{T}\Big\{{\displaystyle p(t)\cdot v(t)+c_g \cdot s(t) +} \Big. \label{eq:microgrid.cost}\\
&\;\;\;\;\;  \left. \textstyle{\sum}_{n=1}^{N}\left[c_{o}\cdot u_{n}(t)+c_{m}\cdot y_{n}(t)+\beta[y_n(t)-y_n(t-1)]^{+}\right]\right\}, \notag
\end{align}
which includes the cost of grid electricity, the cost of the external gas,
and the operating and switching cost of local CHP generators in the entire horizon $[1,T]$. Throughout this paper, we set the initial condition $y_n(0)=0$, $1\leq n\leq N$.

We formally define the \textbf{MCMP}
as a mixed integer programming problem, given electricity demand $a$, heat demand $h$, and grid electricity price $p$ as time-varying inputs:
\begin{subequations}
\begin{align}
 \underset{{y,u,v,s}}{\min}&{\rm Cost}(y,u,v,s) \\
 \mbox{s.t.}\;& 0\leq u_n(t)\leq L\cdot y_n(t), \label{C_max_output}\\
 & \textstyle{\sum}_{n=1}^N u_n(t)+v(t)\geq a(t), \label{C_e-demand}\\
 & \eta \cdot \textstyle{\sum}_{n=1}^N u_n(t)+s(t)\geq h(t), \label{C_h-demand}\\
& u_n(t)-u_n(t-1)\leq{\sf R}_{{\rm up}}, \label{C_ramp_up}\\
& u_n(t-1)-u_n(t)\leq{\sf R}_{{\rm dw}}, \label{C_ramp_down}\\
& y_n(\tau) \geq \mathbf{1}_{\{y_n(t)>y_n(t-1)\}}, t\mbox{+}1\leq \tau \leq t\mbox{+}{\sf T}_{{\rm on}}\mbox{-}1,\label{C_min_on}\\
 & y_n(\tau) \leq 1\mbox{-}\mathbf{1}_{\{y_n(t)<y_n(t-1)\}}, t\mbox{+}1\leq \tau \leq t\mbox{+}{\sf T}_{{\rm off}}\mbox{-}1, \label{C_min_off}\\
 \mbox{var}\;& y_n(t)\in \{0,1\},u_n(t),v(t),s(t)\in \Bbb{R}_0^{+}, n\in [1, N], t\in [1,T],\nonumber
\end{align}
\end{subequations}
where $\mathbf{1}_{\{\cdot\}}$ is the indicator function and
$\Bbb{R}_0^{+}$ represents the set of non-negative numbers. The
constraints are similar to those in the power system literature and
capture the operational constraints of generators. Specifically, constraint \eqref{C_max_output} captures the constraint of maximal output of the local generator. Constraints \eqref{C_e-demand}-\eqref{C_h-demand} ensure that the demands of electricity and heat can be satisfied, respectively. Constraints \eqref{C_ramp_up}-\eqref{C_ramp_down} capture the constraints of maximum ramping-up/down rates. Constraints \eqref{C_min_on}-\eqref{C_min_off} capture the minimum on/off period constraints (note that they can also be expressed in linear but hard-to-interpret forms).


\section{Fast-Responding Generator Case} \label{sec:fastrep}

\newtheorem{lem}{Lemma}
\newdef{defn}{Definition}
\newdef{thm}{Theorem}
\newdef{cor}{Corollary}
This section considers the fast-responding generator scenario. Most CHP generators employed
in microgrids are based on gas or diesel. These generators can be fired up in several minutes and
have high ramping-up/down rates. Thus at the timescale of energy
generation (usually tens of minutes), they can be considered as
having no minimum on/off periods and ramping-up/down rate constraints. That is,
${\sf T}_{{\rm on}}=0$, ${\sf T}_{{\rm off}}=0$, ${\sf R}_{{\rm up}}=\infty$,
${\sf R}_{{\rm dw}}=\infty$.
We remark that this model captures most microgrid scenarios today.
We will extend the algorithm developed for this responsive generator scenario to the general generator scenario in Sec.~\ref{sec:slowrep}.

To proceed, we first study a simple case where there is one unit of generator. We then extend the results to $N$ units of homogenous generators in Sec.~\ref{sec:ngens}.

\subsection{Single Generator Case} \label{sec:singlefastrep}
We first study a basic problem that considers a single generator. Thus, we can drop the subscript $n$ (the index of the
generator) when there is no source of confusion:
\begin{subequations}
\begin{align}
\mathbf{fMCMP}_{\rm s}: \underset{{y,u,v,s}}{\min}&{\rm Cost}(y,u,v,s) \label{basic_prob.obj}\\
 \mbox{s.t.}\;& 0\leq u(t)\leq L\cdot y(t), \label{C_max_output_single}\\
 & u(t)+v(t)\geq a(t), \label{C_e-demand_single}\\
 & \eta \cdot u(t)+s(t)\geq h(t), \label{C_h-demand_single}\\
 \mbox{var}\;& y(t)\in \{0,1\},u(t),v(t),s(t)\in \Bbb{R}_0^{+}, t\in [1,T].\nonumber
\end{align}
\end{subequations}
Note that even this simpler problem is challenging to solve. First,
even to obtain an offline solution (assuming complete knowledge of
future information), we must solve a mixed integer optimization
problem.  Further, the objective function values across different slots are
correlated via the startup cost $\beta[y(t)-y(t-1)]^+$, and thus cannot
be decomposed. {Finally, to obtain an online solution we do not even know
the future.}

\textbf{Remark}: \rev{Readers familiar with online
server scheduling in data centers
\cite{lin2011dynamic,lu2012simple} may see some similarity between our problem and
those in \cite{lin2011dynamic,lu2012simple}, {\em i.e.,} all are dealing with the scheduling difficulty introduced
by the switching cost. Despite such
similarity, however, the
inherent structures of these problems are
significantly different.
First, there is only one category of demand ({\em i.e.,} workload to be satisfied by the servers) in online server scheduling problems. \rev{In contrast,} there are two \rev{categories} of demands ({\em i.e.,} electricity and heat demands) in our problem. \rev{Further,} because of co-generation, they can not be considered separately. Second, there is only one category of supply  ({\em i.e.,} server service capability) in online server scheduling problem, \rev{and} thus the demand must be satisfied by this single supply. However, in our problem, there are three different supplies, including local generation, electricity grid power and external heat supply. Therefore, the design space in our problem \rev{is larger and it requires} us to orchestrate three different supplies, instead of single supply, to satisfy the demands.
}

Next, we introduce the following lemma to simplify the structure of the
problem. Note that if  $(y(t))_{t=1}^{T}$ are given, the startup cost is
determined. Thus, the problem in
\eqref{basic_prob.obj}-\eqref{C_h-demand_single} reduces to a linear
programming and can be solved \emph{independently in each time slot}.
\begin{lem}
\label{lem:fMCMP} Given $\left(y(t)\right)_{t=1}^{T}$ and the input $\left(\sigma(t)\right)_{t=1}^{T},$
the solutions $\left(u(t),v(t),s(t)\right)_{t=1}^{T}$ that minimize
$\mathrm{Cost}(y,u,v,s)$ are given by:
\begin{equation} 
u(t)\mbox{=}\begin{cases}
0, & \mbox{if } p(t)+\eta \cdot c_{g}\leq c_{o}\\
\min\Big\{\frac{h(t)}{\eta},a(t),L\cdot y(t)\Big\}, &  \mbox{if } p(t)<c_{o}<p(t)\mbox{+}\eta \cdot c_{g}\label{eq:optimal_u}\\
\min\Big\{a(t),L\cdot y(t)\Big\}, &  \mbox{if } c_{o}\leq p(t)
\end{cases}
\end{equation}
and
\begin{equation}
v(t)=\left[a(t)-u(t)\right]^{+},\;\;s(t)=\left[h(t)-\eta \cdot u(t)\right]^{+}. \label{eq:optimal_v_and_s}
\end{equation}
We note in each time slot $t$, the above $u(t),$ $v(t)$ and $s(t)$
are computed using only $y(t)$ and $\sigma(t)$ in the same time
slot.
\end{lem}

The result of
Lemma~\ref{lem:fMCMP} can be interpreted as follows. If the grid price is
very high ({\em i.e.,} higher than $c_o$), then
it is always more economical to use local generation as much as
possible, without even considering heating. However, if the grid price
is between $c_o$ and $c_o-\eta\cdot c_g$, local electricity generation alone
is not economical. Rather, it is the benefit of supplying heat through
co-generation that makes local generation more economical. Hence, the amount of
local generation must consider the heat demand $h(t)$. Finally, when the
grid price is very low ({\em i.e.,} lower than $c_o-\eta\cdot c_g$), it is always
more cost-effective not to use local generation.

\rev{As a consequence of Lemma~\ref{lem:fMCMP}, the problem
$\mathbf{fMCMP_{s}}$ can be simplified to
the following problem $\mathbf{SP}$, where we only need to consider the decision of turning on ($y(t)=1$) or
off ($y(t)=0$) the generator.}
\begin{align*}
 \mathbf{SP}: & \min_y \; {\rm Cost}(y) \\
 & \mbox{var}\;\; y(t)\in \{0,1\}, t\in [1,T],
\end{align*}
where
\begin{align*}
{\rm Cost}(y)\triangleq & \textstyle{\sum}_{t=1}^{T}\Big( \psi\big(\sigma(t),y(t)\big)+\beta\cdot[y(t)-y(t-1)]^{+} \Big),
\end{align*}
$\psi\big(\sigma(t),y(t)\big)\triangleq c_{o}u(t)+p(t)v(t)+c_{g}s(t)+c_{m}y(t)$ and \\ $(u(t),v(t),s(t))$ are defined according to Lemma~\ref{lem:fMCMP}.


\subsubsection{Offline Optimal Solution}  \label{sec:offline}

We first study the offline setting, where the input
$(\sigma(t))_{t=1}^{T}$ is given ahead of time. We will reveal an elegant
structure of the optimal solution. Then, in Section~\ref{sec:online} we will
exploit this structure to design an efficient online algorithm.

\rev{The problem $\mathbf{SP}$ can be solved by the classical dynamic programming approach. We present it in Algorithm~\ref{alg:DP}. However, the solution provided by dynamic programming does not seem to bring significant insights for developing online algorithms. Therefore, in what follows we study the offline optimal solution from another angle, which directly reveals its structure.}

\begin{algorithm}[htb!]
{\caption{Dynamic Programming for \bf{SP} \label{alg:DP}}
\begin{algorithmic}[1]
\STATE We construct a graph $G=(V,E),$ where each vertex denoted by $\left\langle y,t\right\rangle $
presents the on $\left(y=1\right)$ or off $\left(y=0\right)$ state of the local generation at time $t.$

\STATE We draw a directed edge, from each vertex $\left\langle y(t-1),t-1\right\rangle $
to each possible vertex $\left\langle y(t),t\right\rangle $ to represent
the fact that the local generation can transit from the first state
to the second. Further, we associate the cost of that transition shown
below as the weight of the edge:
\[
\psi\left(y(t),\sigma(t)\right)+\beta\cdot\left(y(t)-y(t-1)\right)^{+}.
\]

\STATE We find the minimum weighted path from the initial state to the
final state by Dijkstra's shortest path algorithm on graph $G.$

\STATE The optimal solution is $\left(y(t)\right)_{t=1}^{T}$ along the
shortest path.
\end{algorithmic}
}
\end{algorithm}

Define
\begin{equation}
\delta(t)\triangleq\psi\Big(\sigma(t),0\Big)-\psi\Big(\sigma(t),1\Big).\label{eqn:delta-definition}
\end{equation}
$\delta(t)$ can be interpreted as the one-slot cost difference between
using or not using local generation.
Intuitively, if $\delta(t)>0$ (resp. $\delta(t)<0$), it will be
desirable to turn on (resp. off) the generator. However, due to the
startup cost, we should not turn on and off the generator too
frequently. Instead, we should evaluate whether the \emph{cumulative}
gain or loss in the future can offset the startup cost. This intuition
motivates us to define the following cumulative cost difference
$\Delta(t)$. We set the initial value as $\Delta(0)=-\beta$ and define
$\Delta(t)$ inductively:
\begin{equation}
\Delta(t)\triangleq\min\Big\{0,\max\{-\beta,\Delta(t-1)+\delta(t)\}\Big\}.\label{eqn:Delta-definition}
\end{equation}
Note that $\Delta(t)$ is only within the range $[-\beta,0]$. Otherwise, the
minimum cap ($-\beta$) and maximum cap (0) will apply to retain $\Delta(t)$
within $[-\beta,0]$. An important feature of $\Delta(t)$ useful later in online algorithm design is that it can be computed given the past and current input $\sigma(\tau)$, $1\leq \tau\leq t$.

Next, we construct critical segments according to $\Delta(t)$, and then classify segments by types. Each type
of segments captures similar episodes of demands. As shown later in Theorem~\ref{thm:OFA-optimal}, it suffices to solve the cost minimization problem over every segment and combine their solutions to obtain an offline optimal solution for the overall problem \textbf{SP}.
\begin{defn}
\label{def:critical-seg} We divide all time intervals in $[1,T]$
into disjoint parts called {\em critical segments}:
\[
[1,T_{1}^{c}],[T_{1}^{c}+1,T_{2}^{c}],[T_{2}^{c}+1,T_{3}^{c}],...,[T_{k}^{c}+1,T]
\]
The critical segments are characterized by a set of {\em critical
points}: $T_{1}^{c}<T_{2}^{c}<...<T_{k}^{c}$. We define each critical
point $T_{i}^{c}$ along with an auxiliary point $\tilde{T_{i}^{c}}$, such
that the pair $(T_{i}^{c},\tilde{T_{i}^{c}})$ satisfies the following
conditions:
\begin{itemize}
\item (Boundary): Either \big($\Delta(T_{i}^{c})=0$ and $\Delta(\tilde{T_{i}^{c}})=-\beta$\big)\\
 or \big($\Delta(T_{i}^{c})=-\beta$ and $\Delta(\tilde{T_{i}^{c}})=0$\big).
\item (Interior): $-\beta<\Delta(\tau)<0$ for all $T_{i}^{c}<\tau<\tilde{T_{i}^{c}}$.
\end{itemize}
In other words, each pair of $(T_{i}^{c},\tilde{T_{i}^{c}})$ corresponds to an interval where $\triangle(t)$ goes from -$\beta$ to $0$ or $0$ to -$\beta$, without reaching the two extreme values inside the interval. For example, $(T_{1}^{c},\tilde{T_{1}^{c}})$ and $(T_{2}^{c},\tilde{T_{2}^{c}})$
in Fig.~\ref{fig:example1} are two such pairs, while the corresponding
critical segments are $(T_{1}^{c},T_{2}^{c})$ and $(T_{2}^{c},T_{3}^{c})$.
It is straightforward to see that all $(T_{i}^{c},\tilde{T_{i}^{c}})$
are uniquely defined, thus critical segments are well-defined.
See Fig.~\ref{fig:example1} for an example.
\end{defn}
\begin{figure}[htp!]
 	\begin{center}
	\includegraphics[width=0.95\columnwidth]{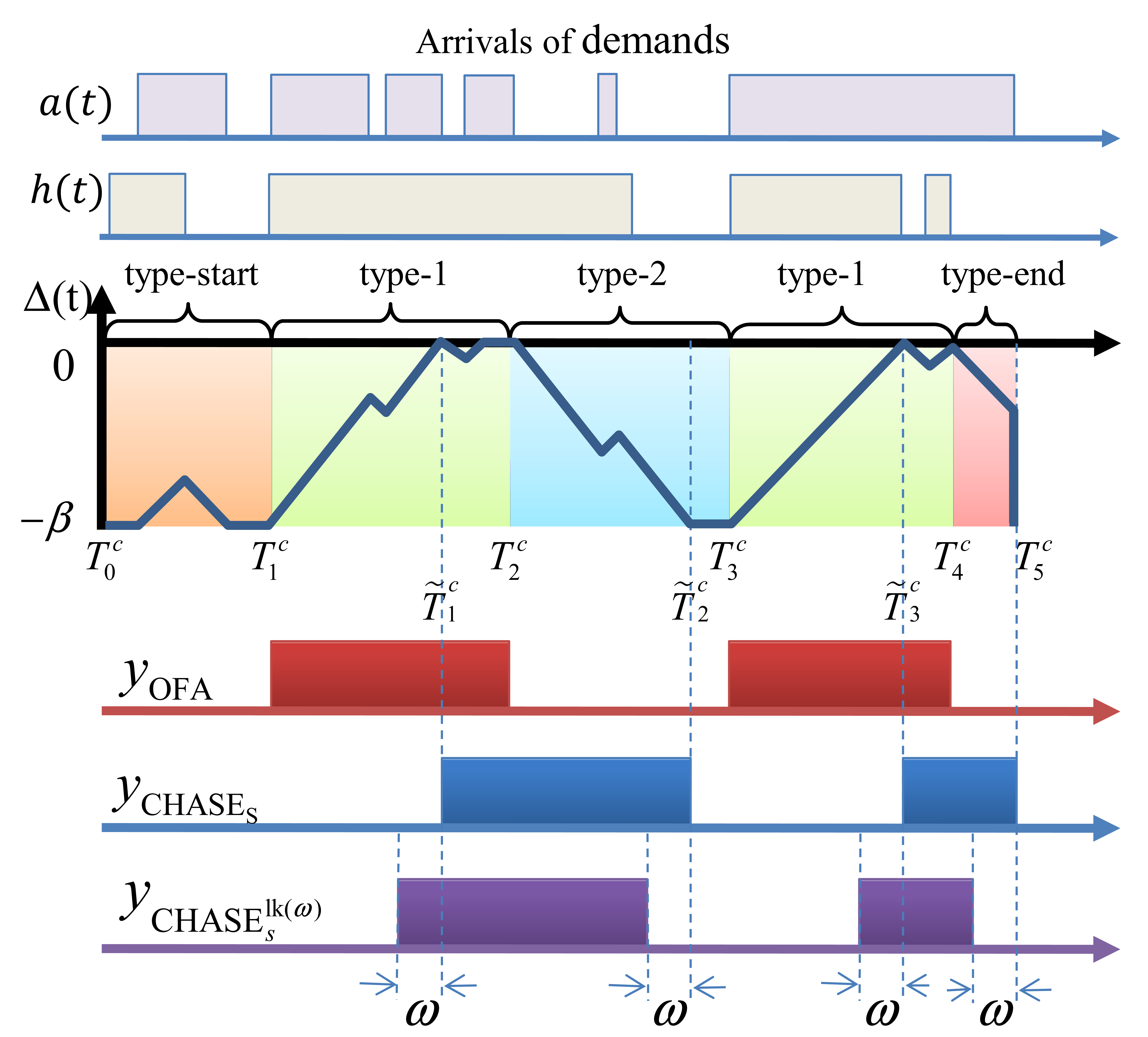}	
	\caption{ An example of $\Delta(t)$, $\mathrm{y_{OFA}}$, $\mathrm{y_{{\sf CHASE}_{s}}}$ and $\mathrm{y_{{\sf CHASE}_{s}^{lk(\omega)}}}$. \rev{In the top two rows, we have $a(t)\in\{0,1\}$, $h(t)\in\{0,\eta\}$. The price $p(t)$ is chosen as a constant in $(c_{o}-\eta\cdot c_{g},c_{o})$. In the next row, we compute $\Delta(t)$ according to $a(t)$ and $h(t)$. For ease of exposition, in this example we set the parameters so that $\Delta(t)$ increases if and only if $a(t)=1$ and $h(t)=\eta$.}
The solutions $\mathrm{y_{OFA}}$, $\mathrm{y_{{\sf CHASE}_{s}}}$ and $\mathrm{y_{{\sf CHASE}_{s}^{lk(\omega)}}}$ \rev{at the bottom rows} are obtained accordingly to \eqref{eq:y_OFA}, Algorithms~\ref{alg:CHASE0} and \ref{alg:CHASE-lk}, respectively.
} \label{fig:example1} 
	\end{center}
\end{figure}
Once the time horizon $[1,T]$ is divided into critical segments,  we can now characterize the optimal solution.
\begin{defn}
\label{def:critical-type}
 We classify the {\em type} of a critical segment by:
\begin{itemize}
\item {\em type-start} (also call {\em type-0}): $[1,T_{1}^{c}]$

\item {\em type-1}: $[T_{i}^{c}+1,T_{i+1}^{c}]$, if $\Delta(T_{i}^{c})=-\beta$
and $\Delta(T_{i+1}^{c})=0$

\item {\em type-2}: $[T_{i}^{c}+1,T_{i+1}^{c}]$, if $\Delta(T_{i}^{c})=0$
and $\Delta(T_{i+1}^{c})=-\beta$

\item {\em type-end} (also call {\em type-3}): $[T_{k}^{c}+1,T]$

\end{itemize}
\end{defn}

We define the cost with regard to a segment $i$ by:\[
\begin{array}{@{}r@{}l@{\ }l}
{\rm Cost}^{{\rm sg\mbox{-}}i}(y) & \triangleq & {\displaystyle \sum_{t=T_{i}^{c}+1}^{T_{i+1}^{c}}\psi\big(\sigma(t),y(t)\big)} +{\displaystyle \sum_{t=T_{i}^{c}+1}^{T_{i+1}^{c}+1}\beta\cdot[y(t)-y(t-1)]^{+}}
\end{array}
\]
and define a subproblem for critical segment $i$ by:
\begin{align*}\label{}
    \mathbf{SP}^{{\rm sg\mbox{-}}i}({\rm y}_{i}^{l},{\rm y}_{i}^{r}): \min\; & {\rm Cost}^{{\rm sg\mbox{-}}i}(y) \\
                                               \mbox{s.t.}\; & y(T_{i}^{c})={\rm y}_{i}^{l},\;y(T_{i+1}^{c}+1)={\rm y}_{i}^{r},\\
                                                \mbox{var}\; & y(t)\in\{0,1\}, t\in[T_{i}^{c}+1,T_{i+1}^{c}].
\end{align*}

Note that due to the startup cost across segment boundaries, in general ${\rm Cost}(y) \neq \sum {\rm Cost}^{{\rm sg\mbox{-}}i}(y)$. In other words, we should not expect that putting together the solutions to each segment will lead to an overall optimal offline solution. However, the following lemma shows an important structure property that \emph{one optimal solution} of $\mathbf{SP^{sg-i}}(y_{i}^{l},y_{i}^{r})$ is independent of boundary conditions $(y_{i}^{l},y_{i}^{r})$ although \emph{the optimal value} depends on boundary conditions.
\begin{lem}
\label{lem:-is-theindep.with.bound.con} $(y_{{\rm OFA}}(t))_{t=T_{i}^{c}+1}^{T_{i+1}^{c}}$
in \eqref{eq:y_OFA} is an optimal solution for ${\bf SP}^{{\rm sg\mbox{-}}i}({\rm y}_{i}^{l},{\rm y}_{i}^{r})$,
despite any boundary conditions $({\rm y}_{i}^{l},{\rm y}_{i}^{r})$.
\end{lem}
This lemma can be intuitively explained by Fig.~\ref{fig:example1}. In type-1 critical
segment, $\Delta(t)$ has an increment of $\beta$, which means
that setting $y(t)=1$ over the entire segment provides at least a benefit
of $\beta$, compared to keeping $y(t)=0$. Such benefit
compensates the possible startup cost $\beta$ if the boundary conditions are not aligned with $y(t)=1$. Therefore,
regardless of the boundary conditions, we should set $y(t)=1$ on type-1
critical segment. Other types of critical segments can be explained similarly.

We then use this lemma to show the following main result on the structure of the offline optimal solution.


\begin{thm}\label{thm:OFA-optimal}
An optimal solution for \textbf{SP} is given by
\begin{equation}\label{eq:y_OFA}
y_{{\rm OFA}}(t)\triangleq
\begin{cases}
0, & \mbox{if\ }t\in[T_{i}^{c}+1,T_{i+1}^{c}]\mbox{\ is type-start/-2/-end},\\
1, & \mbox{if\ }t\in[T_{i}^{c}+1,T_{i+1}^{c}]\mbox{\ is type-1}.
\end{cases}
\end{equation}
\end{thm}
\begin{proof}
Refer to Appendix~\ref{subsec:OFA-optimal}.
\end{proof}
Theorem~\ref{thm:OFA-optimal} can be interpreted as follows. Consider for example a type-1 critical segment in Fig.~\ref{fig:example1} that starts from $T_1^c$. Since $\Delta(t)$ increases from $-\beta$ after $T_1^c$, it implies that $\delta(t)>0$, and thus we are interested in turning on the generator. The difficulty, however, is that immediately after $T_1^c$ we do not know whether the future gain by turning on the generator will offset the startup cost. On the other hand, once $\Delta(t)$ reaches $0$, it means that the cumulative gain in the interval $[T_1^c,\tilde{T}_1^c]$ will be no less than the startup cost. Hence, we can safely turn on the generator at $T_1^c$. Similarly, for each type-2 segment we can turn off the generator at the beginning of the segment. (We note that our offline solution turns on/off the generator at the beginning of each segment because all future information is assumed to be known.)

The optimal solution is easy to compute. More importantly, the insights help us design the online algorithms.

\subsubsection{Our Proposed Online Algorithm ${\sf CHASE}$} \label{sec:online}

Denote an online algorithm for ${\bf SP}$ by ${\cal A}$.
We define the competitive ratio
of ${\cal A}$ by:
\begin{equation}
{\sf CR}({\cal A})\triangleq\max_{\sigma}\frac{{\rm Cost}(y_{{\cal A}})}{{\rm Cost}(y_{{\rm OFA}})}
\end{equation}

Recall the structure of optimal solution $y_{{\rm OFA}}$:
once the process is entering type-1 (resp., type-2) critical segment, we should set
$y(t)=1$ (resp., $y(t)=0$). However, the difficulty
lies in determining the beginnings of type-1 and type-2 critical segments
without future information. Fortunately, as illustrated in Fig.~\ref{fig:example1}, it is certain that the process
is in a type-1 critical segment when $\Delta(t)$ reaches $0$ for the
first time after hitting $-\beta$. This observation motivates us to use the algorithm ${\sf CHASE}_s$, which is given in Algorithm~\ref{alg:CHASE0}. If $-\beta<\Delta(t)<0$,  ${\sf CHASE}_s$ maintains $y(t)=y(t-1)$ (since we do not know whether a new segment has started yet.) However, when $\Delta=0$ (resp. $\Delta(t)=-\beta$), we know for sure that we are inside a new type-1 (resp. type-2) segment. Hence, ${\sf CHASE}_s$ sets $y(t)=1$ (resp. $y(t)=0$). Intuitively, the behavior of ${\sf CHASE}_s$ is to track the offline optimal in an online manner: we change the decision only after we are certain that the offline optimal decision is changed.
\begin{algorithm}[htb!]
{
\caption{ ${\sf CHASE}_s [t, \sigma(t), y(t-1)]$} \label{alg:CHASE0}
\begin{algorithmic}[1]
\STATE find $\Delta(t)$
\IF{$\Delta(t)=-\beta$}
\STATE {$y(t) \leftarrow 0$}
\ELSIF {$\Delta(t)=0$}
\STATE {$y(t) \leftarrow 1$}
\ELSE
\STATE {$y(t) \leftarrow y(t-1)$}
\ENDIF
\STATE set $u(t)$, $v(t)$, and $s(t)$ according to \eqref{eq:optimal_u} and \eqref{eq:optimal_v_and_s}
\STATE return $(y(t), u(t), v(t), s(t))$
\end{algorithmic}
}

\end{algorithm}

Even though ${\sf CHASE}_s$ is a simple algorithm, it has a strong performance guarantee, as given by the following theorem.
\begin{thm} \label{thm:CHASE-competitive-ratio}
The competitive ratio of ${\sf CHASE}_s$ satisfies
\begin{equation}
{\sf CR}({\sf CHASE}_s)\le 3-2\alpha < 3,\label{eq:CHASE_s_Ratio}
\end{equation}
where
\begin{equation}
\alpha \triangleq (c_o + c_m / L )/(P_{\max}+\eta\cdot c_g)\in (0, 1] \label{eq:alpha_def}
\end{equation}
captures the maximum price discrepancy between using local generation and external sources to supply energy.
\end{thm}
\begin{proof}
Refer to Appendix~\ref{subsec:CHASE-competitive-ratio}.
\end{proof}
{\bf Remark}: (i) The intuition that ${\sf CHASE}_s$ is competitive can be explained by studying its worst case input shown in Fig.~\ref{fig:worstcase}. The demands and prices are chosen in a way such that in interval $[T_{0}^{c},T_{1}^{c}]$ $\Delta(t)$ increases from $-\beta$ to $0$, and in
interval $[T_{1}^{c},T_{2}^{c}]$ $\Delta(t)$ decreases from $0$ to $-\beta$. We see that in the worst case, $y_{{\sf CHASE}_s}$ never matches $y_{{\sf OFA}}$. But even in this worst case, ${{\sf CHASE}_s}$ pays only $2\beta$ more than the offline solution ${y_{\sf OFA}}$ on $[T_{0}^{c},T_{2}^{c}]$, while ${y_{\sf OFA}}$ pays at least a startup cost $\beta$ at time $T_{0}^{c}$. Hence, the ratio of the online cost over the offline cost cannot be too bad.
(ii)
Theorem~\ref{thm:CHASE-competitive-ratio} says that ${\sf CHASE}_s$ is more competitive when $\alpha$ is large than it is small. This can be explained intuitively as follows. Large $\alpha$ implies small economic advantage of using local generation over external sources to supply energy. Consequently, the offline solution tends to use local generation less. It turns out ${\sf CHASE}_s$ will also use less local generation\footnote{${\sf CHASE}_s$ will turn on the local generator when $\Delta(t)$ increases to 0. The larger the $\alpha$ is, the slower $\Delta(t)$ increases, and the less likely ${\sf CHASE}_s$ will use the local generator.}  and is competitive to offline solution. Meanwhile, when $\alpha$ is small, ${\sf CHASE}_s$  starts to use local generation. However, using local generation incurs high risk since we have to pay the startup cost to turn
on the generator without knowing whether there are sufficient demands to serve in the future. Lacking future knowledge
leads to a large performance discrepancy between ${\sf CHASE}_s$ and the offline optimal solution, making ${\sf CHASE}_s$ less competitive.


\begin{figure}[t!]

\centering
\includegraphics[width=0.75\columnwidth]{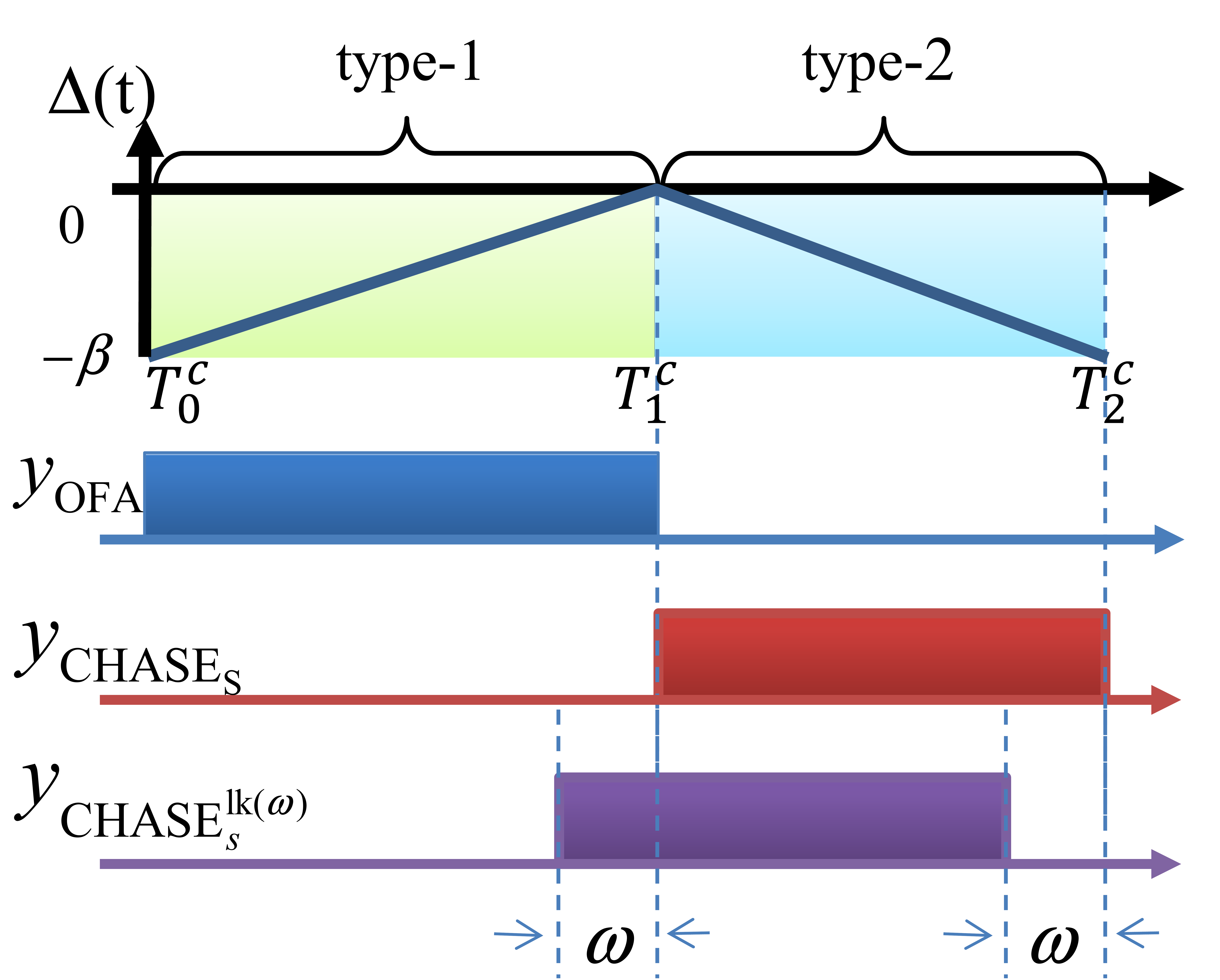}
\caption{\label{fig:worstcase} The worst case input of ${{\sf CHASE}_s}$, and the corresponding ${y_{\sf CHASE_s}}$, $y_{{\sf CHASE}_s^{{\rm lk}(\omega)}}$ and the offline optimal solution ${y_{\sf OFA}}$.}
\end{figure}

The result in Theorem~\ref{thm:CHASE-competitive-ratio} is strong in the sense that ${\sf CR}({\sf CHASE}_s)$ is always upper-bounded by a small constant 3, regardless of system parameters.
This is contrast to large parameter-dependent competitive ratios that one can achieve by using generic approach, {\em e.g.,} the metrical task system framework~\cite{borodin1998online}, to design online algorithms. Furthermore, we show that ${\sf CHASE}_s$ achieves close to the best possible competitive ratio for deterministic algorithms as follow.
\begin{thm}
\label{thm:lower-bound-online} Let $\epsilon>0$ be the slot length under the discrete-time setting we consider in this paper. The competitive ratio for any deterministic
online algorithm ${\cal A}$ for ${\bf SP}$ is lower bounded by
\begin{equation}
{\sf CR}({\cal A})\ge \min (3-2\alpha-o(\epsilon), 1/\alpha), \label{eq:RatioUpperBound}
\end{equation}
where 
$o(\epsilon)$ vanishes to zero as $\epsilon$ goes to zero and the discrete-time setting  approaches the continuous-time setting.
\end{thm}
\begin{proof}
Refer to Appendix~\ref{subsec:lower-bound-online}.
\end{proof}

Note that there is still a gap between the competitive ratios
in (\ref{eq:CHASE_s_Ratio}) and (\ref{eq:RatioUpperBound}).
The difference is due to the term $1/\alpha =
(P_{\max}+\eta\cdot c_{g})/(c_{o}+c_{m}/L)$. This term can be
interpreted as the competitive ratio of a naive strategy that always
uses
external power supply and separate heat supply.
Intuitively, if this $1/\alpha$ term is smaller than $3-2\alpha$, we
should simply use this naive strategy. This observation motivates us to
develop
an improved version of $\textsf{CHASE}_{s}$, called $\textsf{CHASE}_{s+}$,
which is presented in Algorithm~\ref{alg:CHASE}.
Corollary~\ref{cor:chase} shows that $\textsf{CHASE}_{s+}$ closes the above gap
and achieves the asymptotic optimal competitive ratio. Note that whether
or not the $1/\alpha$ term is smaller can be completely determined by the
system parameters.
\begin{algorithm}[htb!]
{
\caption{${\sf CHASE}_{s+} [t, \sigma(t), y(t-1)]$} \label{alg:CHASE}
\begin{algorithmic}[1]
\IF{$1/\alpha \le 3-2\alpha$}
\STATE $y(t)\leftarrow 0,\ \  u(t)\leftarrow 0,\ \  v(t) \leftarrow a(t),\ \  s(t) \leftarrow h(t)$
\STATE return $(y(t), u(t), v(t), s(t))$
\ELSE
\STATE return ${\sf CHASE}_s [t, \sigma(t), y(t-1)]$
\ENDIF
\end{algorithmic}
}
\end{algorithm}
\begin{cor}
\label{cor:chase}
$\textsf{CHASE}_{s+}$ achieves the asymptotic optimal competitive ratio of any deterministic
online algorithm, as
\begin{equation}
{\sf CR}({\sf CHASE}_{s+})\le \min (3-2\alpha, 1/\alpha). \label{eq:CHASE_Ratio}
\end{equation}
\end{cor}

\textbf{Remark}: \rev{At the beginning of Sec.~\ref{sec:singlefastrep}, we have discussed the \rev{structural differences} of online
server scheduling problems \cite{lin2011dynamic,lu2012simple} and ours.
In what follows, we summarize the solution differences among these problems. Note that we share similar intuitions with \cite{lu2012simple}, both make switching decisions when the \emph{penalty cost} equals the switching cost.  The significant difference, however, is when to reset the penalty counting. In \cite{lu2012simple}, the penalty counting is reset when the demand arrives. In contrast, in our solution, we need to reset the penalty counting only when $\Delta(t)$, given in the non-trivial form in (\ref{eqn:Delta-definition}), touches 0 or $-\beta$. This particular way of resetting penalty counting is critical for establishing the optimality of our proposed solution. Meanwhile, to compare with \cite{lin2011dynamic}, the approach in \cite{lin2011dynamic} does not explicitly count the penalty. Furthermore, the online server scheduling problem in \cite{lin2011dynamic} is formulated as a convex problem, while our problem is a mixed integer problem. Thus, there is no known method to apply the approach in \cite{lin2011dynamic} to our problem.}

\subsection{Look-ahead Setting} \label{sec:lookahead}
We consider the setting where the online algorithm can predict a
small window $\omega$ of the immediate future. Note that $\omega=0$ returns to the
case treated in Section~\ref{sec:online}, when there is no future
information at all. Consider again a type-1 segment $[T_1^c,T_2^c]$ in
Fig.~\ref{fig:example1}. Recall that, when there is no future
information, the ${\sf CHASE}_s$ algorithm will wait until $\tilde{T}_1^c$,
{\em i.e.,} when $\Delta(t)$ reaches $0$, to be certain that the offline
solution must turn on the generator. Hence, the ${\sf CHASE}_s$
algorithm will not turn on the generator until this time. Now assume
that the online algorithm has the information about the immediate
future in a time window of length $\omega$. By the time
$\tilde{T}_1^c-w$, the online algorithm has already known that
$\Delta(t)$ will reach $0$ at time
$\tilde{T}_1^c$. Hence, the online algorithm can safely turn on the
generator at time $\tilde{T}_1^c-w$. As a result, the corresponding loss of
performance compared to the offline optimal solution is also reduced.
Specifically, even for the worst-case input in Fig.~\ref{fig:worstcase},
there will be some overlap (of length $\omega$) between $y_{{\sf CHASE}_s}$ and
$y_{{\sf OFA}}$ in each segment. Hence, the competitive ratio should
also improve with future information. This idea leads to the
online algorithm ${\sf CHASE}_s^{{\rm lk}(\omega)}$, which is presented in
Algorithm~\ref{alg:CHASE-lk}.

\begin{algorithm}[htb!]
{\caption{ ${\sf CHASE}_s^{{\rm lk}(\omega)} [t, (\sigma(\tau))_{\tau = t}^{t+w}, y(t-1)]$} \label{alg:CHASE-lk}
\begin{algorithmic}[1]
\STATE find $(\Delta(\tau))_{\tau = t}^{t+w}$
\STATE set $\tau' \leftarrow \min\big\{\tau =t, ...,t+w \mid \Delta(\tau) = 0 \mbox{\ or\ } = -\beta  \big\}$
\IF{$\Delta(\tau')= -\beta$}
\STATE {$y(t) \leftarrow 0$}
\ELSIF {$\Delta(\tau')=0$}
\STATE {$y(t) \leftarrow 1$}
\ELSE
\STATE {$y(t) \leftarrow y(t-1)$}
\ENDIF
\STATE set $u(t)$, $v(t)$, and $s(t)$ according to \eqref{eq:optimal_u} and \eqref{eq:optimal_v_and_s}
\STATE return $(y(t), u(t), v(t), s(t))$
\end{algorithmic}
}
\end{algorithm}

We can show the following improved competitive ratio when limited future information is available.
\begin{thm}
\label{thm:CHASElk-competitive-ratio} The competitive ratio of ${\sf
CHASE}_s^{{\rm lk}(\omega)}$ satisfies
\begin{equation}
{\sf CR}\left({\sf CHASE}_s^{{\rm lk}(\omega)}\right) \le 3-2\cdot g\left(\alpha,\omega\right),\label{eq:LookAheadRatio}
\end{equation}
where $\omega\geq 0$ is the look-ahead window size, {$\alpha\in(0,1]$ is defined in \eqref{eq:alpha_def}}, and
\begin{equation}
g(\alpha,\omega)=\alpha+\frac{(1-\alpha)}{1+\beta\left(Lc_{o}+c_{m}/(1-\alpha)\right)/\left(\omega(Lc_{o}+c_{m})c_{m}\right)}. \label{eq:g_alpha_omega}
\end{equation}
captures the benefit of looking-ahead and monotonically increases from $\alpha$ to 1 as $\omega$ increases. In particular, \\ ${\sf CR}({\sf CHASE}_{s}^{{\rm lk}(0)})= {\sf CR}({\sf CHASE}_{s})$.
\end{thm}
\begin{proof}
Refer to Appendix~\ref{subsec:CHASElk-competitive-ratio}.
\end{proof}
We replace ${\sf CHASE}_s$ by ${\sf CHASE}_s^{{\rm lk}(\omega)}$ in ${\sf CHASE}_{s+}$ and obtain an improved algorithm for the look-ahead setting, named ${\sf CHASE}_{s+}^{{\rm lk}(\omega)}$. Fig.~\ref{fig:CHASE_ratio_with_lookahead} shows the competitive ratio of ${\sf CHASE}_{s+}^{{\rm lk}(\omega)}$ as a function of $\alpha$ and $\omega$.
\begin{figure}[htb!]

\centering
\includegraphics[width=.66\columnwidth]{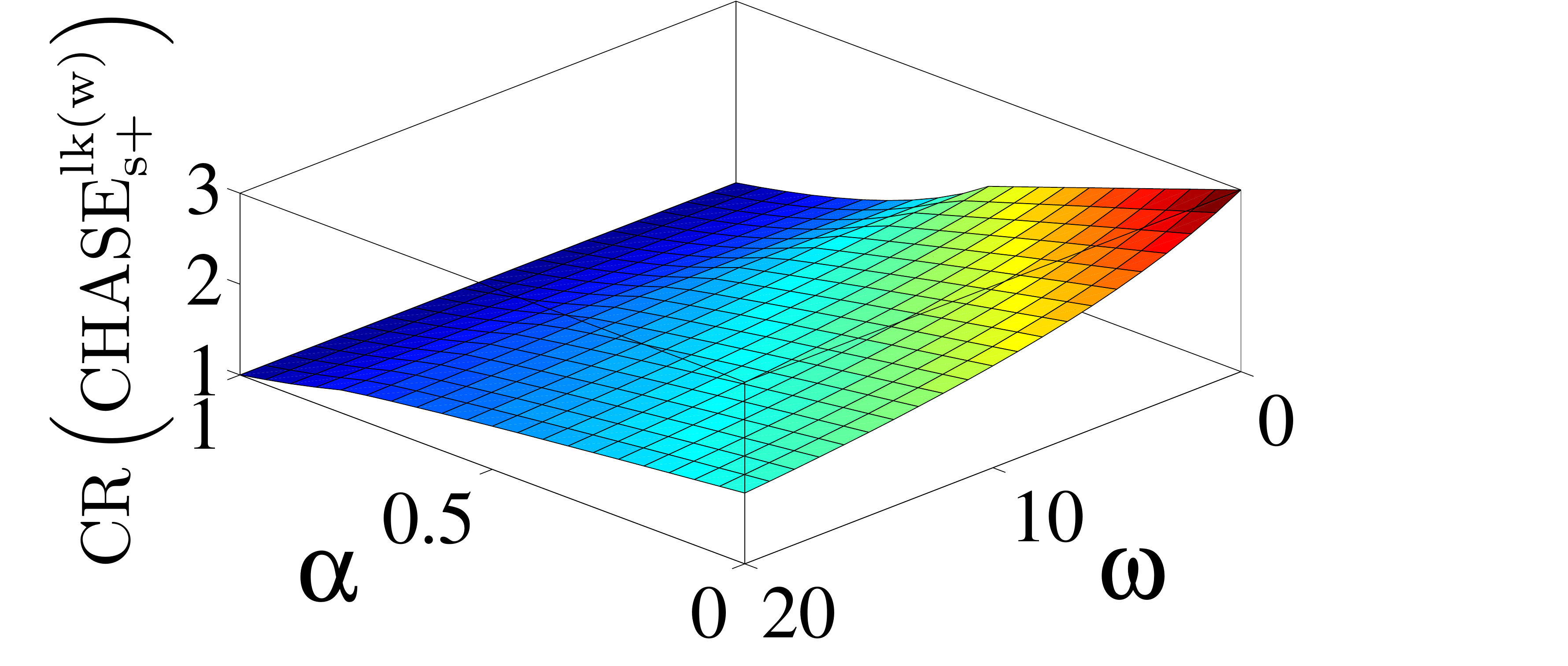}
\caption{\label{fig:CHASE_ratio_with_lookahead} The competitive ratio of ${\sf CHASE}_{s+}^{{\rm lk}(\omega)}$ as a function of $\alpha$ and $\omega$.}

\end{figure}
\subsection{Multiple Generator Case} \label{sec:ngens}
Now we consider the general case with $N$ units of homogeneous generators, each having an maximum power capacity $L$, startup cost $\beta$, sunk cost $c_{m}$ and per unit operational cost $c_{o}$.
We define a generalized version of problem:
\begin{align*}
 \mathbf{fMCMP}: \underset{{y,u,v,s}}{\min}&{\rm Cost}(y,u,v,s) \\
 \mbox{s.t.}\;& \mbox{Constraints }\eqref{C_max_output},\;\eqref{C_e-demand},\mbox{ and }\eqref{C_h-demand}\\
 \mbox{var}\;& y_n(t)\in \{0,1\},u_n(t),v(t),s(t)\in \Bbb{R}_0^{+},
\end{align*}
Next, we will construct both offline and online solutions to
$\mathbf{fMCMP}$ in a divide-and-conquer fashion. We will first partition the
demands into sub-demands for each generator, and then optimize the local
generation \emph{separately} for each sub-demand. Note that the key is
to correctly partition the demand so that the combined solution is
still optimal. Our strategy below essentially slices the demand (as a
function of $t$) into multiple layers from the bottom up (see
Fig.~\ref{fig:layers=example}).
Each layer has at most $L$ units of electricity demand and $\eta \cdot
L$ units of heat demand. The intuition here is that the layers at the
bottom exhibit the least frequent variations of demand. Hence, by assigning
each of the layers at the bottom to a dedicated generator, these generators
will incur the least amount of switching, which helps to reduce the
startup cost.

%

More specifically, given $(a(t), h(t))$, we slice them into $N+1$ layers:
\begin{subequations}
\begin{align} 
a^{{\rm ly\mbox{-}}1}(t) = & \min\{L, a(t)\}, \quad h^{{\rm ly\mbox{-}}1}(t) = \min\{\eta \cdot L, h(t)\} \label{eq:demand_slicing-begin}  \\
a^{{\rm ly\mbox{-}}n}(t) = & \min\{L, a(t) \mbox{-} \textstyle{\sum}_{r=1}^{n-1} a^{{\rm ly\mbox{-}}r}(t)\}, n \in [2,N] \\
h^{{\rm ly\mbox{-}}n}(t) = & \min\{\eta \cdot L, h(t) \mbox{-} \textstyle{\sum}_{r=1}^{n-1} h^{{\rm ly\mbox{-}}r}(t)\}, n \in [2,N] \\
a^{\rm top}(t) = & \min\{L, a(t) - \textstyle{\sum}_{r=1}^{N} a^{{\rm ly\mbox{-}}r}(t)\} \\
h^{\rm top}(t) = & \min\{\eta \cdot L, h(t) - \textstyle{\sum}_{r=1}^{N} h^{{\rm ly\mbox{-}}r}(t) \label{eq:demand_slicing-end}\}
\end{align}
\end{subequations}
\begin{figure}[t!]
\subfloat[ An example of $(a^{\mathrm{ly-n}})$.]{\includegraphics[width=0.5\columnwidth]{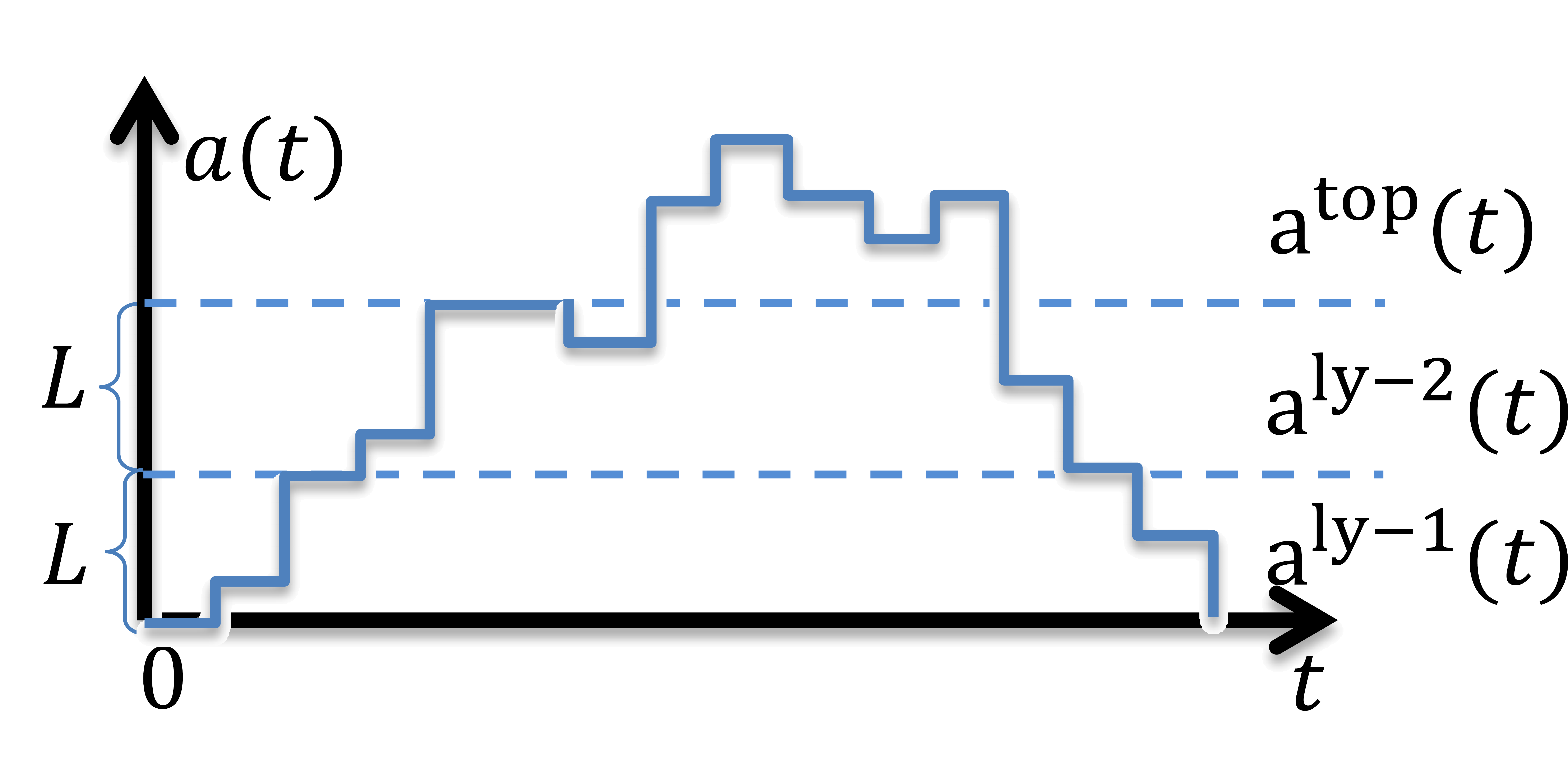}

}\subfloat[ An example of $(h^{\mathrm{ly-n}})$.]{\includegraphics[width=0.5\columnwidth]{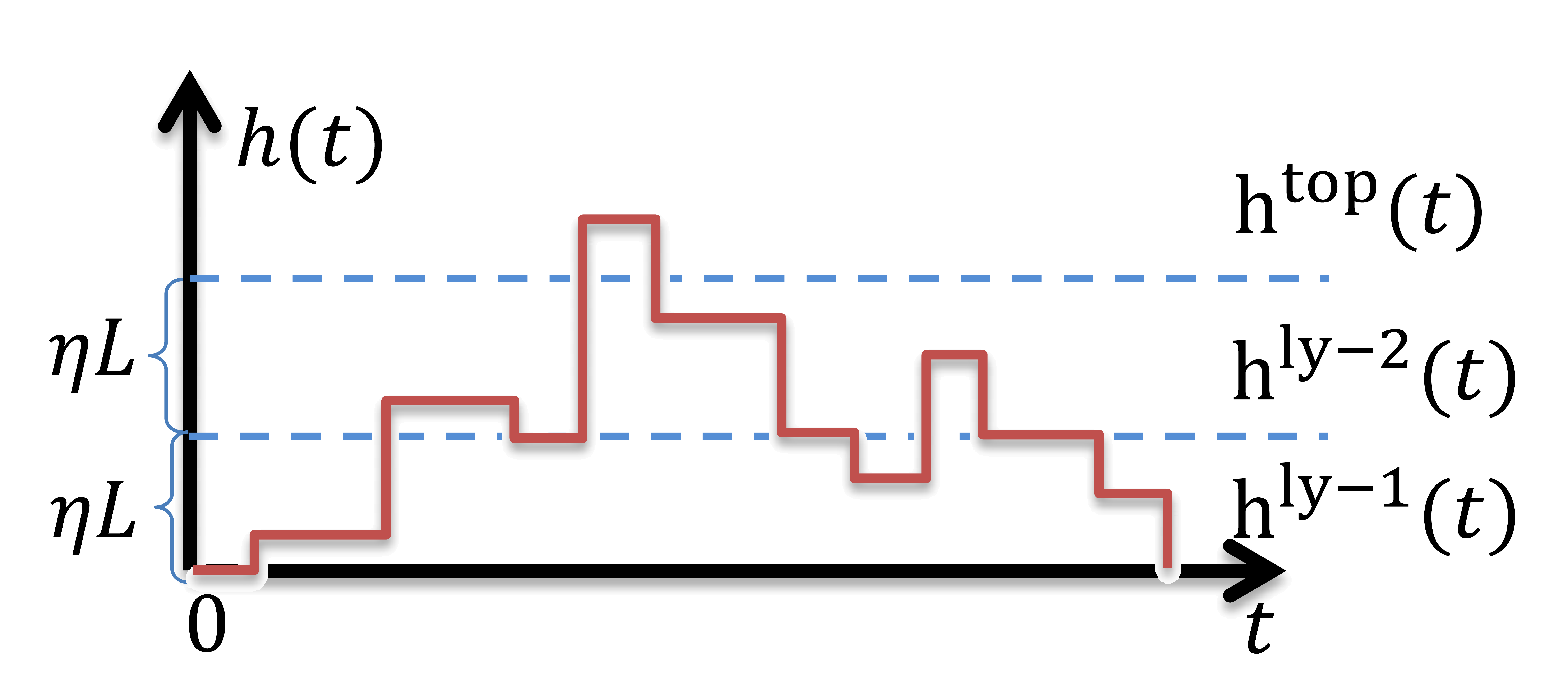}}

\caption{ An example of $(a^{\mathrm{ly-n}})$ and $(h^{\mathrm{ly-n}})$. In this example, $N=2.$ We obtain $3$ layers of
electricity and heat demands, respectively.} \label{fig:layers=example}
\end{figure}
 It is easy to see that electricity demand satisfies $a^{{\rm
ly\mbox{-}}n}(t)\leq L$ and heat demand satisfies $h^{{\rm ly\mbox{-}}n}(t)\leq
\eta \cdot L$. Thus, each layer of sub-demand can be served by a single
local generator if needed. Note that $(a^{\rm top}, h^{\rm top})$ can
only be satisfied from external supplies, because they exceed the
capacity of local generation.

Based on this decomposition of demand, we then decompose the
\textbf{fMCMP} problem into $N$ sub-problems $\textbf{fMCMP}_{\rm s}^{{\rm
ly\mbox{-}}n}$ ($1\leq n\leq N$), each of which is an
$\textbf{fMCMP}_{\rm s}$ problem with input $(a^{{\rm ly\mbox{-}}n}, h^{{\rm
ly\mbox{-}}n}, p)$. We then apply the offline and
online algorithms developed earlier to solve each sub-problem
$\textbf{fMCMP}_{\rm s}^{{\rm ly\mbox{-}}n}$ ($1\leq n\leq N$) {\em
separately}. By combining the solutions to these sub-problems, we obtain
offline and online solutions to \textbf{fMCMP}. For the offline
solution, the following theorem states that such a divide-and-conquer
approach results in no optimality loss.


\begin{thm} \label{thm:nOFA-optimal}
Suppose $(y_n, u_n, v_n, s_n)$ is an optimal offline solution for each
$\textbf{fMCMP}_{\rm s}^{{\rm ly\mbox{-}}n}$ ($1\leq n\leq N$). Then\\
$((y^\ast_{n}, u^\ast_{n})_{n=1}^{N}, v^\ast, s^\ast)$ defined as
follows is an optimal offline solution for $\textbf{fMCMP}$:
\begin{equation}
\begin{array}{@{}r@{\ }r@{\ }l@{}}
y^\ast_{n}(t) & = & y_n(t), \ \ v^\ast(t) = a^{\rm top}(t) + \sum_{n=1}^{N} v_n(t)  \\
u^\ast_{n}(t) & = & u_n(t), \ \ s^\ast(t) =  h^{\rm top}(t) + \sum_{n=1}^{N} s_n(t)
\end{array}
\end{equation}
\end{thm}
\begin{proof}
Refer to Appendix~\ref{subsec:nOFA-optimal}.
\end{proof}
For the online solution, we also apply such a divide-and-conquer approach
by using
(i) a central demand dispatching module that slices and dispatches
demands to individual generators according to
\eqref{eq:demand_slicing-begin}-\eqref{eq:demand_slicing-end}, and (ii)
an online generation scheduling module sitting on each generator $n$
($1\leq n\leq N$) \emph{independently} solving their own
$\textbf{fMCMP}_s^{{\rm ly\mbox{-}}n}$ sub-problem using the online
algorithm ${\sf CHASE}_{s+}^{\rm lk(\omega)}$.

%

The overall online algorithm, named ${\sf CHASE}^{\rm lk(\omega)}$, is simple to implement without the need to coordinate the control among multiple local generators. Since the offline (resp. online) cost of $\textbf{fMCMP}$ is the sum of the offline (resp. online) costs of $\textbf{fMCMP}_s^{{\rm ly\mbox{-}}n}$ ($1\leq n\leq N$), it is not difficult to establish the competitive ratio of ${\sf CHASE}^{\rm lk(\omega)}$ as follows.

\begin{thm} \label{thm:CHASE-NG}
The competitive ratio of ${\sf CHASE}^{\rm lk(\omega)}$ satisfies
\begin{equation}
{\sf CR}({\sf CHASE}^{\rm lk(\omega)})\le \min (3-2\cdot g(\alpha,\omega), 1/\alpha), \label{eq:CHASE-N_Ratio}
\end{equation}
where {$\alpha\in(0,1]$ is defined in \eqref{eq:alpha_def}} and $g(\alpha,\omega)\in [\alpha, 1]$ is defined in \eqref{eq:g_alpha_omega}.
\end{thm}
\begin{proof}
Refer to Appendix~\ref{subsec:MultipleRatio}.
\end{proof}

\section{Slow-responding Generator Case} \label{sec:slowrep}
We next consider the slow-responding generator case, with the generators
having non-negligible constraints on the minimum on/off
periods and the ramp-up/down speeds.
For this slow-responding version of {\bf MCMP}, its offline optimal solution is harder to characterize than {\bf fMCMP} due to the additional challenges introduced by the cross-slot constraints \eqref{C_ramp_up}-\eqref{C_min_off}.

In the slow-responding setting, local generators cannot be turned on and off
immediately when demand changes. Rather, if a generator is turned on (resp., off)
at time $t$, it must remain on for at least ${\sf T}_{{\rm on}}$ (resp., ${\sf T}_{{\rm off}}$) time
.
Further, the changes of $u_n(t)-u_n(t-1)$ must be bounded by ${\sf R}_{{\rm
up}}$ and $-{\sf R}_{{\rm down}}$.



\rev{A simple heuristic is to first compute solutions based on ${\sf CHASE}^{\rm lk(\omega)}$, and
then modify the solutions to respect the above constraints. We name this
heuristic ${\sf CHASE}^{\rm lk(\omega)}_{\rm gen}$ and present it in Algorithm ~\ref{alg:CHASE-general}. For simplicity, Algorithm ~\ref{alg:CHASE-general} is a single-generator version, which can be easily extended to the multiple-generator scenario by following the divide-and-conquer approach elaborated in Sec.~\ref{sec:ngens}.
\begin{algorithm}[htp!]
\caption{ ${\sf CHASE}^{\rm lk(\omega)}_{\rm gen} [t, (\sigma(\tau))_{\tau = 1}^{t+\omega}, y(t-1)]$}
\label{alg:CHASE-general}
{
\begin{algorithmic}[1]
\STATE \mbox{$\big(y_{s}(t),u_{s}(t),v_{s}(t),s_{s}(t)\big)\leftarrow{\sf CHASE}_{s}^{{\rm lk}(\omega)}\big[t,\big(\sigma\big(\tau\big)\big)_{\tau=1}^{t+w},y\big(t\mbox{-}1\big)\big]$}

\IF {$y(\tau_{1})\leq1-\mathbf{1}_{\{y_{s}(t)>y(t-1)\}},\ \forall\tau_{1}\in[\max(1,t-{\sf T}_{{\rm off}}),t-1]$
and $y(\tau_{2})\geq\mathbf{1}_{\{y_{s}(t)<y(t-1)\}},\ \forall\tau_{2}\in[\max(1,t-{\sf T}_{{\rm on}}),t-1]$ }
\STATE $y(t)\leftarrow y_{s}(t)$
\ELSE
\STATE $y(t)\leftarrow y(t-1)$
\ENDIF

\IF {$u_{s}(t)>u(t-1)$}
\STATE $u(t)\leftarrow u(t-1)+\min\big({\sf R}_{{\rm up}},u_{s}(t)-u(t-1)\big)$
\ELSE
\STATE $u(t)\leftarrow u(t-1)-\min\big({\sf R}_{{\rm dw}},u(t-1)-u_{s}(t)\big)$
\ENDIF
\STATE  ${v(t)\leftarrow\big[a(t)-u(t)\big]^{+}}$
\STATE  ${s(t)\leftarrow\big[h(t)-\eta\cdot u(t)\big]^{+}}$
\STATE return $\big(y(t),u(t),v(t),s(t)\big)$
\end{algorithmic}
}
\end{algorithm}

We now explain Algorithm~\ref{alg:CHASE-general} and its competitive ratio.
At each time slot $t$, we obtain the solution of ${\sf CHASE}_{s}^{{\rm lk}(\omega)}$, including $y_{s}(t),u_{s}(t),v_{s}(t),s_{s}(t)$, as a reference solution (Line 1). Then in Line 2-6, we modify the reference solution's $y_{s}(t)$ to our actual solution $y(t)$, to respect the constraints of minimum on/off periods. More specifically, we follow the reference solution's $y_{s}(t)$ ({\em i.e.,} $y(t)=y_{s}(t)$) \emph{if and only if} it respects the minimum on/off periods constraints (Line 2-3). Otherwise, we let our actual solution's $y(t)$ equal our previous slot's solution ($y(t)=y(t-1)$) (Line 4-5). Similarly, we modify the reference solution's $u_{s}(t)$ to our actual solution's $u(t)$, to respect the constraints on ramp-up/down speeds (Line 7-11). At last, in our actual solution, we use $(v(t),s(t))$ to compensate the supply and satisfy the demands (Line 12-13). In summary, our actual solution is designed to be aligned with the reference solution as much as possible.
}
We derive an upper bound on the competitive ratio of ${\sf CHASE}^{\rm lk(\omega)}_{\rm gen}$ as follows.
{\begin{thm}  \label{thm:slowratio}
The competitive ratio of $\mathrm{{\sf CHASE}_{\rm gen}^{{\rm lk}(\omega)}}$ is upper bounded
by $(3-2g(\alpha, \omega))\cdot\max\big(r_{1},r_{2}\big)$, where $g(\alpha, \omega)$ is defined in \eqref{eq:g_alpha_omega} and
\begin{align*}
r_{1}  = & 1+\max\left\{ \frac{\big(P_{\max}+c_{g}\cdot\eta-c_{0}\big)}{Lc_{0}+c_{m}}\max\left\{ 0,\big(L-{\sf R}_{{\rm up}}\big)\right\} \right.\\
  & \left.\frac{c_{o}}{c_{m}}\max\left\{ 0,\big(L-{\sf R}_{{\rm dw}}\big)\right\} \right\}, \\
\mbox{and }r_{2}  = & \frac{\beta+c_m \cdot {\sf T}_{{\rm on}}}{\beta} + \frac{L\big(P_{\max}+c_{g}\cdot\eta\big)}{\beta}\left({\sf T}_{{\rm on}} +{\sf T}_{{\rm off}}\right).
\end{align*}
\end{thm}
\begin{proof}
Refer to Appendix~\ref{subsec:slowratio}.
\end{proof}
We note that when ${\sf T}_{{\rm on}}={\sf T}_{{\rm off}}=0$, ${\sf R}_{{\rm up}}={\sf R}_{{\rm dw}}=\infty$, the above upper bound matches that of ${\sf CHASE}^{\rm lk(\omega)}$ in Theorem~\ref{thm:CHASE-NG} (specifically the first term inside the min function).}

\section{Empirical Evaluations} \label{sec:empirical}

We evaluate the performance of our algorithms based on evaluations using real-world traces. Our objectives are three-fold: (i) evaluating the potential benefits of CHP and the ability of our algorithms to unleash such potential, (ii) corroborating the empirical performance of our online algorithms under various realistic settings, and (iii) understanding how much local generation to invest to achieve substantial economic benefit.

\subsection{Parameters and Settings} \label{sec:setting}

{\bf Demand Trace}: We obtain the demand traces
from California Commercial End-Use Survey (CEUS) \cite{CEUS}. We focus
on a college in San Francisco, which consumes about 154 GWh
electricity and $5.1\times10^{6}$ therms gas per year. The traces contain hourly electricity and \rev{heat} demands
of the college for year 2002. The heat demands for a typical week in summer and spring are shown in Fig. \ref{fig:trace}.
They display regular daily patterns in peak and off-peak hours, and
typical weekday and weekend variations.

{\bf Wind Power Trace}: We obtain the wind power traces from \cite{NREL}.
{We employ power output data for the typical weeks in summer and spring with a resolution of 1 hour of an offshore
wind farm right outside San Francisco with an installed capacity of 12MW.} The net electricity demand, which is computed by subtracting the wind generation from electricity demand is shown in Fig. \ref{fig:trace}. The highly fluctuating and unpredictable
nature of wind generation makes it difficult for the conventional prediction-based energy generation scheduling solutions to work effectively.

{\bf Electricity and Natural Gas Prices}: The electricity and
natural gas price data are from PG\&E \cite{PG_E} and are shown in Table
\ref{tab:PG&E-tariffs}. Besides, the grid electricity prices for a typical week in summer and winter are shown in Fig. \ref{fig:trace}.
Both the electricity demand and the price
show strong diurnal properties: in the daytime, the demand and price
are relatively high; at nights, both are low. This suggests the
feasibility of reducing the microgrid operating cost by generating
cheaper energy locally to serve the demand during
the daytime when both the demand and electricity price are high.

{\bf Generator Model}: We adopt generators
with specifications the same as the one in \cite{tecogen}.
The full output of a single generator is $L=3MW$. The incremental
cost per unit time to generate an additional unit of energy $c_{o}$
is set to be $0.051\ensuremath{/KWh}$, which is calculated according
to the natural gas price and the generator efficiency. We set the heat recovery efficiency of co-generation $\eta$ to be $1.8$ according to \cite{tecogen}. We
also set the unit-time generator running cost to be
 $c_{m}=110\$/h$, which includes the amortized capital cost and maintenance cost according to
a similar setting from \cite{stadlerdistributed}. We set the
startup cost $\beta$ equivalent to running the generator at its
full capacity for about 5 hrs at its own operating cost
which gives
$\beta=1400\$$.  In addition, we assume for each generator ${\sf
T}_{{\rm on}}={\sf T}_{{\rm off}}=3h$ and ${\sf R}_{{\rm up}}={\sf
R}_{{\rm dw}}=1MW/h$, unless mentioned otherwise. For electricity demand trace we use, the peak demand is
30MW. Thus, we assume there are 10 such CHP generators so as to fully satisfy the demand.


{\bf Local Heating System}: We assume an on-demand heating system with capacity sufficiently
large to satisfy all the heat demand by itself and without
on-off cost or ramp limit. The efficiency of a heating system is set to
$0.8$ according to \cite{greenenergy}, and consequently we can compute
the unit heat generation cost to be  $c_{g}=0.0179\$/KWh$.

{\bf Cost Benchmark}: We use the cost incurred by using only
  external electricity, heating and wind energy (without
CHP generators) as a benchmark. We
evaluate the cost reduction due to our algorithms.

{\bf Comparisons of Algorithms}: We compare three algorithms in our simulations.
(1) our online algorithm {\sf CHASE}; (2) the Receding
Horizon Control ({\sf RHC}) algorithm; and (3) {the {\sf OFFLINE} optimal
algorithm we introduce in Sec.~\ref{sec:slowrep}}. {\sf RHC} is a heuristic algorithm
commonly used in the control literature \cite{RHC3}. In {\sf RHC}, an estimate
of the near future ({\em e.g.,} in a window of length $w$) is used to compute
a tentative control trajectory that minimizes the cost over this
time-window.  However, only the first step of this trajectory is
implemented. In the next time slot, the window of future estimates
shifts forward by $1$ slot.  Then, another
control trajectory is computed based on the new future information, and
again only the first step is implemented. This process then continues.
We note that because at each step {\sf RHC} does not consider any
adversarial future dynamics beyond the time-window $w$, there is no
guarantee that {\sf RHC} is competitive.
For the {\sf OFFLINE} algorithm, the inputs are system parameters (such as
$\beta$, $c_{m}$ and ${\sf T}_{{\rm on}}$), electricity demand, heat
demand, wind power output, gas price, and grid electricity price.
For online algorithms {\sf CHASE} and {\sf RHC}, the input is the same as the
{\sf OFFLINE} except that at time $t$, only the demands, wind power
output, and prices in the past and the look-ahead window ({\em i.e.,} $[1,
t+w]$) are available.  The output for all three algorithms is the total
cost incurred during the time horizon $[1,T]$.
%
%
\begin{figure}[htb!]
\subfloat[ \label{fig:dtrace_summer} Summer]{\includegraphics[width={0.53\columnwidth}]{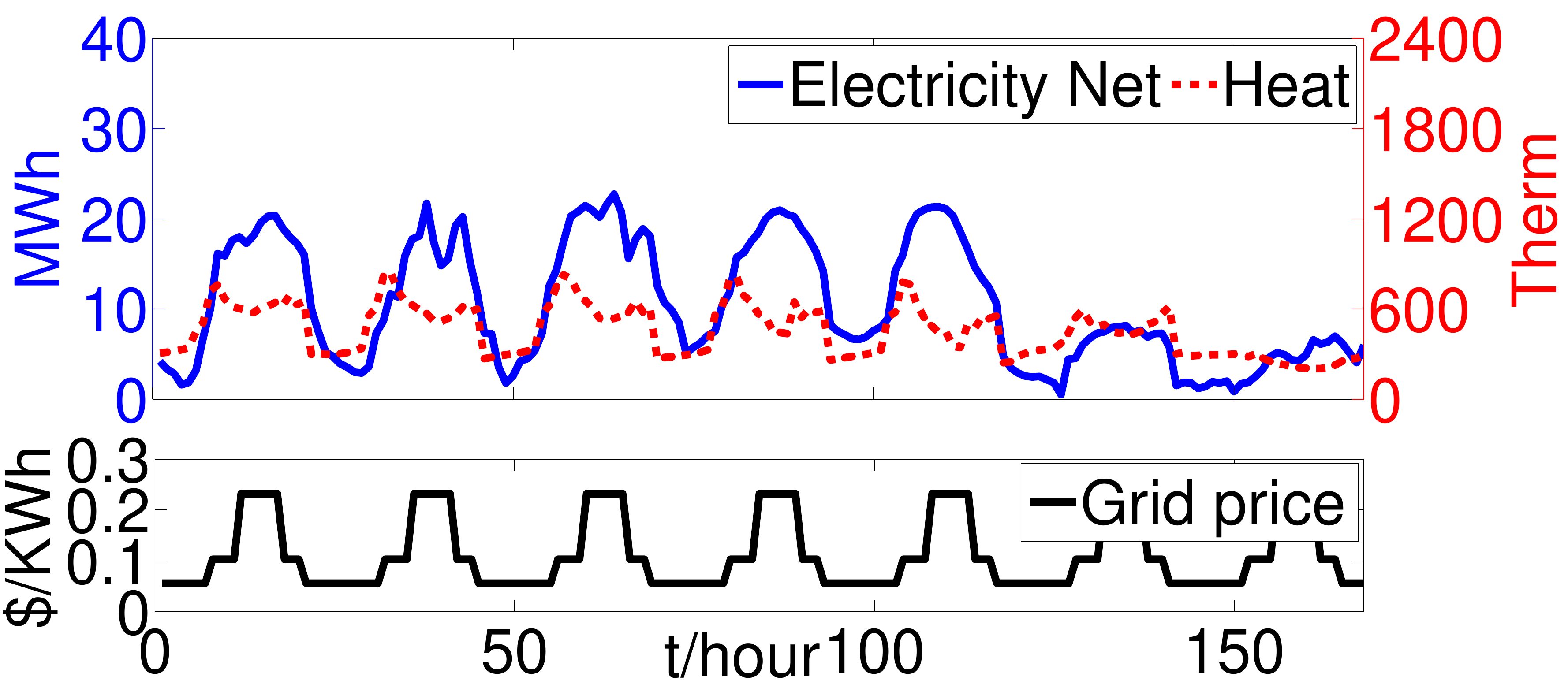}}
\subfloat[ \label{fig:dtrace_winter} Winter]{\includegraphics[width={0.53\columnwidth}]{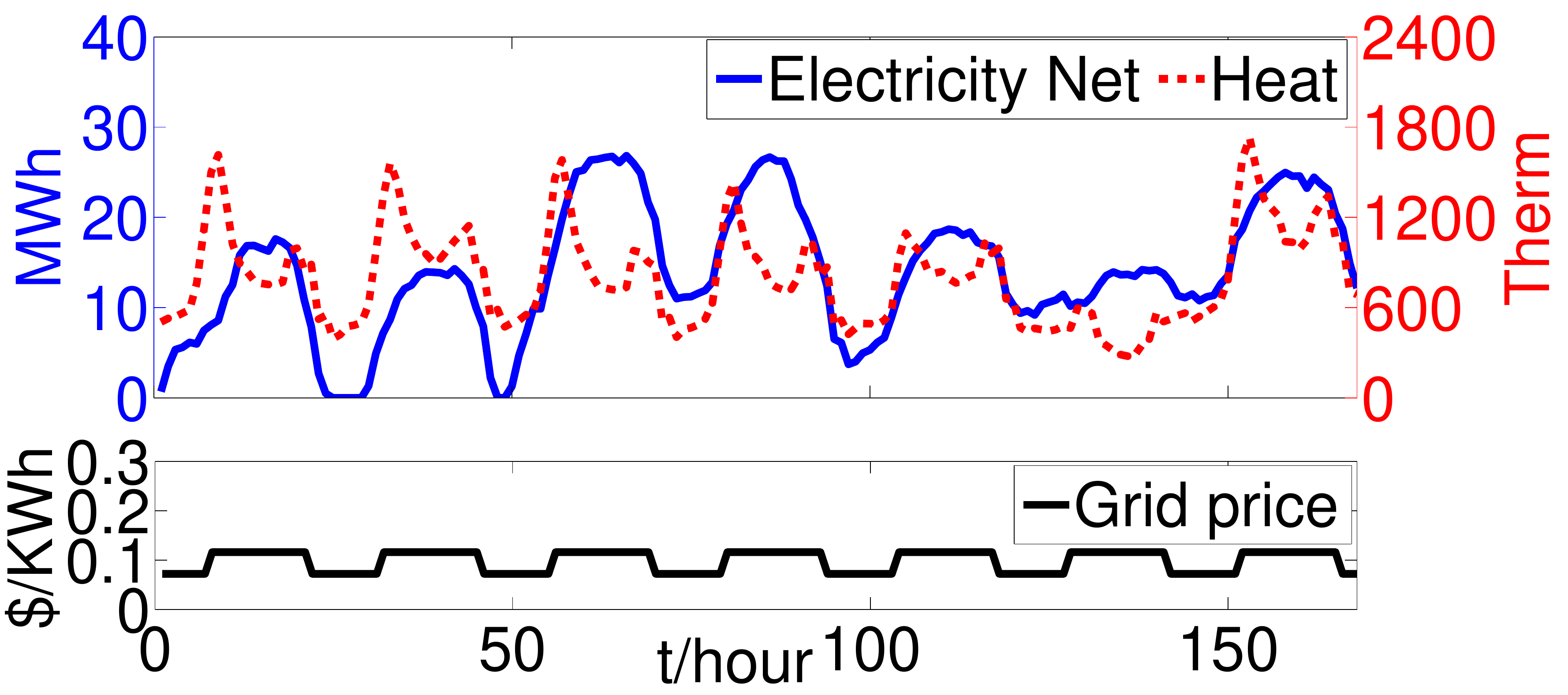}}
\caption{\label{fig:trace} Electricity net demand and heat demand for a typical week in summer and winter. The net demand is computed by subtracting the wind generation from the electricity demand. The net electricity demand and the heat demand need to be satisfied by using the local CHP generators, the electricity grid, and the heating system.}
\end{figure}

\begin{table}[htb!]
\centering
{
\begin{tabular}{c|c|c}
\hline \hline
Electricity & Summer (May-Oct.)  & Winter (Nov.-Apr.)\tabularnewline
 & \$/kWh & \$/kWh\tabularnewline
\hline
On-peak & 0.232 & N/A\tabularnewline
\hline
Mid-peak & 0.103 & 0.116\tabularnewline
\hline
Off-peak & 0.056 & 0.072\tabularnewline
\hline
\hline
Natural Gas & 0.419\$/therm & 0.486\$/therm\tabularnewline
\hline
\end{tabular}\caption{ \label{tab:PG&E-tariffs} PG\&E commercial tariffs and natural gas
tariffs. In the table, summer on-peak, mid-peak, and off-peak hours are weekday 12-18, weekday 8-12, and the remaining hours, respectively. Winter mid-peak and off-peak hours are weekday 8-22 and the remaining hours, respectively. The gas price is an average; monthly prices vary slightly according to PG\&E.}
}
\end{table}

\subsection{Potential Benefits of CHP}
{\bf Purpose}:
The experiments in this subsection aim to answer two
questions. First, what is the potential savings with microgrids? Note
that electricity, heat demand, wind station output as well as energy
price all exhibit seasonal patterns.
As we can see from Figs. \ref{fig:dtrace_summer} and
\ref{fig:dtrace_winter}, during summer (similarly autumn) the electricity price
is high,  while during winter (similarly spring) the heat demand is high. It is
then interesting to evaluate under what settings and inputs the savings
will be higher.
Second, what is the difference in cost-savings with and without
the co-generation capability?
In particular, we conduct two sets of experiments to evaluate the cost
reductions of various algorithms. Both experiments have the same default
settings, except that the first set of experiments (referred to as
CHP)
assumes the CHP technology in the generators is enabled, and the second
set of experiments (referred to as NOCHP) assumes the CHP technology is
not available, in which case the heat demand must be satisfied solely by
the heating system. In all experiments, the look-ahead window size is
set to be $w=3$ hours according to
{power system operation and wind generation forecast practice} \cite{met}.  The cost reductions of different algorithms are shown in
Fig.  \ref{fig:with-chp} and \ref{fig:without-chp}. The vertical
axis is the cost reduction as compared to the cost benchmark presented in
Sec. \ref{sec:setting}. 

{\bf Observations}:
First, the whole-year cost reductions obtained by {\sf OFFLINE} are 21.8\% and
11.3\% for CHP and NOCHP scenarios, respectively. This justifies the
economic potential of using local generation, especially when CHP
technology is enabled. Then, looking at the seasonal performance of
{\sf OFFLINE}, we observe that {\sf OFFLINE} achieves much more cost savings
during summer and autumn than during spring and winter. This is because
the electricity price during summer and autumn is very high, thus we can
benefit much more from using the relatively-cheaper local generation as
compared to using grid energy only. Moreover, {\sf OFFLINE} achieves much more
cost savings when CHP is enabled than when it is not during spring and
winter. This is because, during spring and winter, the electricity
price is relatively low and the heat demand is high. Hence, just
using local generation to supply electricity is not economical. Rather, local
generation becomes more economical only if it can be used to supply both
electricity and heat together ({\em i.e.,} with CHP technology).

Second, {\sf CHASE} performs consistently close to {\sf OFFLINE}
across inputs from different seasons, even though the different settings
have very different characteristics of demand and supply. In contrast,
the performance of {\sf RHC} depends heavily on the input characteristics. For
example, {\sf RHC} achieves some cost reduction during summer and autumn when
CHP is enabled, but achieves 0 cost reduction in all the other cases.

{\bf Ramifications}:
In summary, our experiments suggest that exploiting local generation
can save more cost when the electricity price is high, and CHP
technology is more critical for cost reduction when heat demand is high.
Regardless of the problem setting, it is important to adopt an
intelligent online algorithm (like {\sf CHASE}) to schedule energy generation,
in order to realize the full benefit of microgrids.

\begin{figure}
\subfloat[ \label{fig:with-chp}{ Local generators with CHP}]{\includegraphics[width={0.5\columnwidth}]{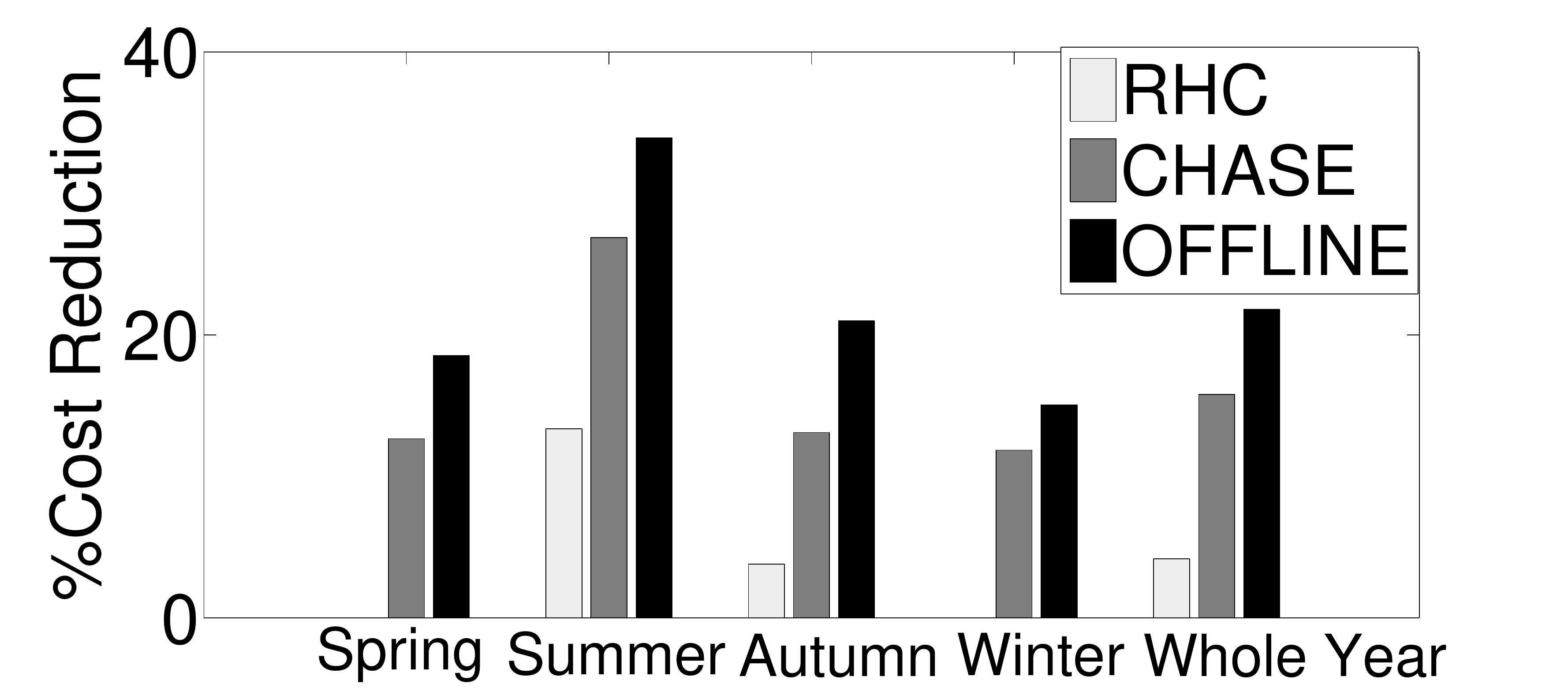}}
\subfloat[ \label{fig:without-chp}{ Local generators without CHP}]{\includegraphics[width={0.5\columnwidth}]{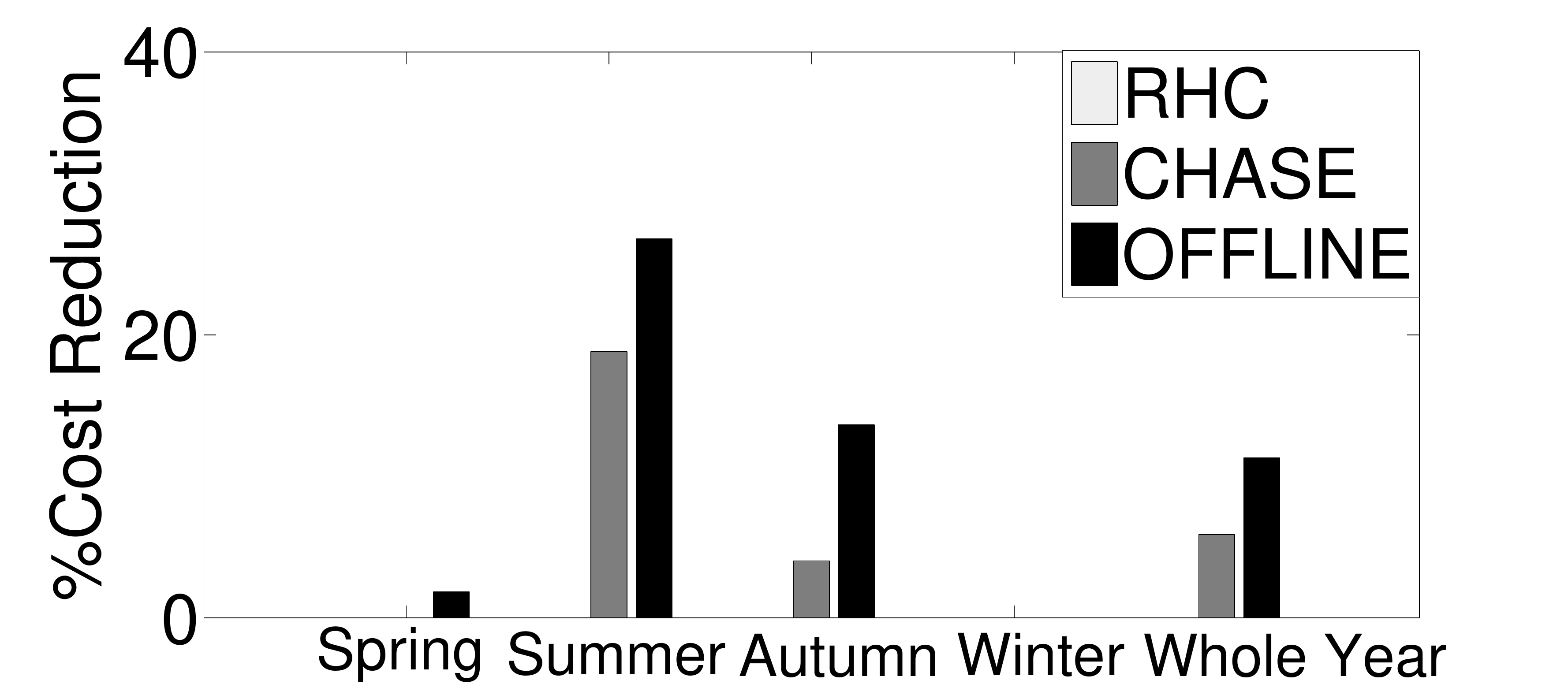}}
\caption{ Cost reductions for different seasons and the whole year.}
\end{figure}

\subsection{Benefits of Looking-Ahead}

{\bf Purpose}:
We compare the performances of {\sf CHASE} to {\sf RHC} and
{\sf OFFLINE} for different sizes of the look-ahead
window and show the results in Fig. \ref{fig:cr_w}. The vertical
axis is the cost reduction as compared to the cost benchmark in Sec. \ref{sec:setting}
and the horizontal axis is the size of lookahead window, which varies
from 0 to 20 hours.

{\bf Observations}:
We observe that the performance of our online algorithm {\sf CHASE} is already
close to {\sf OFFLINE} even when no or little look-ahead information
is available ({\em e.g.,} $w= 0$, $1$, and $2$). In contrast, {\sf RHC} performs
poorly when the look-ahead window is small. When $w$ is large, both {\sf
CHASE} and {\sf RHC} perform very well and their performance are close to {\sf OFFLINE}
 when the look-ahead window $w$ is larger than 15 hours.

An interesting observation is that it is
more important to perform intelligent energy generation scheduling when
there is no or little look-ahead information available. When there are abundant
look-ahead information available, both {\sf CHASE} and {\sf RHC} achieve good performance
and it is less critical to carry out sophisticated algorithm design.

In Fig. \ref{fig:ratio1} and \ref{fig:ratio2},
we separately evaluate the benefit of looking-ahead under the
fast-responding and slow-responding scenarios. We
evaluate the empirical competitive ratio between the cost of {\sf CHASE} and
{\sf OFFLINE}, and compare it with the theoretical competitive
ratio according to our analytical results.
In the fast-responding scenario (Fig.~\ref{fig:ratio1}), for each generator
there are no minimum on/off period and ramping-up/down constraints. Namely,
${\sf T}_{{\rm on}}=0$, ${\sf T}_{{\rm off}}=0$, ${\sf R}_{{\rm up}}=\infty$,
${\sf R}_{{\rm dw}}=\infty$. In the slow-responding scenario
(Fig.~\ref{fig:ratio2}), we set
${\sf T}_{{\rm on}}={\sf T}_{{\rm off}}=3h$ and ${\sf R}_{{\rm up}}={\sf
R}_{{\rm dw}}=1MW/h$. In both experiments, we observe that
%
the theoretical ratio decreases rapidly as look-ahead window size increases.
Further, the empirical ratio is already close to one even when there is no look-ahead information.

\begin{figure}[t!]
\begin{minipage}[t]{0.48\linewidth}
\centering
\includegraphics[width=\columnwidth]{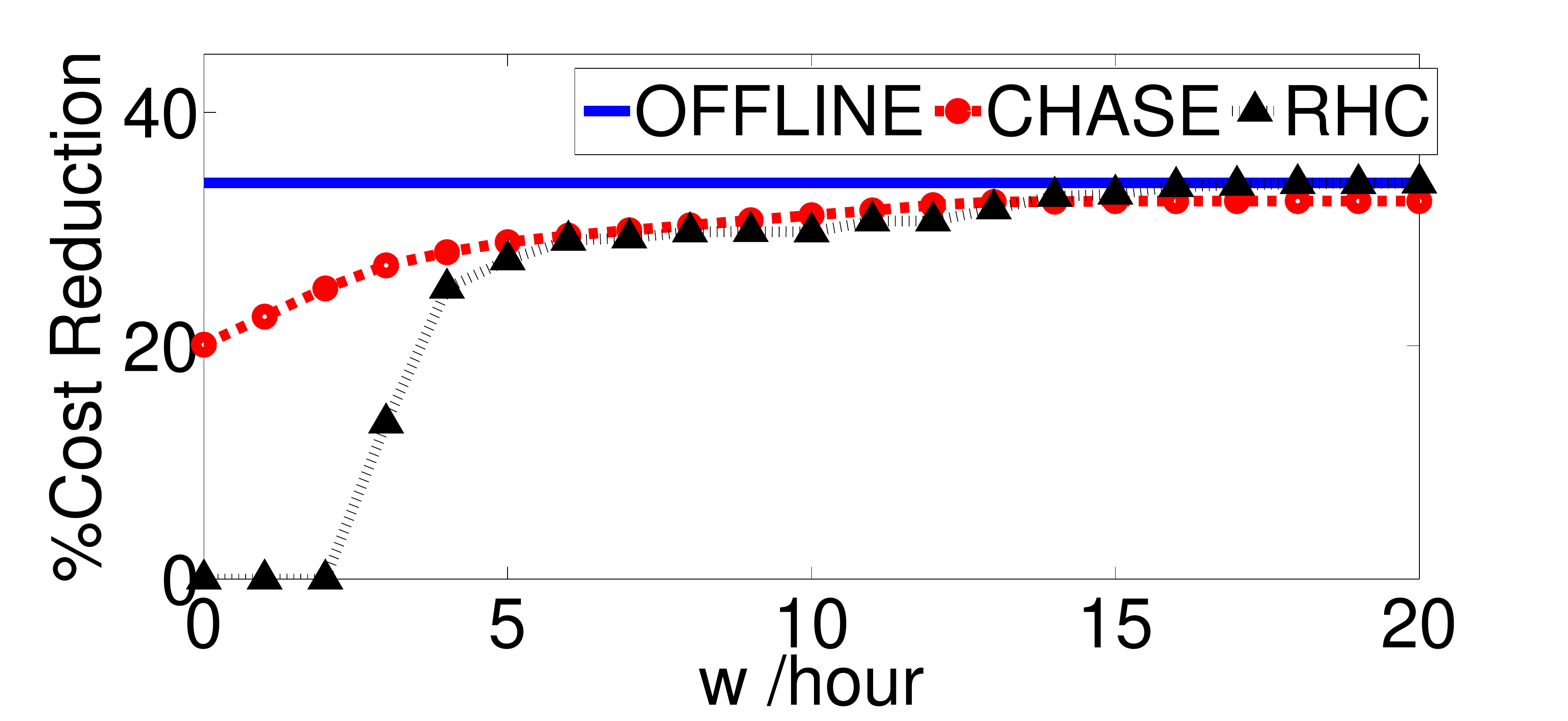}
\caption{\label{fig:cr_w} Cost reduction as a function of look ahead window size $\omega$.}
\end{minipage}
\hfill
\begin{minipage}[t]{0.48\linewidth}
\centering
\includegraphics[width=\columnwidth]{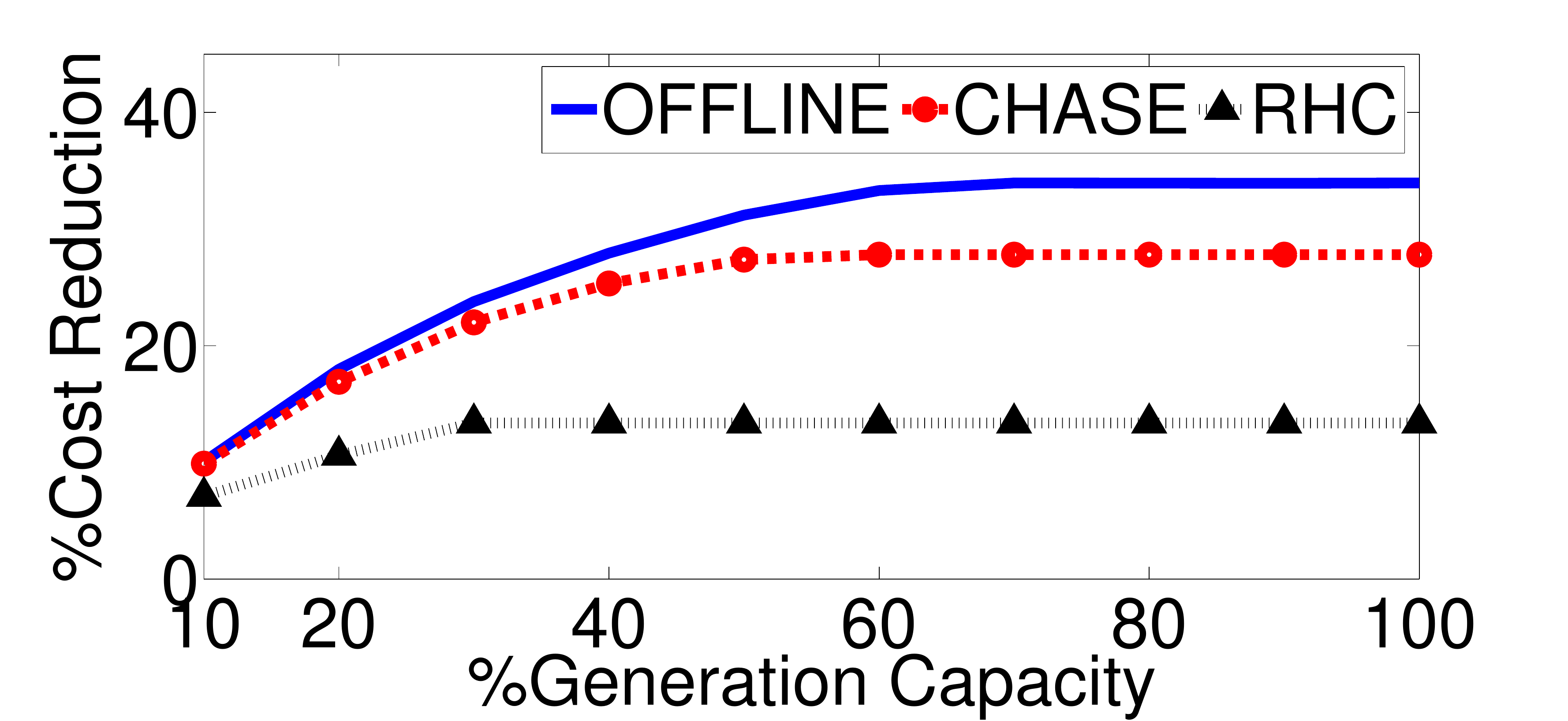}
\caption{\label{fig:cr_yub} Cost reduction as a function of local generation capacity.}
\end{minipage}
\end{figure}

%
%

\begin{figure}[t!]

\subfloat[ \label{fig:ratio1} Fast-responding scenario]{\includegraphics[width={0.5\columnwidth}]{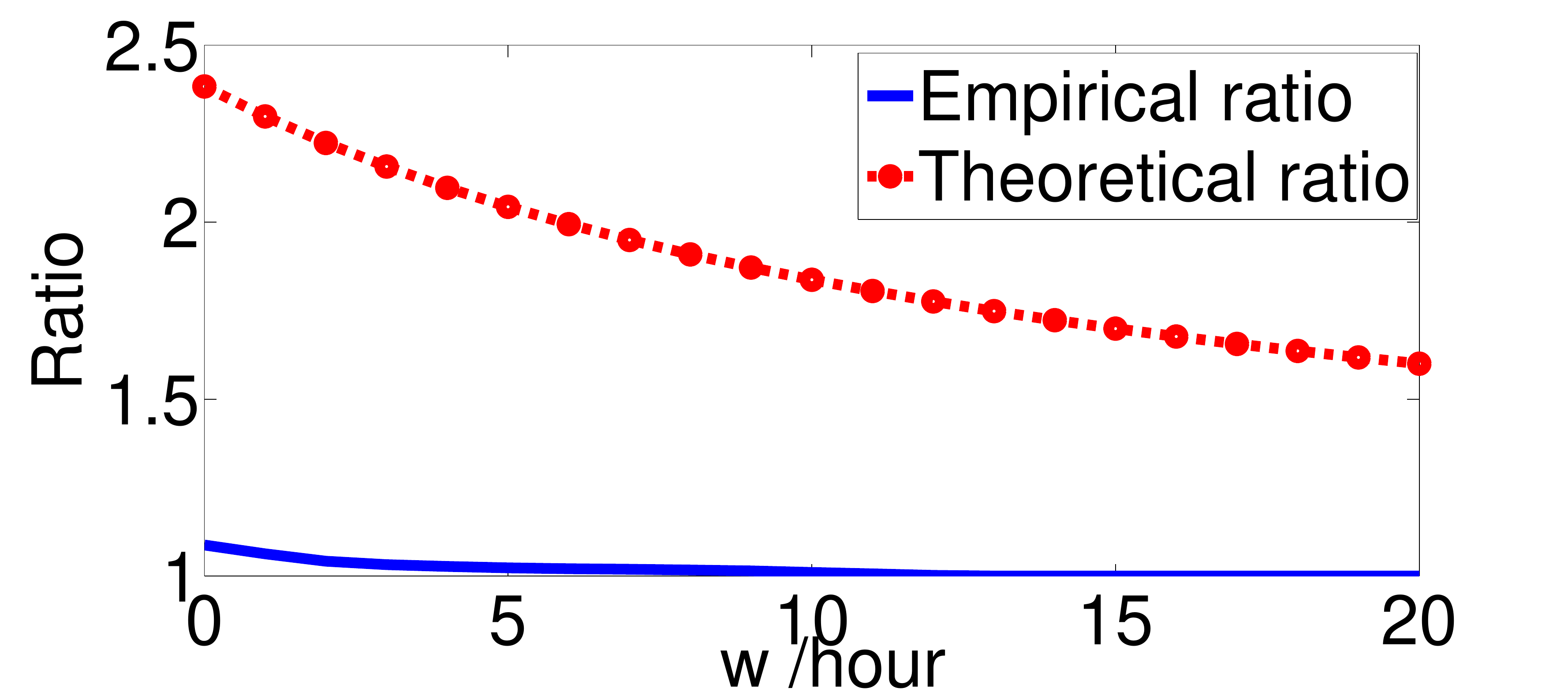}}\subfloat[ \label{fig:ratio2} Slow-responding scenario]{\includegraphics[width={0.5\columnwidth}]{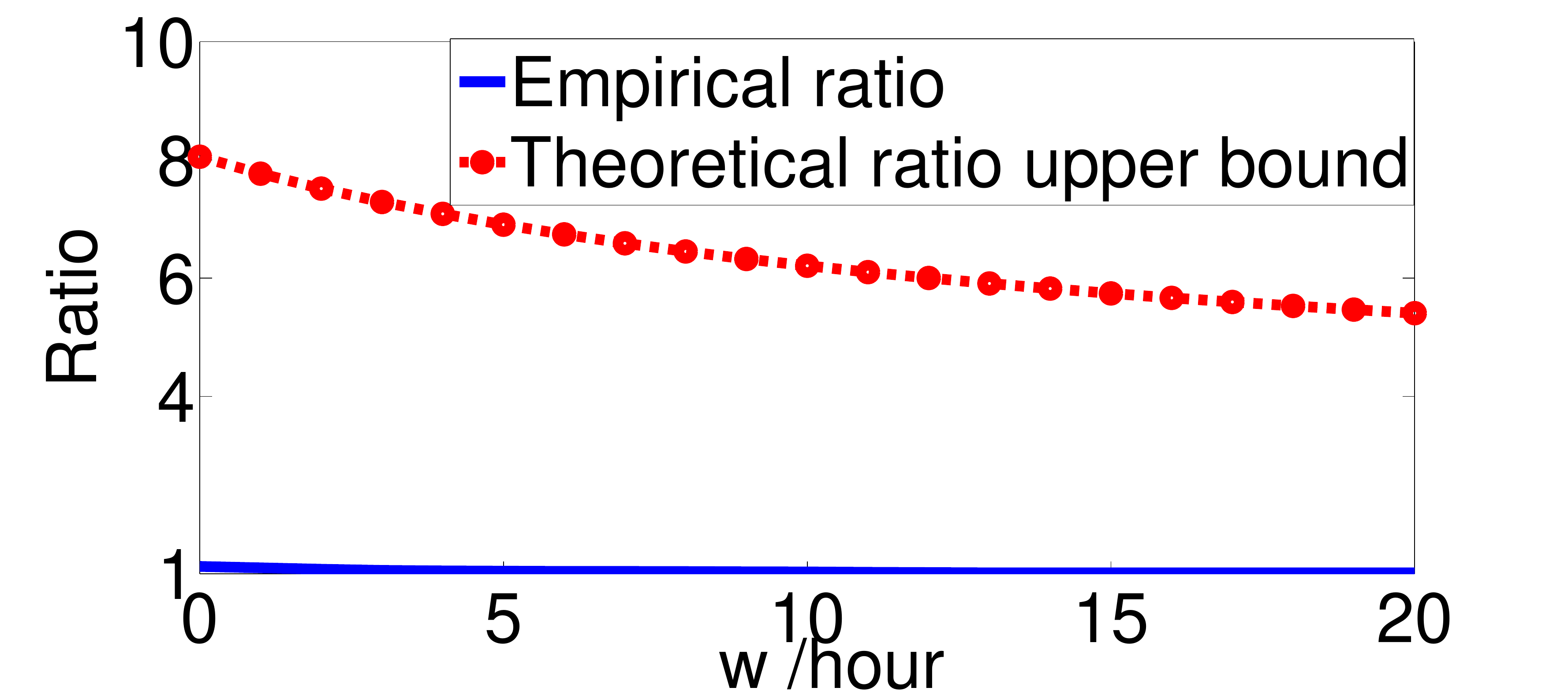}}

\caption{ Theoretical and empirical ratios
for {\sf CHASE}, as functions of look-ahead window size $\omega$. Note that the
theoretical competitive ratios (or their bounds) measure the
worst-case performance and are often much larger than the empirical
ratios observed in practice.}
\end{figure}

\subsection{Impacts of Look-ahead Error}

{\bf Purpose}:
Previous experiments show that our algorithms have better performance
if a larger time-window of accurate look-ahead input information is
available. The input information in the look-ahead window includes the
wind station power output, the electricity
and heat demand, and the central grid electricity price. In practice, these look-ahead
information can be obtained by applying sophisticated prediction techniques based on the historical data.
However, there are always prediction errors. For example, while the day-ahead electricity
demand can be predicted within 2-3\% range, the wind power prediction in the next hours  usually comes
with an error range of 20-50\% \cite{windpredictionerror}.  Therefore, it is important to evaluate the performance of the
algorithms in the presence of prediction error.

{\bf Observations}:
To achieve this goal, we evaluate {\sf CHASE} with look-ahead window
size of 1 and 3 hours. According to \cite{windpredictionerror}, the hour-level
wind-power prediction-error in terms of the percentage of the total
installed capacity usually follows Gaussian distribution. Thus, in
the look-ahead window, a zero-mean Gaussian prediction error is added to
the amount of wind power in each time-slot. We vary the standard
deviation of the Gaussian prediction error from 0 to 120\% of the total
installed capacity. Similarly, a zero-mean Gaussian prediction error is
added to the heat demand, and its standard deviation also varies from 0 to 120\%
of the peak demand.
We note that in practice, prediction errors are often in the range of 20-50\%
for 3-hour prediction \cite{windpredictionerror}. Thus, by using a
standard deviation up to 120\%, we are essentially stress-testing our
proposed algorithms.
We average 20 runs for each algorithm and show the results
in Figs. \ref{fig:cr_aerror} and \ref{fig:cr_berror}.
As we can see, both {\sf CHASE} and {\sf RHC} are fairly robust to the prediction error and both are more
sensitive to the wind-power prediction error than to the heat-demand prediction
error. Besides, the impact of the prediction error is relatively small when
the look-ahead window size is small, which matches with our intuition.

\begin{figure}
\subfloat[ \label{fig:cr_aerror} Wind power forecast error]{\includegraphics[width={0.5\columnwidth}]{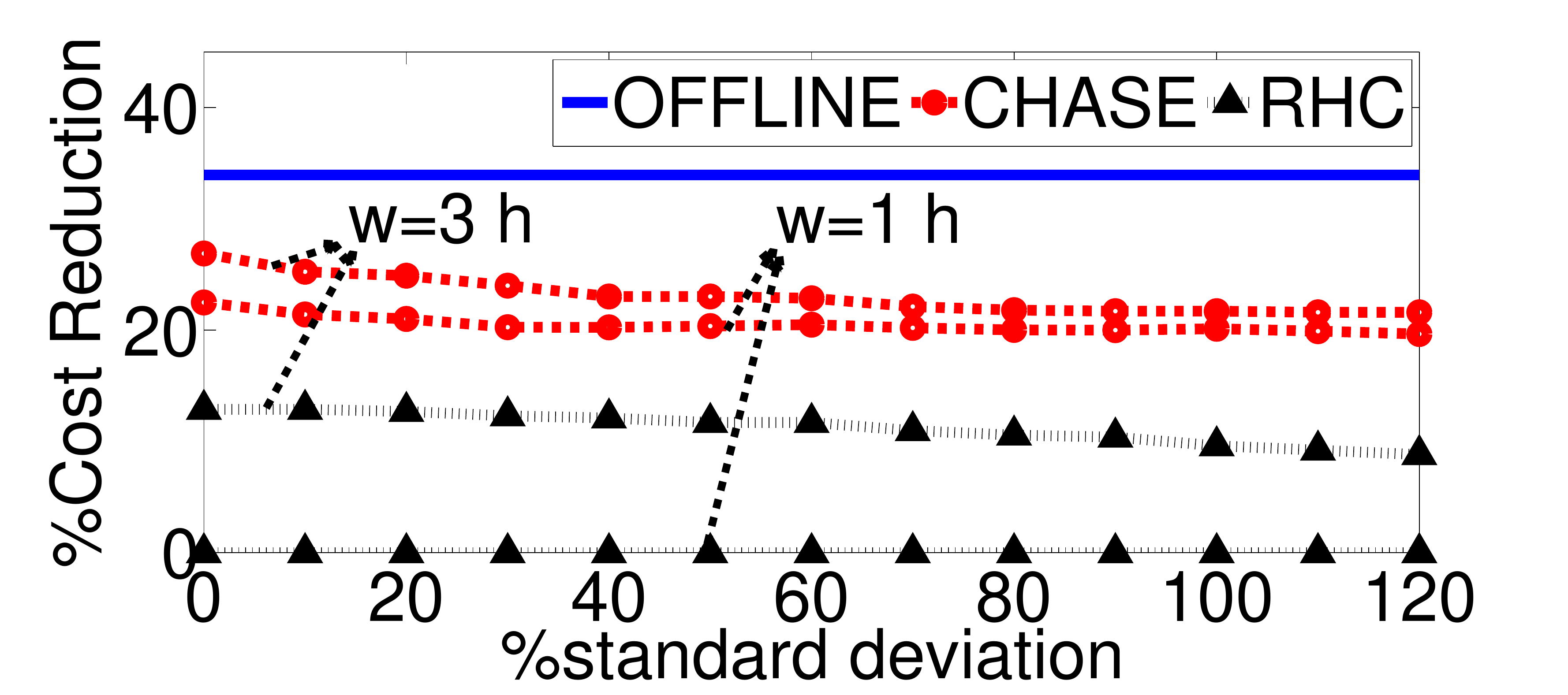}
}
\subfloat[ \label{fig:cr_berror} Heat demand forecast error]{\includegraphics[width={0.5\columnwidth}]{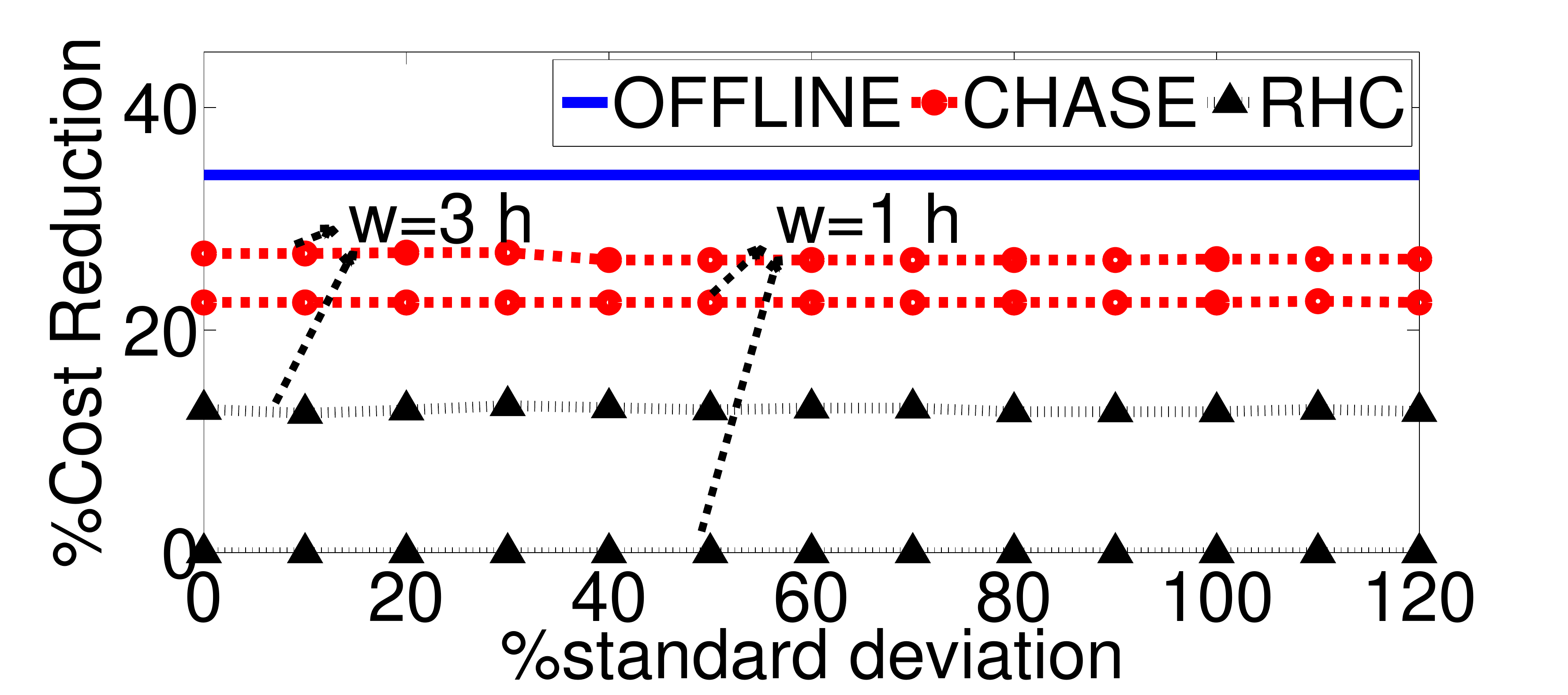}

}

\caption{ Cost reduction as a function of the size prediction
error (measured by the standard deviation of the prediction error as a
percentage of (a)installed capacity and (b)peak heat demand).}
\end{figure}

\subsection{Impacts of System Parameters}

{\bf Purpose}:
Microgrids may employ different types of local generators
with diverse operational constraints (such ramping up/down limits and
minimum on/off times) and heat recovery efficiencies.  It is then
important to understand the impact on cost reduction due to these
parameters.
In this experiment, we study the cost
reduction provided by our offline and online algorithms under different
settings of
${\sf R}_{{\rm up}}$, ${\sf R}_{{\rm dw}}$ , ${\sf T}_{{\rm on}}$,
${\sf T}_{{\rm off}}$ and $\eta$.

{\bf Observations}:
Fig. \ref{fig:cr_Rud} and \ref{fig:cr_Tud}
show the impact of ramp limit and minimum on/off time, respectively,
on the performance of the algorithms. Note that for simplicity we
always set ${\sf R}_{{\rm up}}={\sf R}_{{\rm dw}}$ and ${\sf T}_{{\rm on}}
={\sf T}_{{\rm off}}$.
As we can see in Fig. \ref{fig:cr_Rud}, with ${\sf R}_{{\rm up}}$ and ${\sf
R}_{{\rm dw}}$ of about 40\% of the maximum capacity, {\sf
CHASE} obtains nearly all of the cost reduction benefits, compared with ${\sf RHC}$ which needs
70\% of the maximum capacity. Meanwhile, it can be seen from Fig. \ref{fig:cr_Tud}
that ${\sf T}_{{\rm on}}$ and ${\sf T}_{{\rm
off}}$ do not have much impact on the performance. This suggests that
it is more valuable to invest in generators with fast ramping up/down
capability than those with small minimum on/off periods.  From Fig.
\ref{fig:cr_eta_summer} and \ref{fig:cr_eta_winter}, we
observe that generators with large $\eta$ save much more cost during
the winter because of the high heat demand.  This suggests that in areas
with large heat demand, such as Alaska and Washington, the heat recovery efficiency ratio is a critical parameter when investing CHP generators.

\begin{figure}
\subfloat[ \label{fig:cr_Rud} cost redu. vs. ${\sf R}_{{\rm up}}$ and ${\sf R}_{{\rm dw}}$]{\includegraphics[width={0.5\columnwidth}]{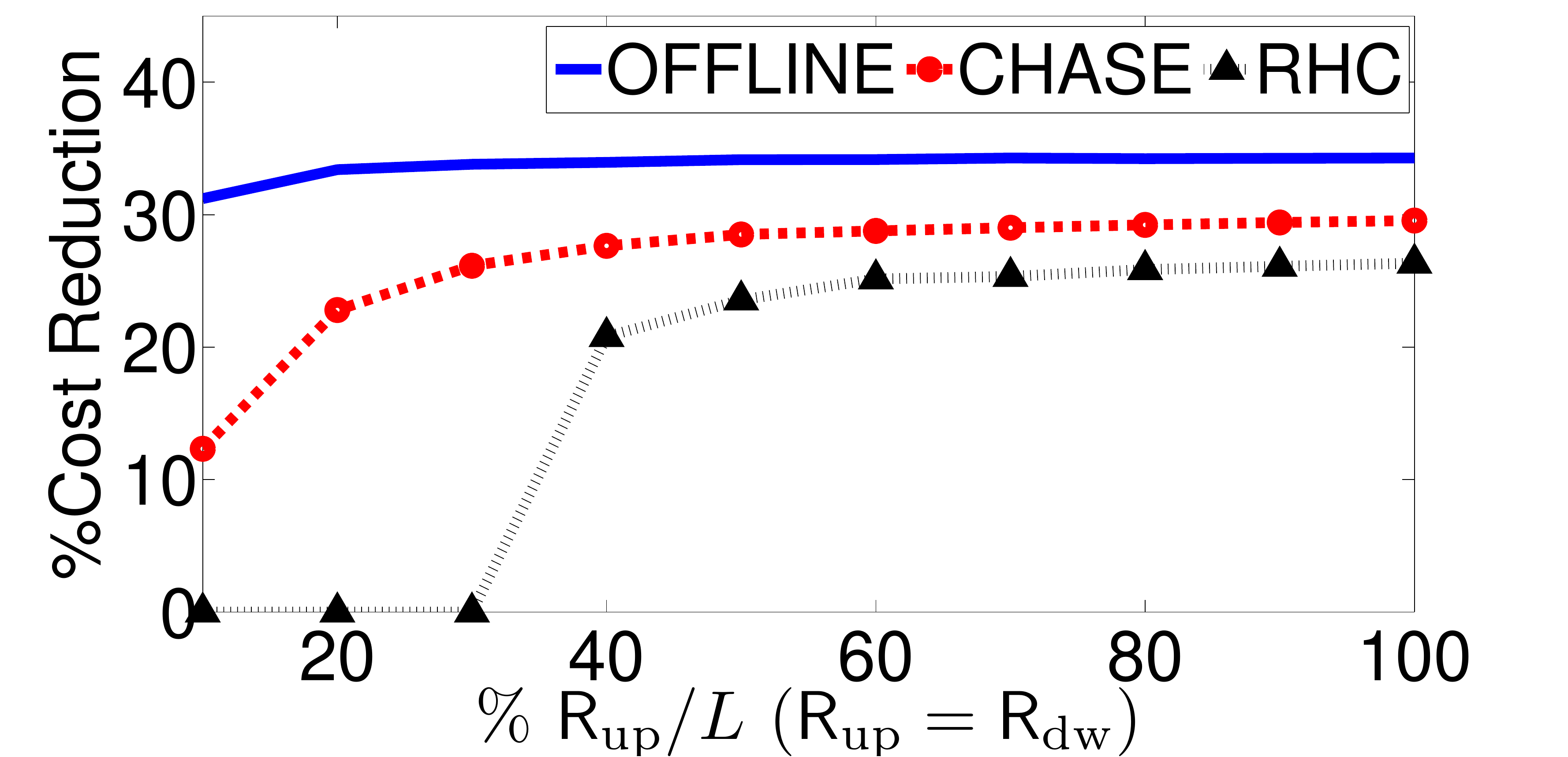}}
\subfloat[ \label{fig:cr_Tud} cost redu. vs. ${\sf T}_{{\rm up}}$ and ${\sf T}_{{\rm dw}}$]{\includegraphics[width={0.5\columnwidth}]{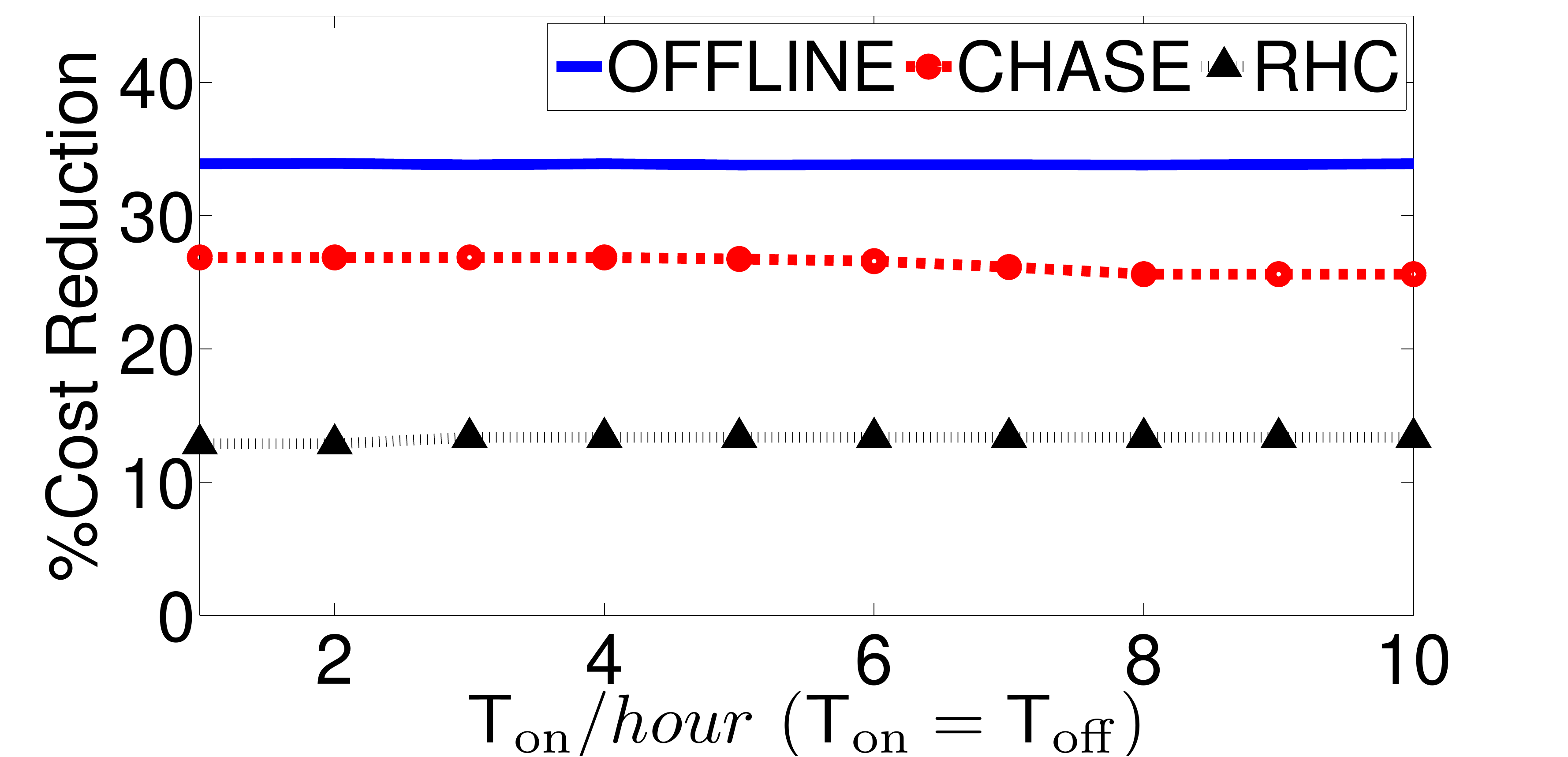}}\\
\subfloat[ \label{fig:cr_eta_summer} cost reduction vs. $\eta$]{\includegraphics[width={0.5\columnwidth}]{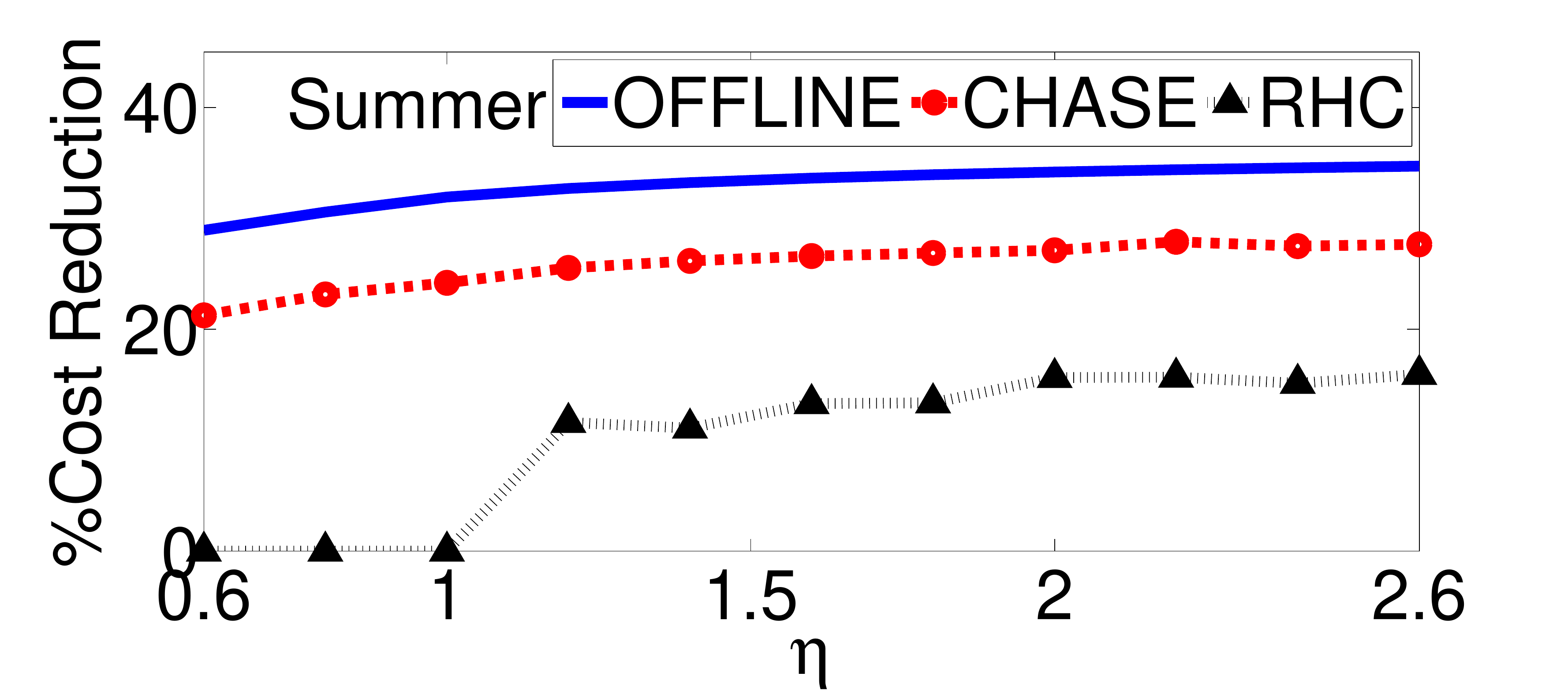}}
\subfloat[ \label{fig:cr_eta_winter} cost reduction vs. $\eta$]{\includegraphics[width={0.5\columnwidth}]{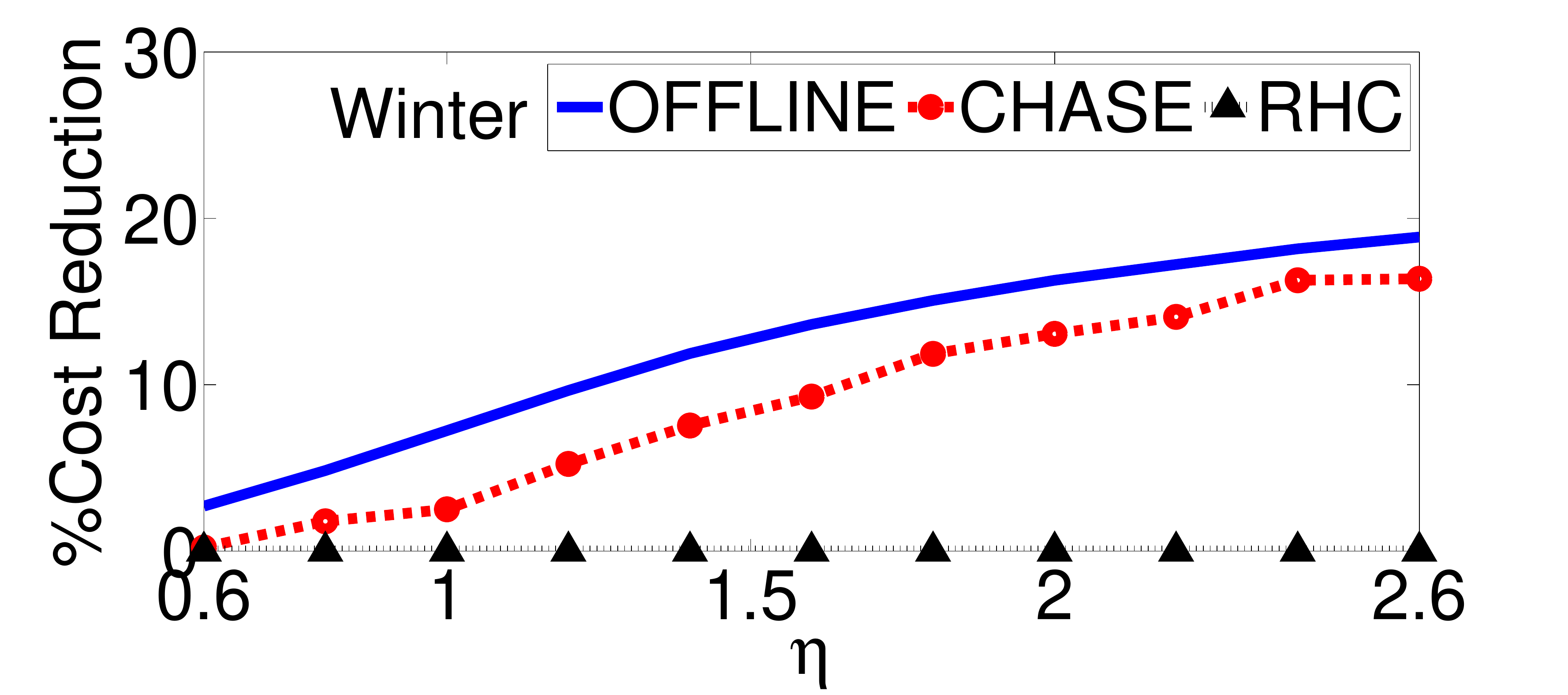}}\caption{  Cost reduction as functions of generator parameters.}
\end{figure}


%
%
%

\subsection{How Much Local Generation is Enough}

Thus far, we assumed that the microgrid had the ability to supply all
energy demand from local power generation in every time-slot.
In practice, local generators can be quite expensive. Hence,
an important question is how much investment should a microgrid operator
makes (in terms of the installed local generator capacity) in order to obtain
the maximum cost benefit. More specifically, we vary the number of CHP generators from
1 to 10 and plot the corresponding cost reductions of algorithms in Fig. \ref{fig:cr_yub}.
Interestingly, our results show that provisioning local generation to produce 60\% of the
peak demand is sufficient to obtain nearly all of the cost reduction
benefits.
Further, with just 50\% local generation capacity we can achieve
about 90\% of the maximum cost reduction. The intuitive reason is
that most of the time demands are significantly lower than their
peaks.
\section{Related Work}\label{sec:related}

Energy generation scheduling is a classical problem in power systems
and involves two aspects, namely Unit Commitment (UC) and Economic
Dispatching (ED). UC optimizes the startup and shutdown schedule of power generations to meet
the forecasted demand over a short period, whereas ED allocates the system demand and spinning
reserve capacity among operating units at each specific hour of operation without considering startup and shutdown of power generators.

For large power systems, UC involves scheduling of a large number
gigantic power plants of several hundred if not thousands of megawatts
with heterogeneous operating constraints and logistics behind each action
\cite{padhy2004unit}. The problem is very challenging to solve and
has been shown to be NP-Complete in general\footnote{We note that $\mathbf{fMCMP}$ in (3a)-(3d) is an instance of UC,
and that UC is NP-hard in general does not imply that the instance
$\mathbf{fMCMP}$ is also NP-hard.} \cite{guan2003optimization}.
Sophisticated
approaches proposed in the literature for solving UC include mixed integer programming
\cite{carrion2006computationally}, dynamic programming \cite{snyder1987dynamic},
and stochastic programming \cite{takriti1996stochastic}. There have also
been investigations on UC with high renewable energy penetration \cite{tuohy2009unit},  based on  over-provisioning approach.
After UC determines the on/off status of generators, ED computes their
output levels by solving a nonlinear optimization problem using various
heuristics without altering the on/off status of generators \cite{gaing2003particle}. There
is also recent interest in involving CHP generators in ED to satisfy
both electricity and heat demand simultaneously \cite{guo1996algorithm}.
See comprehensive surveys on UC in \cite{padhy2004unit}
and on ED in \cite{gaing2003particle}.

However, these studies assume the demand and energy supply (or their distributions) in the entire time horizon are
known \emph{a prior}. As such, the schemes are not readily applicable
to microgrid scenarios where accurate prediction of small-scale demand
and wind power generation is difficult to obtain due to limited management
resources and their unpredictable nature \cite{wang2008security}.

Several recent works have started to study energy generation strategies for microgrids. For
example, the authors in \cite{hawkes2009modelling} develop a linear
programming based cost minimization approach for UC in microgrids.
\cite{hernandez2005fuel} considers the fuel consumption rate minimization in microgrids and advocates to build ICT infrastructure in microgrids. \rev{\cite{respnose1,respnose2} discuss the energy scheduling problems in data centers, whose models are similar with ours.}
The difference between these works and ours is that they assume
the demand and energy supply are given
beforehand, and ours does not rely on input prediction.

Online optimization and algorithm design is an established approach in
optimizing the performance of various computer systems with minimum
knowledge of inputs \cite{borodin1998online,neely2010stochastic}. Recently, it has found
new applications in data centers \cite{gandhi2010optimality,buchbinder2011online,urgaonkar2011optimal,lin2011dynamic,lu2012simple,lu2012journal}. To the best of our knowledge, our work is the first to study the competitive online algorithms for energy generation in
microgrids with intermittent energy sources and co-generation. The
authors in \cite{narayanaswamy2012online} apply online convex
optimization framework \cite{zinkevich2003online} to design ED
algorithms for microgrids. The authors in \cite{huang2013adaptive} adopt Lyapunov optimization framework \cite{neely2010stochastic} to design electricity scheduling for microgrids, with consideration of energy storage. However, neither of the above considers the startup cost of the local generations.  In contrast, our work
jointly consider UC and ED in microgrids with
co-generation.
Furthermore, the above three works adopt different frameworks and provide online algorithms with different types of
performance guarantee.

\section{Conclusion}
In this paper, we study online
algorithms for the micro-grid generation scheduling problem with
intermittent renewable energy sources and co-generation, with the goal
of maximizing the cost-savings with local generation.
Based on insights from the structure of the offline optimal solution, we
propose a class of competitive online algorithms, called {\sf CHASE}
that track the offline optimal in an online fashion. Under typical settings, we show
that {\sf CHASE} achieves the best competitive ratio of all deterministic
online algorithms, and the ratio is no larger than a small constant 3.
We also extend our algorithms to intelligently leverage on {\em limited
prediction} of the future, such as near-term demand or wind forecast. By
extensive empirical evaluations using real-world traces, we show that our
proposed algorithms can achieve near offline-optimal performance.

There are a number of interesting directions for future work. First,
energy storage systems ({\em e.g.,} large-capacity battery) have been proposed
as an alternate approach to reduce energy generation cost (during peak hours) and
to integrate renewable energy sources. It would be interesting to study
whether our proposed microgrid control strategies can be combined with
energy storage systems to further reduce generation cost. However,
current energy storage systems can be very expensive. Hence,
it is critical to study whether the combined control strategy can
reduce sufficient cost with limited amount of energy storage.
Second, it remains an open issue whether {\sf CHASE} can achieve the best
competitive ratios in general cases ({\em e.g.,} in the slow-responding case).

\rev{\section{Acknowledgments}
The work described in this paper was partially supported by China National 973 projects (No.\ 2012CB315904 and 2013CB336700), several grants from the University Grants Committee of the Hong Kong Special Administrative Region, China (Area of Excellence Project No.\ AoE/E-02/08 and General Research Fund Project No.\ 411010 and 411011), two gift grants from Microsoft and Cisco, and Masdar Institute-MIT Collaborative Research Project No.\ 11CAMA1.
Xiaojun Lin would like to thank the Institute of Network Coding at The Chinese University of Hong Kong for the support of his sabbatical visit, during which some parts of the work were done.
}

\bibliographystyle{plain}
{
\bibliography{related}

\begin{thebibliography}{10}

\bibitem{CEUS}
California commercial end-use survey.
\newblock Internet:http://capabilities.itron.com/CeusWeb/.

\bibitem{greenenergy}
Green energy.
\newblock Internet:http://www.green-energy-uk.com/whatischp.html.

\bibitem{met}
The irish meteorological service online.
\newblock Internet:http://www.met.ie/forecasts/.

\bibitem{NREL}
National renewable energy laboratory.
\newblock Internet:http://wind.nrel.gov.

\bibitem{PG_E}
Pacific gas and electric company.
\newblock Internet:http://www.pge.com/nots/rates/tariffs/rateinfo.shtml.

\bibitem{tecogen}
Tecogen.
\newblock Internet:http://www.tecogen.com.

\bibitem{barnes2007real}
M.~Barnes, J.~Kondoh, H.~Asano, J.~Oyarzabal, G.~Ventakaramanan, R.~Lasseter,
  N.~Hatziargyriou, and T.~Green.
\newblock Real-world microgrids-an overview.
\newblock In {\em Proc. IEEE SoSE}, 2007.

\bibitem{borodin1998online}
A.~Borodin and R.~El-Yaniv.
\newblock {\em Online computation and competitive analysis}.
\newblock Cambridge University Press Cambridge, 1998.

\bibitem{buchbinder2011online}
N.~Buchbinder, N.~Jain, and I.~Menache.
\newblock Online job-migration for reducing the electricity bill in the cloud.
\newblock In {\em Proc. IFIP}, 2011.

\bibitem{carrion2006computationally}
M.~Carri{\'o}n and J.~Arroyo.
\newblock A computationally efficient mixed-integer linear formulation for the
  thermal unit commitment problem.
\newblock {\em IEEE Trans. Power Systems}, 21(3):1371--1378, 2006.

\bibitem{chiuelectric}
D.~Chiu, C.~Stewart, and B.~McManus.
\newblock Electric grid balancing through low-cost workload migration.
\newblock In {\em Proc. ACM Greenmetrics}, 2012.

\bibitem{chowdhury2003reliability}
A.~Chowdhury, S.~Agarwal, and D.~Koval.
\newblock Reliability modeling of distributed generation in conventional
  distribution systems planning and analysis.
\newblock {\em IEEE Trans. Industry Applications}, 39(5):1493--1498, 2003.

\bibitem{cullen2011dynamic}
Joseph~A Cullen.
\newblock Dynamic response to environmental regulation in the electricity
  industry.
\newblock {\em University of Arizona.(February 1, 2011)}, 2011.

\bibitem{gaing2003particle}
Z.~Gaing.
\newblock Particle swarm optimization to solving the economic dispatch
  considering the generator constraints.
\newblock {\em IEEE Trans. Power Systems}, 18(3):1187--1195, 2003.

\bibitem{gandhi2010optimality}
A.~Gandhi, V.~Gupta, M.~Harchol-Balter, and M.~Kozuch.
\newblock Optimality analysis of energy-performance trade-off for server farm
  management.
\newblock {\em Performance Evaluation}, 67(11):1155--1171, 2010.

\bibitem{guan2003optimization}
X.~Guan, Q.~Zhai, and A.~Papalexopoulos.
\newblock Optimization based methods for unit commitment: Lagrangian relaxation
  versus general mixed integer programming.
\newblock In {\em Proc. IEEE PES General Meeting}, 2003.

\bibitem{guo1996algorithm}
T.~Guo, M.~Henwood, and M.~van Ooijen.
\newblock An algorithm for combined heat and power economic dispatch.
\newblock {\em IEEE Trans. Power Systems}, 11(4):1778--1784, 1996.

\bibitem{hawkes2009modelling}
A.~Hawkes and M.~Leach.
\newblock Modelling high level system design and unit commitment for a
  microgrid.
\newblock {\em Applied energy}, 86(7):1253--1265, 2009.

\bibitem{hernandez2005fuel}
C.~Hernandez-Aramburo, T.~Green, and N.~Mugniot.
\newblock Fuel consumption minimization of a microgrid.
\newblock {\em IEEE Trans. Industry Applications}, 41(3):673--681, 2005.

\bibitem{windpredictionerror}
B.~Hodge and M.~Milligan.
\newblock Wind power forecasting error distributions over multiple timescales.
\newblock In {\em Proc. IEEE PES General Meeting}, 2011.

\bibitem{huang2013adaptive}
Y.~Huang, S.~Mao, and R.~Nelms.
\newblock Adaptive electricity scheduling in microgrids.
\newblock In {\em Proc. IEEE INFOCOM}, 2013.

\bibitem{kazarlis1996genetic}
S.~Kazarlis, A.~Bakirtzis, and V.~Petridis.
\newblock A genetic algorithm solution to the unit commitment problem.
\newblock {\em IEEE Trans. Power Systems}, 11(1):83--92, 1996.

\bibitem{RHC3}
W.~Kwon and A.~Pearson.
\newblock A modified quadratic cost problem and feedback stabilization of a
  linear system.
\newblock {\em IEEE Trans. Automatic Control}, 22(5):838 -- 842, 1977.

\bibitem{lasseter2004microgrid}
R.~Lasseter and P.~Paigi.
\newblock Microgrid: A conceptual solution.
\newblock In {\em Proc. IEEE Power Electronics Specialists Conference}, 2004.

\bibitem{respnose1}
K.~Le, R.~Bianchini, T.D. Nguyen, O.~Bilgir, and M.~Martonosi.
\newblock Capping the brown energy consumption of internet services at low
  cost.
\newblock In {\em Proc. IEEE IGCC}, 2010.

\bibitem{lin2011dynamic}
M.~Lin, A.~Wierman, L.~Andrew, and E.~Thereska.
\newblock Dynamic right-sizing for power-proportional data centers.
\newblock In {\em Proc. IEEE INFOCOM}, 2011.

\bibitem{respnose2}
Z.~Liu, Y.~Chen, C.~Bash, A.~Wierman, D.~Gmach, Z.~Wang, M.~Marwah, and
  C.~Hyser.
\newblock Renewable and cooling aware workload management for sustainable data
  centers.
\newblock In {\em Proc. ACM SIGMETRICS}, 2012.

\bibitem{lu2012simple}
T.~Lu and M.~Chen.
\newblock Simple and effective dynamic provisioning for power-proportional data
  centers.
\newblock In {\em Proc. CISS}, 2012.

\bibitem{lu2012journal}
T.~Lu, M.~Chen, and L.~Andrew.
\newblock Simple and effective dynamic provisioning for power-proportional data
  centers.
\newblock {\em IEEE Trans. Parallel Distrib. Systems}, 2012.

\bibitem{marnay2007microgrids}
C.~Marnay and R.~Firestone.
\newblock Microgrids: An emerging paradigm for meeting building electricity and
  heat requirements efficiently and with appropriate energy quality.
\newblock {\em European Council for an Energy Efficient Economy Summer Study},
  2007.

\bibitem{narayanaswamy2012online}
B.~Narayanaswamy, V.~Garg, and T.~Jayram.
\newblock Online optimization for the smart (micro) grid.
\newblock In {\em Proc. ACM International Conference on Future Energy Systems},
  2012.

\bibitem{neely2010stochastic}
M.~Neely.
\newblock Stochastic network optimization with application to communication and
  queueing systems.
\newblock {\em Synthesis Lectures on Communication Networks}, 3(1):1--211,
  2010.

\bibitem{DoEReport}
Department of~Energy.
\newblock The smart grid: An introduction.
\newblock Internet:http://www.oe.energy.gov/SmartGridIntroduction.htm.

\bibitem{padhy2004unit}
N.~Padhy.
\newblock Unit commitment-a bibliographical survey.
\newblock {\em IEEE Trans. Power Systems}, 19(2):1196--1205, 2004.

\bibitem{snyder1987dynamic}
W.~Snyder, H.~Powell, and J.~Rayburn.
\newblock Dynamic programming approach to unit commitment.
\newblock {\em IEEE Trans. Power Systems}, 2(2):339--348, 1987.

\bibitem{stadlerdistributed}
M.~Stadler, H.~Aki, R.~Lai, C.~Marnay, and A.~Siddiqui.
\newblock Distributed energy resources on-site optimization for commercial
  buildings with electric and thermal storage technologies.
\newblock {\em Lawrence Berkeley National Laboratory, LBNL-293E}, 2008.

\bibitem{takriti1996stochastic}
S.~Takriti, J.~Birge, and E.~Long.
\newblock A stochastic model for the unit commitment problem.
\newblock {\em IEEE Trans. Power Systems}, 11(3):1497--1508, 1996.

\bibitem{tuohy2009unit}
A.~Tuohy, P.~Meibom, E.~Denny, and M.~O'Malley.
\newblock Unit commitment for systems with significant wind penetration.
\newblock {\em IEEE Trans. Power Systems}, 24(2):592--601, 2009.

\bibitem{urgaonkar2011optimal}
R.~Urgaonkar, B.~Urgaonkar, M.~Neely, and A.~Sivasubramaniam.
\newblock Optimal power cost management using stored energy in data centers.
\newblock In {\em Proc. ACM SIGMETRICS}, 2011.

\bibitem{varaiya2011smart}
P.~Varaiya, F.~Wu, and J.~Bialek.
\newblock Smart operation of smart grid: Risk-limiting dispatch.
\newblock {\em Proc. the IEEE}, 99(1):40--57, 2011.

\bibitem{vuorinen2007planning}
A.~Vuorinen.
\newblock {\em Planning of optimal power systems}.
\newblock Ekoenergo Oy, 2007.

\bibitem{wang2008security}
J.~Wang, M.~Shahidehpour, and Z.~Li.
\newblock Security-constrained unit commitment with volatile wind power
  generation.
\newblock {\em IEEE Trans. Power Systems}, 23(3):1319--1327, 2008.

\bibitem{zinkevich2003online}
M.~Zinkevich.
\newblock Online convex programming and generalized infinitesimal gradient
  ascent.
\newblock In {\em Proc. Int. Conf. Mach. Learn.}, 2003.

\end{thebibliography}
}

\appendix
\section{Proof of Theorem 1} \label{subsec:OFA-optimal}

\setcounter{thm}{0}

\begin{thm}
$(y_{{\rm OFA}}(t))_{t=1}^{T}$ is an optimal solution for \textbf{SP}.
\end{thm}
\begin{proof}
Suppose $(y^{\ast}(t))_{t=1}^{T}$ is an optimal solution for \textbf{SP}.
For completeness, we let $y^{\ast}(0)=0$ and $y^{\ast}(T+1)=0$.
We define a sequence $(y_{0}(t))_{t=1}^{T},(y_{1}(t))_{t=1}^{T}$,
$...,(y_{k+1}(t))_{t=1}^{T}$ as follows:
\begin{enumerate}
\item $y_{0}(t)=y^{\ast}(t)$ for all $t\in[1,T]$.
\item For all $t\in[1,T]$ and $i=1,...,k$
\begin{equation}
y_{i}(t)=\begin{cases}
y_{{\rm OFA}}(t), & \mbox{if\ }t\in[1,T_{i}^{c}]\\
y^{\ast}(t), & \mbox{otherwise\ }
\end{cases}
\end{equation}

\item $y_{k+1}(t)=y_{{\rm OFA}}(t)$ for all $t\in[1,T]$.
\end{enumerate}
We next set the boundary conditions for each ${\bf SP}^{{\rm sg\mbox{-}}i}$
by
\begin{equation}
{\rm y}_{i}^{l}=y_{{\rm OFA}}(T_{i}^{c})\mbox{\ and\ }{\rm y}_{i}^{r}=y^{\ast}(T_{i+1}^{c}+1)
\end{equation}
It follows that
\begin{equation}
{\rm Cost}(y_{i})-{\rm Cost}(y_{i+1})={\rm Cost}^{{\rm sg\mbox{-}}i}(y^{\ast})-{\rm Cost}^{{\rm sg\mbox{-}}i}(y_{{\rm OFA}})
\end{equation}
By Lemma~\ref{lem:-is-theindep.with.bound.con}, we obtain ${\rm Cost}^{{\rm sg\mbox{-}}i}(y^{\ast})\ge{\rm Cost}^{{\rm sg\mbox{-}}i}(y_{{\rm OFA}})$
for all $i$. Hence,
\begin{equation}
\hspace{-20pt}{\rm Cost}(y^{\ast})={\rm Cost}(y_{0})\ge...\ge{\rm Cost}(y_{k+1})={\rm Cost}(y_{{\rm OFA}})
\end{equation}
\end{proof}

\setcounter{lem}{1}
\begin{lem}
\label{lem:-is-theindep.with.bound.con} $(y_{{\rm OFA}}(t))_{t=T_{i}^{c}+1}^{T_{i+1}^{c}}$
is an optimal solution for

${\bf SP}^{{\rm sg\mbox{-}}i}({\rm y}_{i}^{l},{\rm y}_{i}^{r})$,
despite any boundary conditions $({\rm y}_{i}^{l},{\rm y}_{i}^{r})$.
\end{lem}
\begin{proof}
Consider given any boundary condition $({\rm y}_{i}^{l},{\rm y}_{i}^{r})$
for ${\bf SP}^{{\rm sg\mbox{-}}i}$. Suppose $(\bar{y}(t))_{t=T_{i}^{c}+1}^{T_{i+1}^{c}}$
is an optimal solution for ${\bf SP}^{{\rm sg\mbox{-}}i}$ w.r.t.
$({\rm y}_{i}^{l},{\rm y}_{i}^{r})$, and $\bar{y}\ne y_{{\rm OFA}}$.
We aim to show ${\rm Cost}^{{\rm sg\mbox{-}}i}(\bar{y})\ge{\rm Cost}^{{\rm sg\mbox{-}}i}(y_{{\rm OFA}})$,
by considering the types of critical segment.

(\textbf{type-1}): First, suppose that critical segment $[T_{i}^{c}+1,T_{i+1}^{c}]$
is type-1. Hence, $y_{{\rm OFA}}(t)=1$ for all $t\in[T_{i}^{c}+1,T_{i+1}^{c}]$.
Hence,
\begin{equation}
{\rm Cost}^{{\rm sg\mbox{-}}i}(y_{{\rm OFA}})=\beta\cdot(1-{\rm y}_{i}^{l})+\sum_{t=T_{i}^{c}+1}^{T_{i+1}^{c}}\psi\big(\sigma(t),1\big)
\end{equation}

\textbf{Case 1}:
Suppose $\bar{y}(t)=0$ for all $t\in[T_{i}^{c}+1,T_{i+1}^{c}]$.
Hence,
\begin{equation}
{\rm Cost}^{{\rm sg\mbox{-}}i}(\bar{y})=\beta\cdot{\rm y}_{i}^{r}+\sum_{t=T_{i}^{c}+1}^{T_{i+1}^{c}}\psi\big(\sigma(t),0\big)
\end{equation}
We obtain:
\begin{eqnarray}
 &  & {\rm Cost}^{{\rm sg\mbox{-}}i}(\bar{y})-{\rm Cost}^{{\rm sg\mbox{-}}i}(y_{{\rm OFA}})\notag\\
 & = & \beta\cdot{\rm y}_{i}^{r}+\sum_{t=T_{i}^{c}+1}^{T_{i+1}^{c}}\delta(t)-\beta(1-{\rm y}_{i}^{l})\label{eqn:delta-def1}\\
 & \geq & \beta\cdot{\rm y}_{i}^{r}+\Delta(T_{i+1}^{c})-\Delta(T_{i}^{c})-\beta(1-{\rm y}_{i}^{l})\label{eqn:Delta-func1}\\
 & = & \beta\cdot{\rm y}_{i}^{r}+\beta-\beta+\beta{\rm y}_{i}^{l}\geq0
\end{eqnarray}
where Eqn.~(\ref{eqn:delta-def1}) follows from the definition of
$\delta(t)$ (see Eqn.~(\ref{eqn:delta-definition})) and Eqn.~(\ref{eqn:Delta-func1})
follows from Lemma~\ref{lem:Delta-function}. This completes the
proof for Case 1.

\textbf{Case 2}: Suppose $\bar{y}(t)=1$ for some $t\in[T_{i}^{c}+1,T_{i+1}^{c}]$.
This implies that ${\rm Cost}^{{\rm sg\mbox{-}}i}(\bar{y})$ has to
involve the startup cost $\beta$.

Next, we denote the minimal set of segments within $[T_{i}^{c}+1,T_{i+1}^{c}]$
by
\[
[\tau_{1}^{b},\tau_{1}^{e}],[\tau_{2}^{b},\tau_{2}^{e}],[\tau_{3}^{b},\tau_{3}^{e}],...,[\tau_{p}^{b},\tau_{p}^{e}]
\]
such that $\bar{y}[t]\ne y_{{\rm OFA}}(t)$ for all $t\in[\tau_{j}^{b},\tau_{j}^{e}]$,
$j\in\{1,...,p\}$, where $\tau_{j}^{e}<\tau_{j+1}^{b}$.

Since $\bar{y}\ne y_{{\rm OFA}}$, then there exists at least one
$t\in[T_{i}^{c}+1,T_{i+1}^{c}]$ such that $\bar{y}(t)=0$. Hence,
$\tau_{1}^{b}$ is well-defined.

Note that upon exiting each segment $[\tau_{j}^{b},\tau_{j}^{e}]$,
$\bar{y}$ switches from 0 to 1. Hence, it incurs the startup cost
$\beta$. However, when $\tau_{p}^{e}=T_{i+1}^{c}$ and ${\rm y}_{i}^{r}=0$,
the startup cost is not for critical segment $[T_{i}^{c}+1,T_{i+1}^{c}]$.

\rev{Therefore, we obtain:
\begin{eqnarray}
\hspace{-20pt} &  & {\rm Cost}^{{\rm sg\mbox{-}}i}(\bar{y})-{\rm Cost}^{{\rm sg\mbox{-}}i}(y_{{\rm OFA}})\label{eq:delta-ALL}\\
 & = & \sum_{t=\tau_{1}^{b}}^{\tau_{1}^{e}}\delta(t)+\beta\cdot\boldsymbol{1}[\tau_{1}^{b}\ne T_{i}^{c}+1]\label{eqn:delta-def21}\\
 &  & +\sum_{j=2}^{p-1}\Big(\sum_{t=\tau_{j}^{b}}^{\tau_{j}^{e}}\delta(t)+\beta\Big)\label{eq:delta-def22}\\
 &  & +\sum_{t=\tau_{p}^{b}}^{\tau_{p}^{e}}\delta(t)+\beta{\rm y}_{i}^{r}\cdot\boldsymbol{1}[\tau_{p}^{e}=T_{i+1}^{c}]+\beta\cdot\boldsymbol{1}[\tau_{p}^{e}\ne T_{i+1}^{c}].\label{eq:delta-def23}
\end{eqnarray}
Now we prove the terms (\ref{eqn:delta-def21})(\ref{eq:delta-def22})(\ref{eq:delta-def23})
are all larger than $0.$

\textbf{First}, we prove (\ref{eqn:delta-def21})$\geq0.$ If $\tau_{1}^{b}=T_{i}^{c}+1,$
then
\begin{eqnarray*}
\sum_{t=\tau_{1}^{b}}^{\tau_{1}^{e}}\delta(t)+\beta\cdot\boldsymbol{1}[\tau_{1}^{b}\ne T_{i}^{c}+1] & = & \sum_{t=T_{i}^{c}+1}^{\tau_{1}^{e}}\delta(t)\\
 & \geq & \Delta(\tau_{1}^{e})-\Delta(T_{i}^{c})\\
 & \geq & \Delta(\tau_{1}^{e})+\beta\\
 & \geq & 0.
\end{eqnarray*}

Else if $\tau_{1}^{b}\ne T_{i}^{c}+1,$

\begin{eqnarray*}
\sum_{t=\tau_{1}^{b}}^{\tau_{1}^{e}}\delta(t)+\beta\cdot\boldsymbol{1}[\tau_{1}^{b}\ne T_{i}^{c}+1] & = & \sum_{t=\tau_{1}^{b}}^{\tau_{1}^{e}}\delta(t)+\beta\\
 & \geq & \Delta(\tau_{1}^{e})-\Delta(\tau_{1}^{b}-1)+\beta\\
 & \geq & \Delta(\tau_{1}^{e})+\beta\\
 & \geq & 0.
\end{eqnarray*}

Thus we proved (\ref{eqn:delta-def21})$\geq0.$

\textbf{Second}, we prove (\ref{eq:delta-def22})$\geq0.$
\begin{eqnarray*}
\sum_{t=\tau_{j}^{b}}^{\tau_{j}^{e}}\delta(t)+\beta & \ge & \Delta(\tau_{j}^{e})-\Delta(\tau_{j}^{b}-1)+\beta\\
 & \ge & \Delta(\tau_{j}^{e})+\beta\\
 & \ge & 0.
\end{eqnarray*}

Thus we proved (\ref{eq:delta-def22})$\geq0.$

\textbf{Last}, we prove (\ref{eq:delta-def23})$\geq0.$ If $\tau_{p}^{e}=T_{i+1}^{c},$
then
\begin{eqnarray*}
 &  & \sum_{t=\tau_{p}^{b}}^{\tau_{p}^{e}}\delta(t)+\beta{\rm y}_{i}^{r}\cdot\boldsymbol{1}[\tau_{p}^{e}=T_{i+1}^{c}]+\beta\cdot\boldsymbol{1}[\tau_{p}^{e}\ne T_{i+1}^{c}]\\
 & \ge & \sum_{t=\tau_{p}^{b}}^{T_{i+1}^{c}}\delta(t)\\
 & \geq & \Delta(T_{i+1}^{c})-\Delta(\tau_{p}^{b}-1)\\
 & = & -\Delta(\tau_{p}^{b}-1)\\
 & \ge & 0.
\end{eqnarray*}

Else if $\tau_{p}^{e}\ne T_{i+1}^{c}$, then
\begin{eqnarray*}
 &  & \sum_{t=\tau_{p}^{b}}^{\tau_{p}^{e}}\delta(t)+\beta{\rm y}_{i}^{r}\cdot\boldsymbol{1}[\tau_{p}^{e}=T_{i+1}^{c}]+\beta\cdot\boldsymbol{1}[\tau_{p}^{e}\ne T_{i+1}^{c}]\\
 & = & \sum_{t=\tau_{p}^{b}}^{\tau_{p}^{e}}\delta(t)+\beta\\
 & \ge & \Delta(\tau_{p}^{e})-\Delta(\tau_{p}^{b}-1)+\beta\\
 & \ge & 0.
\end{eqnarray*}
Thus we proved (\ref{eq:delta-def23})$\geq0.$

Overall, we prove (\ref{eq:delta-ALL})$\ge0.$}

(\textbf{type-2}): Next, suppose that critical segment $[T_{i}^{c}+1,T_{i+1}^{c}]$
is type-2. Hence, $y_{{\rm OFA}}(t)=0$ for all $t\in[T_{i}^{c}+1,T_{i+1}^{c}]$.
Note that the above argument applies similarly to type-2 setting,
when we consider (Case 1): $\bar{y}(t)=1$ for all $t\in[T_{i}^{c}+1,T_{i+1}^{c}]$
and (Case 2): $\bar{y}(t)=0$ for some $t\in[T_{i}^{c}+1,T_{i+1}^{c}]$.

(\textbf{type-start} and \textbf{type-end}): We note that the argument
of type-2 applies similarly to type-start and type-end settings.

Therefore, we complete the proof by showing ${\rm Cost}^{{\rm sg\mbox{-}}i}(\bar{y})\ge{\rm Cost}^{{\rm sg\mbox{-}}i}(y_{{\rm OFA}})$
for all $i\in[0,k]$. \end{proof}
\begin{lem}
\label{lem:Delta-function} Suppose $\tau_{1},\tau_{2}\in[T_{i}^{c}+1,T_{i+1}^{c}]$
and $\tau_{1}<\tau_{2}$. Then,
\begin{equation}
\Delta(\tau_{2})-\Delta(\tau_{1})\begin{cases}
\le\sum_{t=\tau_{1}+1}^{\tau_{2}}\delta(t), & \mbox{if\ }[T_{i}^{c}+1,T_{i+1}^{c}]\mbox{\ is type-1}\\
\ge\sum_{t=\tau_{1}+1}^{\tau_{2}}\delta(t), & \mbox{if\ }[T_{i}^{c}+1,T_{i+1}^{c}]\mbox{\ is type-2}
\end{cases}
\end{equation}

\end{lem}
\begin{proof} We recall that
\begin{equation}
\Delta(t)\triangleq\min\Big\{0,\max\{-\beta,\Delta(t-1)+\delta(t)\}\Big\}
\end{equation}
First, we consider $[T_{i}^{c}+1,T_{i+1}^{c}]$ as type-1. This implies
that only $\Delta(T_{i}^{c})=-\beta$, whereas $\Delta(t)>-\beta$
for $t\in[T_{i}^{c}+1,T_{i+1}^{c}]$. Hence,
\begin{equation}
\Delta(t)=\min\{0,\Delta(t-1)+\delta(t)\}\le\Delta(t-1)+\delta(t)
\end{equation}
Iteratively, we obtain
\begin{equation}
\Delta(\tau_{2})\le\Delta(\tau_{1})+\sum_{t=\tau_{1}+1}^{\tau_{2}}\delta(t)
\end{equation}

When $[T_{i}^{c}+1,T_{i+1}^{c}]$ is type-2, we proceed with a similar
proof, except
\begin{equation}
\Delta(t)=\max\{-\beta,\Delta(t-1)+\delta(t)\}\ge\Delta(t-1)+\delta(t)
\end{equation}
Therefore,
\begin{equation}
\Delta(\tau_{2})\ge\Delta(\tau_{1})+\sum_{t=\tau_{1}+1}^{\tau_{2}}\delta(t)
\end{equation}
\end{proof}

\section{Proof of Theorem 2} \label{subsec:CHASE-competitive-ratio}
\setcounter{thm}{1}
\begin{thm}
The competitive ratio of ${\sf CHASE_{s}}$
\begin{equation}
{\sf CR}({\sf CHASE_{s}})\le3-\frac{2(L\cdot c_{o}+c_{m})}{L\cdot (P_{\max}+c_{g}\cdot\eta)}
\end{equation}

\end{thm}
\begin{proof} We denote the outcome of ${\sf CHASE_{s}}$ by $(y_{{\rm CHASE_{s}}}(t))_{t=1}^{T}$.
We aim to show that
\begin{equation}
\max_{a,p,h}\frac{{\rm Cost}(y_{{\rm CHASE_{s}}})}{{\rm Cost}(y_{{\rm OFA}})}\le3-\frac{2(L\cdot c_{o}+c_{m})}{L\cdot (P_{\max}+c_{g}\cdot\eta)}
\end{equation}
First, we denote the set of indexes of critical segments for type-$j$
by ${\cal T}_{j}\subseteq\{0,..,k\}$. Note that we also refer to
type-start and type-end by type-0 and type-3 respectively.

Define the sub-cost for type-$j$ by
\begin{eqnarray}
{\rm Cost}^{{\rm ty\mbox{-}}j}(y) & \triangleq & \sum_{i\in{\cal T}_{j}}\sum_{t=T_{i}^{c}+1}^{T_{i+1}^{c}}\psi\big(\sigma(t),y(t)\big)\notag\\
 &  & +\beta\cdot[y(t)-y(t-1)]^{+}
\end{eqnarray}
Hence, ${\rm Cost}(y)=\sum_{j=0}^{3}{\rm Cost}^{{\rm ty\mbox{-}}j}(y)$.
We prove by comparing the sub-cost for each type-$j$.

(\textbf{Type-0}): Note that both $y_{{\rm OFA}}(t)=y_{{\rm CHASE_{s}}}(t)=0$
for all $t\in[1,T]$. Hence,
\begin{equation}
{\rm Cost}^{{\rm ty\mbox{-}}0}(y_{{\rm OFA}})={\rm Cost}^{{\rm ty\mbox{-}}0}(y_{{\rm CHASE_{s}}})
\end{equation}

(\textbf{Type-1}): Based on the definition of critical segment (Definition~\ref{def:critical-seg}),
we recall that there is an auxiliary point $\tilde{T_{i}^{c}}$, such
that either \big($\Delta(T_{i}^{c})=0$ and $\Delta(\tilde{T_{i}^{c}})=-\beta$\big)
or \big($\Delta(T_{i}^{c})=-\beta$ and $\Delta(\tilde{T_{i}^{c}})=0$\big).

We also recall that $y_{{\rm OFA}}(t)=1$, whereas
\begin{equation}
y_{{\rm CHASE_{s}}}(t)=\begin{cases}
0, & \mbox{if\ }t\in[T_{i}^{c}+1,\tilde{T_{i}^{c}}]\\
1, & \mbox{if\ }t\in[\tilde{T_{i}^{c}},T_{i+1}^{c}]
\end{cases}
\end{equation}
We consider a particular type-1 critical segment $[T_{i}^{c}+1,T_{i+1}^{c}]$.
We note that by the definition of type-1, $y_{{\rm OFA}}(T_{i}^{c})=y_{{\rm CHASE_{s}}}(T_{i}^{c})=0$.
$y_{{\rm OFA}}(t)$ and $y_{{\rm CHASE_{s}}}(t)$ switch from 0 to 1 within
$[T_{i}^{c}+1,T_{i+1}^{c}]$, both incurring startup cost $\beta$.
The cost difference between $y_{{\rm CHASE_{s}}}$
and $y_{{\rm OFA}}$ within $[T_{i}^{c}+1,T_{i+1}^{c}]$ is
\begin{eqnarray}
 &  & \sum_{t=T_{i}^{c}+1}^{\tilde{T}_{i}^{c}-1}\big(\psi\big(\sigma(t),0\big)-\psi\big(\sigma(t),1\big)\big) +\beta-\beta\\
 & = & \sum_{t=T_{i}^{c}+1}^{\tilde{T}_{i}^{c}-1}\big(\psi\big(\sigma(t),0\big)-\psi\big(\sigma(t),1\big)\big) \notag\\
 & = & \sum_{t=T_{i}^{c}+1}^{\tilde{T}_{i}^{c}-1}\delta(t)=\Delta(\tilde{T}_{i}^{c}-1)-\Delta(T_{i}^{c})\\
 & \le & \Delta(\tilde{T}_{i}^{c})-\Delta(T_{i}^{c})=\beta
\end{eqnarray}
Let the number of type-$j$ critical segments be $m_{j}\triangleq|{\cal T}_{j}|$.

Then, we obtain
\begin{equation}
{\rm Cost}^{{\rm ty\mbox{-}}1}(y_{{\rm CHASE_{s}}})\le{\rm Cost}^{{\rm ty\mbox{-}}1}(y_{{\rm OFA}})+m_{1}\cdot\beta
\end{equation}

(\textbf{Type-2}) and (\textbf{Type-3}): We consider a particular
type-2 (or type-3) critical segment $[T_{i}^{c}+1,T_{i+1}^{c}]$,
we derive similarly for $j=2$ or 3 as
\begin{equation}
{\rm Cost}^{{\rm ty\mbox{-}}j}(y_{{\rm CHASE_{s}}})\le{\rm Cost}^{{\rm ty\mbox{-}}j}(y_{{\rm OFA}})+m_{j}\cdot\beta
\end{equation}

Furthermore, we note $m_{1}=m_{2}+m_{3}$, because it takes equal
numbers of critical segments for increasing $\Delta(\cdot)$ from
$-\beta$ to 0 and for decreasing from 0 to $-\beta$.

Overall, we obtain
\begin{equation}
\hspace{-25pt}\begin{array}{@{}r@{}l@{\ }l}
{\displaystyle \frac{{\rm Cost}(y_{{\rm CHASE_{s}}})}{{\rm Cost}(y_{{\rm OFA}})}} & = & {\displaystyle \frac{\sum_{j=0}^{3}{\rm Cost}^{{\rm ty\mbox{-}}j}(y_{{\rm CHASE_{s}}})}{\sum_{j=0}^{3}{\rm Cost}^{{\rm ty\mbox{-}}j}(y_{{\rm OFA}})}}\\
 & \leq & {\displaystyle \frac{(m_{1}+m_{2}+m_{3})\beta+\sum_{j=0}^{3}{\rm Cost}^{{\rm ty\mbox{-}}j}(y_{{\rm OFA}})}{\sum_{j=0}^{3}{\rm Cost}^{{\rm ty\mbox{-}}j}(y_{{\rm OFA}})}}\\
 & \leq & {\displaystyle 1+\frac{2m_{1}\beta}{\sum_{j=0}^{3}{\rm Cost}^{{\rm ty\mbox{-}}j}(y_{{\rm OFA}})}}\\
 & \leq & 1+\begin{cases}
0 & \mbox{if\ }m_{1}=0,\\
{\displaystyle \frac{2m_{1}\beta}{{\rm Cost}^{{\rm ty\mbox{-}}1}(y_{{\rm OFA}})}} & \mbox{otherwise}
\end{cases}\label{eq:CHASE_{s}ratioeq1}
\end{array}
\end{equation}
By Lemma~\ref{lem:type-1.cost.bound} and simplications, we obtain
\begin{equation}
\frac{{\rm Cost}(y_{{\rm CHASE_{s}}})}{{\rm Cost}(y_{{\rm OFA}})}\le3-\frac{2(L\cdot c_{o}+c_{m})}{L\cdot\big(P_{\max}+\eta \cdot c_{g}\big)}
\end{equation}
\end{proof}

\begin{lem}
\label{lem:type-1.cost.bound}
\begin{equation}
{\rm Cost}^{{\rm ty\mbox{-}}1}(y_{{\rm OFA}})\ge m_{1}\Big(\beta+\beta\frac{L \cdot c_{o}+c_{m}}{L\big(P_{\max}+\eta \cdot c_{g}-c_{o}\big)-c_{m}}\Big)
\end{equation}
\end{lem}
\begin{proof}
Considering Type-1 critical segment, we have
\begin{equation} \label{eq:type-1costeq1}
{\rm Cost}^{{\rm ty\mbox{-}}1}(y_{{\rm OFA}})\\
 = {\displaystyle \sum_{i\in{\cal T}_{1}}\Big(\beta+\sum_{t=T_{i}^{c}+1}^{T_{i+1}^{c}}\psi\big(\sigma(t),1\big)\Big)}
\end{equation}
On the other hand, we obtain
\begin{equation}
\hspace{-20pt}\begin{array}{@{}r@{}l@{\ }l}
 &  & {\displaystyle \sum_{t=T_{i}^{c}+1}^{T_{i+1}^{c}}\Big(\psi\big(\sigma(t),1\big)-c_{m}\Big)}\\
 & = & \displaystyle\frac{\sum_{t=T_{i}^{c}+1}^{T_{i+1}^{c}}\Big(\psi\big(\sigma(t),1\big)-c_{m}\Big)}{\sum_{t=T_{i}^{c}+1}^{T_{i+1}^{c}}\Big(\psi\big(\sigma(t),0\big)-\psi\big(\sigma(t),1\big)+c_{m}\Big)}\\
 &  & \times{\displaystyle \sum_{t=T_{i}^{c}+1}^{T_{i+1}^{c}}\Big(\psi\big(\sigma(t),0\big)-\psi\big(\sigma(t),1\big)+c_{m}\Big)}\\
 & \geq & \displaystyle\min_{\tau\in[T_{i}^{c}+1,T_{i+1}^{c}]}\frac{\psi\big(\sigma(\tau),1\big)-c_{m}}{\psi\big(\sigma(\tau),0\big)-\psi\big(\sigma(\tau),1\big)+c_{m}}\\
 &  & {\displaystyle \times\sum_{t=T_{i}^{c}+1}^{T_{i+1}^{c}}\Big(\psi\big(\sigma(t),0\big)-\psi\big(\sigma(t),1\big)+c_{m}\Big)\label{eq:type-1costeq2}}
\end{array}
\end{equation}
whereas Lemma~\ref{lem:segment-1 minimum cost} shows
\begin{equation}
\min_{\tau\in[T_{i}^{c}+1,T_{i+1}^{c}]}\frac{\psi\big(\sigma(\tau),1\big)-c_{m}}{\psi\big(\sigma(\tau),0\big)-\psi\big(\sigma(\tau),1\big)+c_{m}}\geq\frac{c_{o}}{P_{\max}+\eta \cdot c_{g}-c_{o}}
\end{equation}
and
\begin{equation}
\begin{array}{@{}r@{}l@{\ }l}
 &  & {\displaystyle \sum_{t=T_{i}^{c}+1}^{T_{i+1}^{c}}\Big(\psi\big(\sigma(t),0\big)-\psi\big(\sigma(t),1\big)+c_{m}\Big)}\\
 & = & {\displaystyle \sum_{t=T_{i}^{c}+1}^{T_{i+1}^{c}}\Big(\psi\big(\sigma(t),0\big)-\psi\big(\sigma(t),1\big)\Big)+(T_{i+1}^{c}-T_{i}^{c})c_{m}}\\
 & \geq & {\displaystyle \beta+(T_{i+1}^{c}-T_{i}^{c})c_{m}\label{eq:type-1costeq3}}
\end{array}
\end{equation}
Furthermore, we note that $T_{i+1}^{c}-T_{i}^{c}$ is lower bounded
by the steepest descend when $p(t)=P_{\max}$, $a(t)=L$ and $h(t)=\eta L$
for all $t\in[T_{i}^{c}+1,T_{i+1}^{c}]$:
\begin{equation}
T_{i+1}^{c}-T_{i}^{c}\geq\frac{\beta}{L\cdot\big(P_{\max}+\eta \cdot c_{g}-c_{o}\big)-c_{m}}
\end{equation}

Together, we obtain:
\begin{eqnarray*}
 &  & \sum_{t=T_{i}^{c}+1}^{T_{i+1}^{c}}\psi\big(\sigma(t),1\big)\\
 & = & \sum_{t=T_{i}^{c}+1}^{T_{i+1}^{c}}\big(\psi\big(\sigma(t),1\big)-c_{m}\big)+\sum_{t=T_{i}^{c}+1}^{T_{i+1}^{c}}c_{m}\\
 & \geq & \frac{c_{o}\big(\beta+(T_{i+1}^{c}-T_{i}^{c})c_{m}\big)}{P_{\max}+\eta \cdot c_{g}-c_{o}}+(T_{i+1}^{c}-T_{i}^{c})c_{m}\\
 & \geq & \beta\frac{L \cdot c_{o}+c_{m}}{L\cdot\big(P_{\max}+\eta \cdot c_{g}-c_{o}\big)-c_{m}}
\end{eqnarray*}

Therefore,
\begin{equation}
\begin{array}{@{}r@{}l@{\ }l}
 &  & {\rm Cost}^{{\rm ty\mbox{-}}1}(y_{{\rm OFA}})\\
 & = & \sum_{i\in{\cal T}_{1}}\Big(\beta+\sum_{t=T_{i}^{c}+1}^{T_{i+1}^{c}}\big(\psi\big(\sigma(t),1\big)\big)\Big)\\
 & \geq & {\displaystyle m_{1}\Big(\beta+\beta\frac{L \cdot c_{o}+c_{m}}{L\cdot\big(P_{\max}+\eta \cdot c_{g}-c_{o}\big)-c_{m}}\Big)}
\end{array}
\end{equation}
\end{proof}

\begin{lem} \label{lem:segment-1 minimum cost}
\begin{equation} \hspace{-20pt}
\min_{\tau\in[T_{i}^{c}+1,T_{i+1}^{c}]}\frac{\psi\big(\sigma(\tau),1\big)-c_{m}}{\psi\big(\sigma(\tau),0\big)-\psi\big(\sigma(\tau),1\big)+c_{m}}\geq\frac{c_{o}}{P_{\max}+\eta \cdot c_{g}-c_{o}}
\end{equation}
\end{lem}
\begin{proof}
We expand $\psi(\sigma(t),y(t))$ for each case:

{\em Case 1}: $c_{o}\geq p(t)+\eta \cdot c_{g}$.
When $y(t)=1$, by Lemma \ref{lem:fMCMP}
\begin{equation}
u(t)=0, v(t)=a(t), s(t)=h(t)
\end{equation}
Thus,
\begin{eqnarray}
\psi\big(\sigma(\tau),1\big)&=&p(t)a(t)+c_{g}s(t)+c_{m} \\
\psi\big(\sigma(\tau),0\big)&=&p(t)a(t)+c_{g}s(t)
\end{eqnarray}
Therefore,
\begin{equation}
\frac{\psi\big(\sigma(\tau),1\big)-c_{m}}{\psi\big(\sigma(\tau),0\big)-\psi\big(\sigma(\tau),1\big)+c_{m}} = \infty
\end{equation}

{\em Case 2}: $c_{o}\leq p(t)$.
When $y(t)=1$, by Lemma \ref{lem:fMCMP}
\begin{equation}
u(t)=a(t), v(t)=0, s(t)=\big[h(t)-\eta\cdot a(t)\big]^{+}
\end{equation}
Thus,
\begin{eqnarray}
\psi\big(\sigma(\tau),1\big)&=&c_{o}a(t)+c_{g}\big[h(t)-\eta\cdot a(t)\big]^{+}+c_{m} \\
\psi\big(\sigma(\tau),0\big)&=&p(t)a(t)+c_{g}h(t)
\end{eqnarray}
Therefore,
\begin{eqnarray*}
 &  & \frac{\psi\big(\sigma(\tau),1\big)-c_{m}}{\psi\big(\sigma(\tau),0\big)-\psi\big(\sigma(\tau),1\big)+c_{m}}\\
 & = & \frac{c_{o}a(t)+c_{g}\big[h(t)-\eta\cdot a(t)\big]^{+}}{\big(p(t)-c_{o}\big)a(t)+c_{g}\big(h(t)-\big[h(t)-\eta\cdot a(t)\big]^{+}\big)}\\
 & \geq & \frac{c_{o}a(t)}{\big(p(t)-c_{o}\big)a(t)+c_{g}\min\{h(t),\eta \cdot a(t)\}}\\
 & \geq & \frac{c_{o}a(t)}{\big(p(t)-c_{o}\big)a(t)+\eta \cdot c_{g}a(t)}\\
 & = & \frac{c_{o}}{p(t)-c_{o}+\eta \cdot c_{g}}
\end{eqnarray*}

{\em Case 3}: $p(t)+\eta \cdot c_{g}>c_{o}>p(t)$.
When $y(t)=1$, by Lemma \ref{lem:fMCMP}
\begin{eqnarray}
u(t)&=&\min\big\{\frac{h(t)}{\eta},a(t)\big\}, \\
v(t)&=&\max\big\{0,a(t)-\frac{h(t)}{\eta}\big\}, \\
s(t)&=&\max\big\{0,h(t)-\eta \cdot a(t)\big\}
\end{eqnarray}
Thus,
\begin{eqnarray}
\psi\big(\sigma(\tau),1\big)&=&c_{o}\min\big\{\frac{h(t)}{\eta},a(t)\big\}+p(t)\big[a(t)-\frac{h(t)}{\eta}\big]^{+} \notag\\
&&+c_{g}\big[h(t)-\eta \cdot a(t)\big]^{+}+c_{m} \\
\psi\big(\sigma(\tau),0\big)&=&p(t)a(t)+c_{g}h(t)
\end{eqnarray}
Therefore,
\begin{equation} \hspace{-10pt}
\begin{array}{@{}r@{\ }r@{\ }l@{}}
& &\psi\big(\sigma(\tau),0\big)-\psi\big(\sigma(\tau),1\big)+c_{m}\\
&=&p(t)\min\big\{a(t),\frac{h(t)}{\eta}\big\}+c_{g}\min\big\{h(t),\eta \cdot a(t)\big\} \notag -c_{o}\min\big\{\frac{h(t)}{\eta},a(t)\big\}
\end{array}
\end{equation}
and
\begin{eqnarray}
 &  & \frac{\psi\big(\sigma(\tau),1\big)-c_{m}}{\psi\big(\sigma(\tau),0\big)-\psi\big(\sigma(\tau),1\big)+c_{m}} \notag \\
 & = & \begin{cases} \displaystyle
\frac{c_{o}a(t)+c_{g}\big(h(t)-\eta \cdot a(t)\big)}{p(t)a(t)+c_{g}\eta \cdot a(t)-c_{o}a(t)} & \mbox{if }h(t)-\eta \cdot a(t)>0 \\
\displaystyle
\frac{c_{o}\frac{h(t)}{\eta}+p(t)\big(a(t)-\frac{h(t)}{\eta}\big)}{p(t)\frac{h(t)}{\eta}+c_{g}h(t)-c_{o}\frac{h(t)}{\eta}} & \mbox{otherwise}
\end{cases} \notag\\
 & \geq & \frac{c_{o}}{p(t)+\eta \cdot c_{g}-c_{o}}
\end{eqnarray}

Combing all the cases, we obtain
\[
\min_{\tau\in[T_{i}^{c}+1,T_{i+1}^{c}]}\frac{\psi\big(\sigma(\tau),1\big)-c_{m}}{\psi\big(\sigma(\tau),0\big)-\psi\big(\sigma(\tau),1\big)+c_{m}}\geq\frac{c_{o}}{P_{\max}+\eta \cdot c_{g}-c_{o}}
\]
\end{proof}

\section{Proof of Theorem 3} \label{subsec:lower-bound-online}

\begin{lem} Denote an online algorithm by ${\cal A}$, and an input
sequence by $\sigma=(a(t),h(t),p(t))_{t=1}^{T}$. More specifically,
we write ${\rm Cost}(y_{{\cal A}},u_{{\cal A}},v_{{\cal A}};\sigma)$
and ${\rm Cost}(y_{{\cal A}};\sigma)$, when it explicitly refers
to input sequence by $\sigma$. Define
\begin{equation}
\underline{c}({\bf fMCMP})\triangleq\min_{{\cal A}}\max_{\sigma}\frac{{\rm Cost}(y_{{\cal A}},u_{{\cal A}},v_{{\cal A}},s_{{\cal A}};\sigma)}{{\rm Cost}(y_{{\rm OFA}},u_{{\rm OFA}},v_{{\rm OFA}},s_{{\rm OFA}};\sigma)}
\end{equation}
\begin{equation}
\underline{c}({\bf SP})\triangleq\min_{{\cal A}}\max_{\sigma}\frac{{\rm Cost}(y_{{\cal A}};\sigma)}{{\rm Cost}(y_{{\rm OFA}};\sigma)}
\end{equation}
We have
\begin{equation}
\underline{c}({\bf SP})\le\underline{c}({\bf fMCMP})
\end{equation}
\end{lem} \begin{proof} We prove this lemma by contradiction. Suppose
that there exists a deterministic online algorithm ${\cal A}$ for
\textbf{fMCMP} with output $(y_{{\cal A}},u_{{\cal A}},v_{{\cal A}},s_{{\cal A}})$,
such that
\begin{equation}
\max_{\sigma}\frac{{\rm Cost}(y_{{\cal A}},u_{{\cal A}},v_{{\cal A}},s_{{\cal A}};\sigma)}{{\rm Cost}(y_{{\rm OFA}},u_{{\rm OFA}},v_{{\rm OFA}},s_{{\rm OFA}};\sigma)}<\underline{c}({\bf SP})
\end{equation}

Also, it follows that for any an input sequence $\sigma'$,
\begin{eqnarray}
 &  & \frac{{\rm Cost}(y_{{\cal A}},u_{{\cal A}},v_{{\cal A}},s_{{\cal A}};\sigma')}{{\rm Cost}(y_{{\rm OFA}},u_{{\rm OFA}},v_{{\rm OFA}},s_{{\rm OFA}};\sigma')}\\
 & \le & \max_{\sigma}\frac{{\rm Cost}(y_{{\cal A}},u_{{\cal A}},v_{{\cal A}},s_{{\cal A}};\sigma)}{{\rm Cost}(y_{{\rm OFA}},u_{{\rm OFA}},v_{{\rm OFA}},s_{{\rm OFA}};\sigma)}
\end{eqnarray}

It follows that (by Lemma~\ref{lem:fMCMP})
\begin{equation}
{{\rm Cost}(y_{{\cal A}};\sigma)}\le{{\rm Cost}(y_{{\cal A}},u_{{\cal A}},v_{{\cal A}};\sigma)}
\end{equation}

Based on ${\cal A}$, we can construct an online algorithm ${\cal A}'$
for \textbf{SP}, such that $y_{{\cal A}'}=y_{{\cal A}}$. By Lemma~\ref{lem:fMCMP},
\begin{equation}
{{\rm Cost}(y_{{\rm OFA}};\sigma)}={{\rm Cost}(y_{{\rm OFA}},u_{{\rm OFA}},v_{{\rm OFA}};\sigma)}
\end{equation}

Therefore, we obtain
\begin{equation}
\frac{{\rm Cost}(y_{{\cal A}};\sigma')}{{\rm Cost}(y_{{\rm OFA}};\sigma')}<\underline{c}({\bf SP})
\end{equation}
However, as $\underline{c}({\bf SP})$ is a lower bound of competitive
ratio for any deterministic online algorithm for \textbf{SP}. This
is contradiction, and it completes our proof. \end{proof}

\begin{thm} The competitive ratio for any deterministic online algorithm
${\cal A}$ for ${\bf SP}$ is lower bounded by a function:
\begin{equation}
{\sf CR}({\cal A})\ge{\sf cr}(\beta)
\end{equation}
When $\frac{\beta}{L(P_{\max}+\eta\cdot c_{g}-c_{o})}\to\infty$,
we have
\begin{equation}
{\sf cr}(\beta)\to\min\Big\{\frac{L\cdot P_{\max}}{L\cdot c_{o}+c_{m}},3-\frac{2(L\cdot c_{o}+c_{m})}{L\cdot\big(P_{\max}+c_{g}\cdot\eta\big)}\Big\}
\end{equation}
\end{thm}

\begin{proof} The basic idea is as follows. Given any deterministic
online algorithm ${\cal A}$, we construct a special input sequence
$\sigma$, such that
\begin{equation}
\frac{{\rm Cost}(y_{{\cal A}};\sigma)}{{\rm Cost}(y_{{\rm OFA}};\sigma)}\ge{\sf cr}(\beta)
\end{equation}
for a function ${\sf cr}(\beta)$.

First, we note that at time $t$, ${\cal A}$ determines $y_{{\cal A}}(t)$
only based on the past input in $[1,t-1]$. Thus, we construct an
input sequence $\sigma=(a(t),h(t),p(t))_{t=1}^{T}$ progressively
as follows:
\begin{itemize}
\item $p(t)=P_{\max}$.
\item $a(t)=L\cdot(1-y_{{\cal A}}(t-1))$. Namely,
\begin{equation}
a(t)=\begin{cases}
\label{eq:a(t).construction}L, & \mbox{if\ }y_{{\cal A}}(t-1)=0\\
0, & \mbox{if\ }y_{{\cal A}}(t-1)=1
\end{cases}
\end{equation}

\item $h(t)=\eta\cdot a(t)$.
\end{itemize}
For completeness, we set the boundary conditions: $p(0)=p(T+1)=0$
and $a(0)=a(T+1)=0$.

\textbf{Step 1}: Computing ${\rm Cost}(y_{{\cal A}};\sigma)$:

Based on our construction of $(\sigma(t))_{t=1}^{T}$, we can partition
$[1,T]$ into disjoint segments of consecutive intervals of full demand
or zero demand:
\begin{itemize}
\item Full-demand segment: $[t_{1},t_{2}]$, if $a(t)=L$ for all $t\in[t_{1},t_{2}]$,
and $a(t_{1}-1)=a(t_{2}+1)=0$.
\item Zero-demand segment: $[t_{1},t_{2}]$, if $a(t)=0$ for all $t\in[t_{1},t_{2}]$,
and $a(t_{1}-1)=a(t_{2}+1)=L$.
\end{itemize}

Note that according to Eqn.~(\ref{eq:a(t).construction}), $a(1)=L$.
Thus, the time $t=1$ must belong to a full-demand segment. Also,
full-demand and zero-demand segments appear alternating.

Let $n_{f}$ and $n_{z}$ be the number of full-demand and zero-demand
segments in $[1,T]$ respectively. Thus, we have
\begin{equation}
n_{z}\le n_{f}\le n_{z}+1\label{eq:the.difference.between.tr.ph}
\end{equation}

In a full-demand segment $[t_{1},t_{2}]$, since $a(t)=L$, for all
$t\in[t_{1},t_{2}]$ and according to the construction of $a(t)$
in Eqn.~(\ref{eq:a(t).construction}), we conclude that $y(t)$ generated
by ${\cal A}$ must be
\begin{equation}
y_{{\cal A}}(t)=\begin{cases}
0, & t_{1}\le t\le t_{2}-1\\
1, & t=t_{2}
\end{cases}
\end{equation}
As a result, in a full-demand segment $[t_{1},t_{2}]$ with length
$(t_{2}-t_{1}+1)$, ${\cal A}$ incurs a cost
\begin{equation}
L\cdot\big(P_{\max}+\eta\cdot c_{g}\big)\cdot\big(t_{2}-t_{1}\big)+\beta+(L\cdot c_{o}+c_{m})\cdot1.
\end{equation}

Similarly, in a zero-demand segment $[t_{1},t_{2}]$, ${\cal A}$
incurs a cost
\begin{equation}
\big(t_{2}-t_{1}\big)c_{m}.
\end{equation}


Let $\Sigma_{f}$ and $\Sigma_{z}$ be the total lengths of full-demand
and zero-demand segments in $[1,T]$ respectively. By summing the
costs over all full-demand and zero-demand segments and simplifying
terms, we obtain a compact expression of the cost of ${\cal A}$ w.r.t.
$\sigma$ as follows:
\begin{eqnarray}
 &  & {\rm Cost}(y_{{\cal A}};\sigma)\notag\label{eq:costofA}\\
 & = & L\cdot\big(P_{\max}+\eta\cdot c_{g}\big)\cdot(\Sigma_{f}-n_{f})+n_{f}\big(\beta+L\cdot c_{o}+c_{m}\big)\notag\\
 &  & +c_{m}\cdot(\Sigma_{z}-n_{z})\notag\\
 & = & L\cdot\big(P_{\max}+\eta\cdot c_{g}\big)\cdot\Sigma_{f}+n_{f}\Big(\beta+L\cdot c_{o}+c_{m}\notag\\
 &  & -L\cdot\big(P_{\max}+\eta\cdot c_{g}\big)\Big)+c_{m}\Sigma_{z}-c_{m}n_{z}
\end{eqnarray}

\textbf{Step 2}: (Bounding ${\rm Cost}(y_{{\cal A}};\sigma)/{\rm Cost}(y_{{\rm OFA}};\sigma)$):

We divide the input $\sigma$ into critical segments. We then define
${\cal S}_{up}$ be the set of all type-0, type-2, type-3, and the
``increasing'' parts of type-1 critical segments, and ${\cal S}_{pt}$
be set of the ``plateau'' parts of type-1 critical segments.

Here, for a type-1 critical segment $[T_{i}^{c}+1,T_{i+1}^{c}]$,
the ``increasing'' part is defined as $[T_{i}^{c}+1,\tilde{T}_{i}^{c}]$
and the ``plateau'' part is defined as $[\tilde{T}_{i}^{c}+1,T_{i+1}^{c}]$.
We define ${\rm Cost}_{up}(\cdot)$ and ${\rm Cost}_{pt}(\cdot)$
be the costs for ${\cal S}_{up}$ and ${\cal S}_{pt}$ respectively.

On the increasing part $[T_{i}^{c}+1,\tilde{T}_{i}^{c}]$, the deficit
function wriggles up from $-\beta$ to $0$, and it cost the same
to served the part by either buying power from the grid or using on-site
generator (which incurs a turning-on cost). Hence, we can simplify
the offline cost on the increasing part as
\[
{\rm Cost}_{up}(y_{{\rm OFA}};\sigma)=\beta+\sum_{t=T_{i}^{c}+1}^{\tilde{T_{i}^{c}}}\psi\big(\sigma(t),1\big)=\sum_{t=T_{i}^{c}+1}^{\tilde{T_{i}^{c}}}\psi\big(\sigma(t),0\big).
\]

With this simplification, we proceed with the ratio analysis as follows:
\begin{eqnarray*}
\frac{{\rm Cost}(y_{{\cal A}};\sigma)}{{\rm Cost}(y_{{\rm OFA}};\sigma)} & = & \frac{{\rm Cost}_{up}(y_{{\cal A}};\sigma)+{\rm Cost}_{pt}(y_{{\cal A}};\sigma)}{{\rm Cost}_{up}(y_{{\rm OFA}};\sigma)+{\rm Cost}_{pt}(y_{{\rm OFA}};\sigma)}\\
 & \ge & \min\big\{\begin{array}{cc}
\frac{{\rm Cost}_{up}(y_{{\cal A}};\sigma)}{{\rm Cost}_{up}(y_{{\rm OFA}};\sigma)}, & \frac{{\rm Cost}_{pt}(y_{{\cal A}};\sigma)}{{\rm Cost}_{pt}(y_{{\rm OFA}};\sigma)}\end{array}\big\}
\end{eqnarray*}
As $T$ goes to infinity, it is clear that to lower-bound the above
ratio, it suffices to consider only those ${\cal S}_{i}$ ($i\in\{up,pt\}$)
with unbounded length in time. Next, we study each term in the lower
bound of the competitive ratio. We define $\Sigma_{f}^{up}\;(\Sigma_{z}^{up})$
and $\Sigma_{f}^{pt}\;(\Sigma_{z}^{pt})$ as the total length of full-demand
(zero-demand) intervals in the increasing parts and plateau, respectively.
Similarly, we define $n_{f}^{up}\;(n_{z}^{up})$ and $n_{f}^{pt}\;(n_{z}^{pt})$
as the number of full-demand (zero-demand) intervals in the increasing
parts and plateau, respectively.

\textbf{Step 2-1}: (Bounding ${\rm Cost}_{up}(y_{{\cal A}};\sigma)/{\rm Cost}_{up}(y_{{\rm OFA}};\sigma)$)

First, we seek to lower-bound the term $\frac{{\rm Cost}_{up}(y_{{\cal A}};\sigma)}{{\rm Cost}_{up}(y_{{\rm OFA}};\sigma)}$
under the assumption that $|{\cal S}_{up}|$ is unbounded. From the
offline solution structure, we know that on type-0, type-2, type-3,
and the ``increasing'' parts of type-1 critical segments, the offline
optimal cost is given by
\begin{equation}
{\rm Cost}_{up}(y_{{\rm OFA}};\sigma)=L\cdot\big(P_{\max}+\eta\cdot c_{g}\big)\cdot\Sigma_{f}^{up}
\end{equation}
Noticing that we also have $n_{f}^{up}\ge n_{z}^{up}$, we obtain
\begin{eqnarray*}
 &  & \frac{{\rm Cost}_{up}(y_{{\cal A}};\sigma)}{{\rm Cost}_{up}(y_{{\rm OFA}};\sigma)}\\
 & = & \Big(L\cdot\big(P_{\max}+\eta\cdot c_{g}\big)\cdot\Sigma_{f}^{up}+\big(\beta+L\cdot c_{o}\\
 &  & \quad-L\big(P_{\max}+\eta\cdot c_{g}\big)+c_{m}\big)n_{f}^{up}+c_{m}\cdot\Sigma_{z}^{up}-c_{m}n_{z}^{up}\Big)\Big/\\
 &  & \Big(L\cdot\big(P_{\max}+\eta\cdot c_{g}\big)\cdot\Sigma_{f}^{up}\Big)\\
 & \ge & \Big(L\cdot\big(P_{\max}+\eta\cdot c_{g}\big)\cdot\Sigma_{f}^{up}+\big(\beta+L\cdot c_{o}\\
 &  & -L\big(P_{\max}+\eta\cdot c_{g}\big)\big)n_{f}^{up}+c_{m}\cdot\Sigma_{z}^{up}\Big)\Big/\\
 &  & \Big(L\cdot\big(P_{\max}+\eta\cdot c_{g}\big)\cdot\Sigma_{f}^{up}\Big)
\end{eqnarray*}

In ${\cal S}_{up}$, either there is only one type-0 segment, or there
are equal number of type-2/3 critical segments and type-1 critical
segment ``increasing'' parts. Hence, the total deficit function
increment, i.e., $(LP_{{\max}}-L\cdot c_{o}-c_{m})\Sigma_{f}^{up}$,
must be no more than the total deficit function decrement, which is
upper bounded by $c_{m}\Sigma_{z}^{up}+\beta$ where the term $\beta$
accounts for that the deficit function does not end naturally at but
get dragged down to$-\beta$ at the end of $T$. That is,
\begin{equation}
c_{m}\Sigma_{z}^{up}+\beta-(L\big(P_{\max}+\eta\cdot c_{g}-c_{o}\big)-c_{m})\Sigma_{f}^{up}\ge0.\label{eq:type-start-relation-increment-decrement}
\end{equation}
Moreover, since ${\cal S}_{up}$ contains only type-0, type-2, type-3,
and the ``increasing'' parts of type-1 critical segments, the deficit
function increment introduced by every full-demand segment must be
no more than $\beta$. That is,

\begin{equation}
n_{f}^{up}\beta\ge(L\big(P_{\max}+\eta\cdot c_{g}-c_{o}\big)-c_{m})\Sigma_{f}^{up}.\label{eq:type-start-relation-n_z*beta}
\end{equation}

By the above inequalities, we continue the derivation as follows:
\begin{eqnarray}
\frac{{\rm Cost}_{up}(y_{{\cal A}};\sigma)}{{\rm Cost}_{up}(y_{{\rm OFA}};\sigma)} & \ge & 1+\frac{\beta+L\cdot\left(c_{o}-P_{\max}-\eta\cdot c_{g}\right)}{L\cdot\big(P_{\max}+\eta\cdot c_{g}\big)\Sigma_{f}^{up}/n_{f}^{up}}\notag\\
 &  & \quad+\frac{c_{m}\Sigma_{z}^{up}}{L\cdot\big(P_{\max}+\eta\cdot c_{g}\big)\cdot\Sigma_{f}^{up}}\label{eq:type-0.interval.ratio.lower.bound}
\end{eqnarray}

Now we discuss the second term in Eqn. (\ref{eq:type-0.interval.ratio.lower.bound}),
denoted as $(I)$. Recall in the problem setting, $\beta+L\cdot c_{o}-L\cdot\big(P_{\max}+\eta\cdot c_{g}\big)>0$.
Hence, term $(I)$ is monotonically decreasing in $\Sigma_{f}^{up}/n_{f}^{up}$,
and its minimum value is taken when $\Sigma_{f}^{up}/n_{f}^{up}$
is replaced with the upper-boundary value $\beta(LP_{{\max}}-L\cdot c_{o}-c_{m})^{-1}$:
\begin{eqnarray}
 &  & (I)\nonumber \\
 & \ge & \left(\beta+L\cdot\left(c_{o}-P_{\max}-\eta\cdot c_{g}\right)\right)/\left(L\cdot\big(P_{\max}+\eta\cdot c_{g}\big)\right.\nonumber \\
 &  & \left.\cdot\beta\cdot(L\cdot\big(P_{\max}+\eta\cdot c_{g}-c_{o}\big)-c_{m})^{-1}\right)\nonumber \\
 & = & \frac{L\cdot\big(P_{\max}+\eta\cdot c_{g}-c_{o}\big)-c_{m}}{L\cdot\big(P_{\max}+\eta\cdot c_{g}\big)}\label{eq:type-0.interval.term(I).lower.bound}
\end{eqnarray}
The above inequality holds for arbitrary $|{\cal S}_{up}|$.

Now we discuss the third term in Eqn. (\ref{eq:type-0.interval.ratio.lower.bound}),
denoted as $(II)$. When $|{\cal S}_{up}|$ goes to infinity, $(II)$
can be discussed by two cases. In the first case, $\Sigma_{f}^{up}$
remains bounded when $|{\cal S}_{up}|$ goes to infinity. Since $\Sigma_{f}^{up}+\Sigma_{z}^{up}=|{\cal S}_{up}|$
we must have unbounded $\Sigma_{z}^{up}$ as $|{\cal S}_{up}|$ goes
to infinity. As a result, the term $(II)$ is unbounded, and so is
$\frac{{\rm Cost}_{up}(y_{{\cal A}};\sigma)}{{\rm Cost}_{up}(y_{{\rm OFA}};\sigma)}$.
In the second case, $\Sigma_{f}^{up}$ goes unbounded when $|{\cal S}_{up}|$
goes to infinity. Then by Eqn.~(\ref{eq:type-start-relation-increment-decrement}),
we know that $\Sigma_{z}^{up}$ also goes unbounded and,
\begin{eqnarray}
(II) & \ge & \frac{c_{m}\cdot\Sigma_{z}^{up}}{L\cdot\big(P_{\max}+\eta\cdot c_{g}\big)}\cdot\frac{L\cdot\big(P_{\max}+\eta\cdot c_{g}-c_{o}\big)-c_{m}}{c_{m}\cdot\Sigma_{z}^{up}+\beta}\notag\\
 & = & \frac{L\cdot\big(P_{\max}+\eta\cdot c_{g}-c_{o}\big)-c_{m}}{L\cdot\big(P_{\max}+\eta\cdot c_{g}\big)}\cdot\frac{c_{m}\cdot\Sigma_{z}^{up}}{c_{m}\cdot\Sigma_{z}^{up}+\beta}\notag\\
 & = & \frac{L\cdot\big(P_{\max}+\eta\cdot c_{g}-c_{o}\big)-c_{m}}{L\cdot\big(P_{\max}+\eta\cdot c_{g}\big)}(\mbox{as }|{\cal S}_{up}|\rightarrow\infty)\label{eq:type-0.interval.term(II).lower.bound}
\end{eqnarray}
Overall, by substituting Eqn.~(\ref{eq:type-0.interval.term(I).lower.bound})
and Eqn.~(\ref{eq:type-0.interval.term(II).lower.bound}) into Eqn.~(\ref{eq:type-0.interval.ratio.lower.bound}),
we obtain
\begin{eqnarray}
\frac{{\rm Cost}_{up}(y_{{\cal A}};\sigma)}{{\rm Cost}_{up}(y_{{\rm OFA}};\sigma)} & \ge & 3-2\frac{L\cdot c_{o}+c_{m}}{L\cdot\big(P_{\max}+\eta\cdot c_{g}\big)}.
\end{eqnarray}

\textbf{Step 2-2}: (Bounding ${\rm Cost}_{pt}(y_{{\cal A}};\sigma)/{\rm Cost}_{pt}(y_{{\rm OFA}};\sigma)$)

We now lower-bound the term $\frac{{\rm Cost}_{pt}(y_{{\cal A}};\sigma)}{{\rm Cost}_{pt}(y_{{\rm OFA}};\sigma)}$
under the assumption that $|{\cal S}_{pt}|$ is unbounded. Since ${\cal S}_{pt}$
only contains the ``plateau'' parts of type-1 critical segments,
we have

\begin{equation}
{\rm Cost}_{pt}(y_{{\rm OFA}};\sigma)=(L\cdot c_{o}+c_{m})\Sigma_{f}^{pt}+c_{m}\Sigma_{z}^{pt}
\end{equation}
Therefore,
\begin{eqnarray*}
 &  & \frac{{\rm Cost}_{pt}(y_{{\cal A}};\sigma)}{{\rm Cost}_{pt}(y_{{\rm OFA}};\sigma)}\\
 & = & \Big(L\cdot\big(P_{\max}+\eta\cdot c_{g}\big)\cdot\Sigma_{f}^{pt}+\big(\beta+L\cdot c_{o}\\
 &  & -L\big(P_{\max}+\eta\cdot c_{g}\big)+c_{m}\big)n_{f}^{pt}+c_{m}\cdot\Sigma_{z}^{pt}-c_{m}n_{z}^{pt}\Big)\Big/\\
 &  & \Big((L\cdot c_{o}+c_{m})\Sigma_{f}^{pt}+c_{m}\Sigma_{z}^{pt}\Big)
\end{eqnarray*}

On ${\cal S}_{pt}$, the total deficit function increment,
i.e., $(L\big(P_{\max}+\eta \cdot c_{g}\big)-L \cdot c_{o}-c_{m})\Sigma_{f}^{pt}$,
must be no less than the total deficit function decrement, which is
$c_{m}\Sigma_{z}^{pt}$. That is,

\begin{equation}
\big(L\cdot\big(P_{\max}+\eta\cdot c_{g}-c_{o}\big)-c_{m}\big)\Sigma_{f}^{pt}\ge c_{m}\Sigma_{z}^{pt}.\label{eq:type-1-relation-increment-decrement}
\end{equation}

Moreover, on ${\cal S}_{pt}$, the deficit function decrement caused
by every zero-demand segment must be less than $\beta$; otherwise,
the deficit function will reach value $-\beta$ and it cannot be a
``plateau'' part of a type-1 critical segment. That is,
\begin{equation}
n_{z}^{pt}\beta\ge c_{m}\Sigma_{z}^{pt}\label{eq:type-1-relation n star}
\end{equation}
Since a plateau part of a type-1 critical segment must end with a
full-demand interval, we must have $n_{f}^{pt}\ge n_{z}^{pt}$.

We continue the lower bound analysis as follows:

\begin{eqnarray*}
 &  & \frac{{\rm Cost}_{pt}(y_{{\cal A}};\sigma)}{{\rm Cost}_{pt}(y_{{\rm OFA}};\sigma)}\\
 & \ge & \frac{L\cdot\big(P_{\max}+\eta\cdot c_{g}\big)\cdot\Sigma_{f}^{pt}+c_{m}\cdot\Sigma_{z}^{pt}}{(L\cdot c_{o}+c_{m})\Sigma_{f}^{pt}+c_{m}\Sigma_{z}^{pt}}\\
 &  & +\frac{\big(\beta+L\cdot c_{o}-L\big(P_{\max}+\eta\cdot c_{g}\big)\big)n_{f}^{pt}}{(L\cdot c_{o}+c_{m})\Sigma_{f}^{pt}+c_{m}\Sigma_{z}^{pt}}\\
 & \ge & \frac{L\cdot\big(P_{\max}+\eta\cdot c_{g}\big)+\big(1+\frac{\beta+L\cdot\left(c_{o}-P_{\max}-\eta\cdot c_{g}\right)}{\beta}\big)c_{m}\cdot\Sigma_{z}^{pt}/\Sigma_{f}^{pt}}{L\cdot c_{o}+c_{m}+c_{m}\Sigma_{z}^{pt}/\Sigma_{f}^{pt}}.
\end{eqnarray*}

By checking the derivative, we know that the last term is monotonically
increasing/decreasing in the ratio $c_{m}\cdot\Sigma_{z}^{pt}/\Sigma_{f}^{pt}$.
Hence, its minimum value is taken when the ratio is replaced with
the lower-boundary value $0$ or the upper-boundary value $L\big(P_{\max}+\eta\cdot c_{g}-c_{o}\big)-c_{m}$.
Carrying out the derivation and taking into account the problem setting
$\frac{\beta}{L(P_{\max}+\eta\cdot c_{g}-c_{o})}\to\infty$, we obtain

\begin{eqnarray*}
 &  & \frac{{\rm Cost}_{pt}(y_{{\cal A}};\sigma)}{{\rm Cost}_{pt}(y_{{\rm OFA}};\sigma)}\\
 & \geq & \min\Big\{\frac{L\big(P_{\max}+\eta\cdot c_{g}\big)}{L\cdot c_{o}+c_{m}},3-2\frac{L\cdot c_{o}+c_{m}}{L\big(P_{\max}+\eta\cdot c_{g}\big)}\Big\}.
\end{eqnarray*}

At the end, we obtain the desired result of
\begin{equation}
\frac{{\rm Cost}(y_{{\cal A}};\sigma)}{{\rm Cost}(y_{{\rm OFA}};\sigma)}\ge\min\big\{\begin{array}{cc}
\frac{L\big(P_{\max}+\eta\cdot c_{g}\big)}{L\cdot c_{o}+c_{m}}, & 3-2\frac{L\cdot c_{o}+c_{m}}{L\big(P_{\max}+\eta\cdot c_{g}\big)}\end{array}\big\}.
\end{equation}

\end{proof}

\section{Proof of Theorem 4} \label{subsec:CHASElk-competitive-ratio}

\setcounter{thm}{3}
\begin{thm}
The competitive ratio of ${\sf CHASE}_s^{lk(\omega)}$
\begin{eqnarray}
 &  & {\sf CR}({\sf CHASE}_s^{{\rm lk}(w)})\\
 & \leq & 1+\frac{2\beta L\big(P_{\max}+\eta\cdot c_{g}-c_{o}-\frac{c_{m}}{L}\big)}{\big(P_{\max}+\eta\cdot c_{g}\big)\big(\beta L+w \cdot c_{m}\big(L-\frac{c_{m}}{P_{\max}+\eta\cdot c_{g}-c_{o}}\big)\big)} \notag
\end{eqnarray}
\end{thm}

\begin{proof}
We note that the proof is similar that of Theorem~\ref{thm:CHASE-competitive-ratio}, except with modifications considering time window $w$. We denote the outcome of ${\sf CHASE}_s^{{\rm lk}(w)}$ by $\big(y_{{\sf CHASE}(w)}(t)\big)_{t=1}^{T}.$

(\textbf{type-0}): Similar to the proof of Theorem~\ref{thm:CHASE-competitive-ratio}

(\textbf{type-1}): Based on the definition of critical segment (Definition~\ref{def:critical-seg}),
we recall that there is an auxiliary point $\tilde{T_{i}^{c}}$, such
that either \big($\Delta(T_{i}^{c})=0$ and $\Delta(\tilde{T_{i}^{c}})=-\beta$\big)
or \big($\Delta(T_{i}^{c})=-\beta$ and $\Delta(\tilde{T_{i}^{c}})=0$\big).
We focus on the segment $T_{i}^{c}+1+w<\tilde{T}_{i}^{c}$. We observe
\begin{equation}
y_{{\sf CHASE}(w)}(t)=\begin{cases}
0, & \mbox{for all\ } t\in[T_{i}^{c}+1,\tilde{T}_{i}^{c}-w),\\
1, & \mbox{for all\ } t\in[\tilde{T}_{i}^{c}-w,T_{i+1}^{c}],
\end{cases}
\end{equation}
We consider a particular type-1 critical segment, i.e., $k$th type-1
critical segment: $[T_{i}^{c}+1,T_{i+1}^{c}]$. Note that by the definition
of type-1, $y_{{\rm OFA}}(T_{i}^{c})=y_{{\sf CHASE}(w)}(T_{i}^{c})=0$.
$y_{{\rm OFA}}(t)$ switches from $0$ to $1$ at time $T_{i}^{c}+1$,
while $y_{{\sf CHASE}(w)}$ switches at time $\tilde{T}_{i}^{c}-w-1$,
both incurring startup cost $\beta$. The cost
difference between $y_{{\rm CHASE}}$ and $y_{{\rm OFA}}$ within $[T_{i}^{c}+1,T_{i+1}^{c}]$
is
\begin{eqnarray}
 &  & \sum_{t=T_{i}^{c}+1}^{\tilde{T}_{i}^{c}-1}\Big(\psi\Big(\sigma(t),0\Big)-\psi\Big(\sigma(t),1\Big)\Big)+\beta-\beta\\
 & = & \sum_{t=T_{i}^{c}+1}^{\tilde{T}_{i}^{c}-w-1}\delta(t)=\Delta(\tilde{T}_{i}^{c}-w-1)-\Delta(T_{i}^{c})= q_{i}^{1}+\beta \notag
\end{eqnarray}
where $q_{i}^{1} \triangleq \Delta(\tilde{T}_{i}^{c}-w-1)$.

Recall the number of type-$j$ critical segments $m_{j}\triangleq|{\cal T}_{j}|$.
\begin{equation}
{\rm Cost}^{{\rm ty\mbox{-}}1}(y_{{\sf CHASE}(w)})\le{\rm Cost}^{{\rm ty\mbox{-}}1}(y_{{\rm OFA}})+m_{1}\cdot\beta+\sum_{i=1}^{m_{1}}q_{i}^{1}
\end{equation}

(\textbf{type-2}) and (\textbf{type-3}):
We derive similarly for $j=2$ or 3 as
\begin{equation}
{\rm Cost}^{{\rm ty\mbox{-}}j}(y_{{\sf CHASE}(w)})\le{\rm Cost}^{{\rm ty\mbox{-}}j}(y_{{\rm OFA}})-\sum_{i=1}^{m_{j}}q_{i}^{j}.
\end{equation}
Note that $|q_{i}^{j}|\leq\beta$ for all $i,j$.

Furthermore, we note $m_{1}=m_{2}+m_{3}$. Overall, we obtain
\[
\begin{array}{@{}r@{}l@{\ }l}
 &  & {\displaystyle \frac{{\rm Cost}(y_{{\sf CHASE}(w)})}{{\rm Cost}(y_{{\rm OFA}})}}= {\displaystyle \frac{\sum_{j=0}^{3}{\rm Cost}^{{\rm ty\mbox{-}}j}(y_{{\sf CHASE}(w)})}{\sum_{j=0}^{3}{\rm Cost}^{{\rm ty\mbox{-}}j}(y_{{\rm OFA}})}}\\
 & \leq & {\displaystyle \frac{m_{1}\beta+\sum_{k=1}^{m_{1}}q_{i}^{1}+(m_{2}+m_{3})\beta+\sum_{j=0}^{3}{\rm Cost}^{{\rm ty\mbox{-}}j}(y_{{\rm OFA}})}{\sum_{j=0}^{3}{\rm Cost}^{{\rm ty\mbox{-}}j}(y_{{\rm OFA}})}}\\
 & = & {\displaystyle 1+\frac{2m_{1}\beta+\sum_{k=1}^{m_{1}}q_{i}^{1}}{\sum_{j=0}^{3}{\rm Cost}^{{\rm ty\mbox{-}}j}(y_{{\rm OFA}})}}\\
 & \leq & 1+\begin{cases}
0 & \mbox{if\ }m_{1}=0,\\
{\displaystyle \frac{2m_{1}\beta+\sum_{k=1}^{m_{1}}q_{i}^{1}}{{\rm Cost}^{{\rm ty\mbox{-}}1}(y_{{\rm OFA}})}} & \mbox{otherwise}
\end{cases}
\end{array}
\]

By Lemma \ref{lem:w lower bound type-1} and simplifications, we obtain
\begin{eqnarray}
 &  & \frac{\mbox{Cost}(y_{{\sf CHASE}(w)})}{\mbox{Cost}(y_{\mathrm{OFA}})}\\
 & \leq & 1+\frac{2\beta L\big(P_{\max}+\eta\cdot c_{g}-c_{o}-\frac{c_{m}}{L}\big)}{\big(P_{\max}+\eta\cdot c_{g}\big)\big(\beta L+w \cdot c_{m}\big(L-\frac{c_{m}}{P_{\max}+\eta\cdot c_{g}-c_{o}}\big)\big)} \notag
\end{eqnarray}
\end{proof}

\begin{lem}
\label{lem:w lower bound type-1}
\begin{eqnarray}
{\rm Cost}^{{\rm ty\mbox{-}}1}(y_{{\rm OFA}}) & \geq & m_{1}\beta+\sum_{k=1}^{m_{1}}\Big(\frac{(q_{i}^{1}+\beta)(Lc_{o}+c_{m})}{L\big(P_{\max}+\eta\cdot c_{g}-c_{o}\big)-c_{m}}\notag \\
 &  & \quad w \cdot c_{m}+\frac{c_{o}(-q_{i}^{1}+w \cdot c_{m})}{P_{\max}+\eta\cdot c_{g}-c_{o}}\Big)
\end{eqnarray}
\end{lem}

\begin{proof}
Consider a particular type-1 segment $[T_{i}^{c}+1,T_{i+1}^{c}]$.
Denote the costs of $\mathrm{y_{OFA}}$ during $[T_{i}^{c}+1,\tilde{T}_{i}^{c}-w-1]$
and $[\tilde{T}_{i}^{c}-w,T_{i+1}^{c}]$ by $\mathrm{Cost^{\rm up}}$
and $\mathrm{Cost^{\rm pt}}$ respectively.

{\bf Step 1:} We bound $\mathrm{Cost^{\rm up}}$ as follows:
\begin{eqnarray}
\mathrm{Cost^{\rm up}} & = & \beta+\sum_{t=T_{i}^{c}+1}^{\tilde{T}_{i}^{c}-w-1}\psi\big(\sigma(t),1\big)\\
 & = & \beta+(\tilde{T}_{i}^{c}-w-1-T_{i}^{c})c_{m}\label{eq:type-1 cost eq1-1} +\sum_{t=T_{i}^{c}+1}^{\tilde{T}_{i}^{c}-w-1}\big(\psi\big(\sigma(t),1\big)-c_{m}\big).\notag
\end{eqnarray}

On the other hand, we obtain
\begin{eqnarray}
& & \sum_{t=T_{i}^{c}+1}^{\tilde{T}_{i}^{c}-w-1}\big(\psi\big(\sigma(t),1\big)-c_{m}\big)\notag \\
 & = & \frac{\sum_{t=T_{i}^{c}+1}^{\tilde{T}_{i}^{c}-w-1}\big(\psi\big(\sigma(t),1\big)-c_{m}\big)}{\sum_{t=T_{i}^{c}+1}^{\tilde{T}_{i}^{c}-w-1}\big(\psi\big(\sigma(t),0\big)-\psi\big(\sigma(t),1\big)+c_{m}\big)}\notag \\
 &  & \times\sum_{t=T_{i}^{c}+1}^{\tilde{T}_{i}^{c}-w-1}\big(\psi\big(\sigma(t),0\big)-\psi\big(\sigma(t),1\big)+c_{m}\big)\\
 & \geq & \frac{c_{o}}{P_{\max}+\eta\cdot c_{g}-c_{o}}\label{eq:w type-1 cost eq2}\\
 &  & \times\sum_{t=T_{i}^{c}+1}^{\tilde{T}_{i}^{c}-w-1}\big(\psi\big(\sigma(t),0\big)-\psi\big(\sigma(t),1\big)+c_{m}\big)
\end{eqnarray}
The last inequality follows from Lemma~\ref{lem:segment-1 minimum cost}.

Next, we bound the second term by
\begin{eqnarray}
 &  & \sum_{t=T_{i}^{c}+1}^{\tilde{T}_{i}^{c}-w-1}\big(\psi\big(\sigma(t),0\big)-\psi\big(\sigma(t),1\big)+c_{m}\big) \\
 & \geq &  \sum_{t=T_{i}^{c}+1}^{\tilde{T}_{i}^{c}-w-1}\big(\delta(t)+c_{m}\big)  \\
 & \geq & \triangle\big(\tilde{T}_{i}^{c}-w-1\big)-\triangle\big(T_{i}^{c}\big)+(\tilde{T}_{i}^{c}-w-1-T_{i}^{c})c_{m} \notag \\
 & = & q_{i}^{1}+\beta+(\tilde{T}_{i}^{c}-w-1-T_{i}^{c})c_{m}
\end{eqnarray}

Together, we obtain
\begin{eqnarray}
\mathrm{} &  & \mathrm{Cost^{\rm up}}\notag \\
 & \geq & \beta+(\tilde{T}_{i}^{c}-w-1-T_{i}^{c})c_{m}+\\
 &  & \frac{c_{o}}{P_{\max}+\eta\cdot c_{g}-c_{o}}\Big(q_{i}^{1}+\beta+(\tilde{T}_{i}^{c}-w-1-T_{i}^{c})c_{m}\Big) \\
 & = & \beta+\frac{(q_{i}^{1}+\beta)c_{o}+(\tilde{T}_{i}^{c}-w-1-T_{i}^{c})\big(P_{\max}+\eta\cdot c_{g}\big)c_{m}}{P_{\max}+\eta\cdot c_{g}-c_{o}} \notag \label{eq:w type-1 cost eq6}
\end{eqnarray}

Furthermore, we note that $\big(\tilde{T}_{i}^{c}-w-1-T_{i}^{c}\big)$
is lower bounded by the steepest descend when $p(t)=P_{\max}$, $a(t)=L$
and $h(t)=\eta L$,
\begin{equation}
\tilde{T}_{i}^{c}-w-1-T_{i}^{c}\geq\frac{q_{i}^{1}+\beta}{L\big(P_{\max}+\eta\cdot c_{g}-c_{o}\big)-c_{m}}
\label{eq:w type-1 cost eq7}
\end{equation}

By Eqns.~(\ref{eq:w type-1 cost eq6})-(\ref{eq:w type-1 cost eq7}), we obtain
\begin{eqnarray}
 &  & \mathrm{Cost^{\rm up}}\notag \\
 & \geq & \beta+\frac{(q_{i}^{1}+\beta)c_{o}+(\tilde{T}_{i}^{c}-w-1-T_{i}^{c})\big(P_{\max}+\eta\cdot c_{g}\big)c_{m}}{P_{\max}+\eta\cdot c_{g}-c_{o}} \notag \\
 & \geq & \beta+\frac{(q_{i}^{1}+\beta)(Lc_{o}+c_{m})}{L\big(P_{\max}+\eta\cdot c_{g}-c_{o}\big)-c_{m}} \label{eq:w type-1 cost eq8}
\end{eqnarray}

{\bf Step 2}: We bound $\mathrm{Cost^{\rm pt}}$ as follows.
\begin{eqnarray}
 &  & \mathrm{Cost^{\rm pt}} =  \sum_{t=\tilde{T}_{i}^{c}-w}^{T_{i+1}^{c}}\psi\big(\sigma(t),1\big)\\
 & = & (T_{i+1}^{c}-\tilde{T}_{i}^{c}+w+1)c_{m}+\sum_{t=\tilde{T}_{i}^{c}-w}^{T_{i+1}^{c}}\big(\psi\big(\sigma(t),1\big)-c_{m}\big)\\
 & \geq & w \cdot c_{m}+\\
 &  & \frac{c_{o}}{P_{\max}+\eta\cdot c_{g}-c_{o}}\sum_{t=\tilde{T}_{i}^{c}-w}^{T_{i+1}^{c}}\big(\psi\big(\sigma(t),0\big)-\psi\big(\sigma(t),1\big)+c_{m}\big) \notag
\end{eqnarray}
On the other hand, we obtain
\begin{eqnarray}
 &  & \sum_{t=\tilde{T}_{i}^{c}-w}^{T_{i+1}^{c}}\big(\psi\big(\sigma(t),0\big)-\psi\big(\sigma(t),1\big)+c_{m}\big)\\
 & = & \sum_{t=\tilde{T}_{i}^{c}-w}^{T_{i+1}^{c}}\delta(t)+(T_{i+1}^{c}-\tilde{T}_{i}^{c}+w+1)c_{m}\\
 & \geq & \triangle(T_{i+1}^{c})-\triangle(\tilde{T}_{i}^{c}-w-1)+w \cdot c_{m}= w \cdot c_{m}-q_{i}^{1} \notag
\end{eqnarray}

Therefore,
\begin{eqnarray}
\mathrm{Cost^{\rm pt}} & \geq & w \cdot c_{m}+\frac{c_{o}(w \cdot c_{m}-q_{i}^{1})}{P_{\max}+\eta\cdot c_{g}-c_{o}}\label{eq:w type-1 cost roof-1}
\end{eqnarray}

Since there are $m_{1}$ type-1 critical segments, according to Eqna.~(\ref{eq:w type-1 cost eq8})-(\ref{eq:w type-1 cost roof-1}), we obtain
\begin{eqnarray}
 &  & {\rm Cost}^{{\rm ty\mbox{-}}1}(y_{{\rm OFA}})\\
 & \geq & m_{1}\beta+\sum_{k=1}^{m_{1}}\Big(\frac{(q_{i}^{1}+\beta)(Lc_{o}+c_{m})}{L\big(P_{\max}+\eta\cdot c_{g}-c_{o}\big)-c_{m}} \notag\\
 &  & \qquad w \cdot c_{m}+\frac{c_{o}(-q_{i}^{1}+w \cdot c_{m})}{P_{\max}+\eta\cdot c_{g}-c_{o}}\Big).
\end{eqnarray}
\end{proof}

\section{Proof of Theorem 5}\label{subsec:nOFA-optimal}
\setcounter{thm}{4}
\begin{thm} \label{thm:nOFA-optimal}
Suppose $(y_n, u_n, v_n, s_n)$ is an optimal offline solution for each
$\textbf{fMCMP}_{\rm s}^{{\rm ly\mbox{-}}n}$ ($1\leq n\leq N$). Then\\
$((y^\ast_{n}, u^\ast_{n})_{n=1}^{N}, v^\ast, s^\ast)$ defined as
follows is an optimal offline solution for $\textbf{fMCMP}$:
\begin{equation}
\hspace{-10pt}
\begin{array}{@{}r@{\ }r@{\ }l@{}}
y^\ast_{n}(t) & = & y_n(t), \ \ v^\ast(t) = a^{\rm top}(t) + \sum_{n=1}^{N} v_n(t)  \\
u^\ast_{n}(t) & = & u_n(t), \ \ s^\ast(t) =  h^{\rm top}(t) + \sum_{n=1}^{N} s_n(t)
\end{array}
\end{equation}
\end{thm}

\begin{proof}
First, since that $(a^{\rm top}, h^{\rm top})$ can only be satisfied from external supplies, we can assume $a^{\rm top}(t) = 0, h^{\rm top}(t) = 0$ for all $t$ without loss of generality.

We then present the basic idea of proof. Suppose that $((\tilde{y}_{n}, \tilde{u}_{n})_{n=1}^{N}, \tilde{v}, \tilde{s})$ is an optimal solution for $\textbf{fMCMP}$. We will show that we can construct a new feasible solution $((\widehat{y}_{n}, \widehat{u}_{n})_{n=1}^{N}, \widehat{v}, \widehat{s})$  for $\textbf{fMCMP}$, and a new feasible solution $(\widehat{y}_{n}, \widehat{u}_{n}, \widehat{v}_n, \widehat{s}_n)$ for each $\textbf{fMCMP}_{\rm s}^{{\rm ly\mbox{-}}n}$, such that
\begin{equation} \label{eqn:nOFA-cost}
{\rm Cost}(\tilde{y},\tilde{u},\tilde{v},\tilde{s}) \ge
{\rm Cost}(\widehat{y},\widehat{u},\widehat{v},\widehat{s}) \ge \sum_{n = 1}^{N}{\rm Cost}(\widehat{y}_n)
\end{equation}
where the second inequality follows from Lemma~\ref{lem:fMCMP}.

 $(y_n, u_n, v_n, s_n)$ is an optimal solution for each $\textbf{fMCMP}_{\rm s}^{{\rm ly\mbox{-}}n}$. Hence, ${\rm Cost}(\widehat{y}_n) \ge {\rm Cost}(y_n)$ for each $n$. Thus,
\begin{eqnarray}
{\rm Cost}(\tilde{y},\tilde{u},\tilde{v},\tilde{s}) \ge  \sum_{n = 1}^{N}{\rm Cost}(\widehat{y}_n)
& \ge &  \sum_{n = 1}^{N}{\rm Cost}(y^\ast_n)  \\
& = & {\rm Cost}(y^\ast,u^\ast,v^\ast,s^\ast) \notag
\end{eqnarray}
Hence, $((y^\ast_{n}, u^\ast_{n})_{n=1}^{N}, v^\ast, s^\ast)$ is an optimal solution for $\textbf{fMCMP}$.

It only remains to prove Eqn~(\ref{eqn:nOFA-cost}). Define $(\widehat{y}_{n}, \widehat{u}_{n})$ based on $(\tilde{y}_{n}, \tilde{u}_{n})$ by:
\begin{equation}
\widehat{y}_{n}(t) =
\begin{cases}
1,& \mbox{if\ } n \le \sum_{r=1}^{N} \tilde{y}_{r}(t) \\
0, & \mbox{otherwise}
\end{cases}
\end{equation}
and
\begin{equation}
\begin{array}{@{}r@{\ }r@{\ }l@{}}
\widehat{u}_{1}(t) & = & \min\{L, \sum_{r=1}^{N} \tilde{u}_{r}(t)\} \\
\widehat{u}_{n}(t) & = & \min\{L, \sum_{r=1}^{N} \tilde{u}_{r}(t) - \sum_{r=1}^{n-1} \tilde{u}_{r}(t)\}
\end{array}
\end{equation}
Note that we have sliced $\sum_{r=1}^{N} \tilde{u}_{r}(t)$ into $N$ layers $(\widehat{u}_n(t))_{n=1}^N$ in the same manner as $(a^{{\rm ly\mbox{-}}n}(t))_{n=1}^N$.

It is straightforward to see that
\begin{equation}
\sum_{r=1}^{N} \tilde{y}_{r}(t) = \sum_{r=1}^{N} \hat{y}_{r}(t) \mbox{\ and\ }
\sum_{r=1}^{N} \tilde{u}_{r}(t) = \sum_{r=1}^{N} \hat{u}_{r}(t)
\end{equation}
Furthermore, we define
\[
\widehat{v}_{n}(t) = \big[a^{{\rm ly\mbox{-}}n}(t)-\widehat{u}_n(t)\big]^{+}, \
\widehat{s}_{n}(t) = \big[h^{{\rm ly\mbox{-}}n}(t)-\eta\cdot\widehat{u}_n(t)\big]^{+}
\]
It follows that $(\widehat{y}_n, \widehat{u}_n, \widehat{v}_n, \widehat{s}_n)$ is feasible for $\textbf{fMCMP}_{\rm s}^{{\rm ly\mbox{-}}n}$.

Finally, by Lemmas~\ref{lem:ngen-v}-\ref{lem:ngen-beta}, we can show that
\begin{equation}
{\rm Cost}(\tilde{y},\tilde{u},\tilde{v},\tilde{s}) \ge
{\rm Cost}(\widehat{y},\widehat{u},\widehat{v},\widehat{s})
\end{equation}
Therefore, it completes the proof.
\end{proof}
\begin{lem} \label{lem:ngen-v}
\begin{equation}
\sum_{r=1}^{N} \widehat{v}_{r}(t) \le \sum_{r=1}^{N} \tilde{v}_{r}(t) = \tilde{v}(t)
\end{equation}
\begin{equation}  \label{eqn:ngen-s}
\sum_{r=1}^{N} \widehat{s}_{r}(t) \le \sum_{r=1}^{N} \tilde{s}_{r}(t) = \tilde{s}(t)
\end{equation}
\end{lem}
\begin{proof}
First, we note that $((\tilde{y}_{n}, \tilde{u}_{n})_{n=1}^{N}, \tilde{v}, \tilde{s})$ is an optimal solution for $\textbf{fMCMP}$. Hence, $a(t) \ge \sum_{r=1}^{N} \tilde{u}_{r}(t)$ for all $t$. Then,
\begin{equation}
a^{{\rm ly\mbox{-}}n}(t) \ge \widehat{u}_n(t)
\end{equation}
Because we slice $\sum_{r=1}^{N} \tilde{u}_{r}(t)$ into $N$ layers $(\widehat{u}_n(t))_{n=1}^N$ in a same manner as $a(t)$ into $(a^{{\rm ly\mbox{-}}n}(t))_{n=1}^N$.
Hence,
\begin{eqnarray}
& & \sum_{r=1}^{N} \widehat{v}_{r}(t) = \sum_{r=1}^{N} \big[a^{{\rm ly\mbox{-}}r}(t)-\widehat{u}_r(t)\big]^{+} \\
& = & \sum_{r=1}^{N} \Big(a^{{\rm ly\mbox{-}}r}(t)-\widehat{u}_r(t)\Big) = a(t) - \sum_{r=1}^{N} \widehat{u}_r(t) \\
& = & a(t) - \sum_{r=1}^{N} \tilde{u}_r(t) \le  \sum_{r=1}^{N}  \tilde{v}_{r}(t) = \tilde{v}(t)
\end{eqnarray}
where the last inequality follows constraint (C'4).

Similarly, we can prove Eqn.~(\ref{eqn:ngen-s}) considering (C'5).
\end{proof}

\begin{lem} \label{lem:ngen-beta}
\begin{equation}
\sum_{r=1}^{N} [\widehat{y}_r(t) - \widehat{y}_r(t\mbox{-}1)]^{+}  \le
\sum_{r=1}^{N} [\tilde{y}_r(t) - \tilde{y}_r(t\mbox{-}1)]^{+}
\end{equation}
\end{lem}
\begin{proof}
First, note that $\widehat{y}_1(t) \ge ... \ge \widehat{y}_N(t)$ is a decreasing sequence. Because $\widehat{y}_n(t) \in \{0, 1\}$ for all $n$, we obtain
\begin{eqnarray}
 & &  \sum_{r=1}^{N} [\widehat{y}_r(t) - \widehat{y}_r(t\mbox{-}1)]^{+}  \\
& = & \begin{cases}
  0,  \mbox{\qquad \quad if\ } \sum_{r=1}^{N} \widehat{y}_r(t) \le \sum_{r=1}^{N} \widehat{y}_r(t\mbox{-}1) \\
   \sum_{r=1}^{N} \widehat{y}_r(t) - \sum_{r=1}^{N} \widehat{y}_r(t\mbox{-}1), \mbox{\ \ \ otherwise }
 \end{cases} \\
 & = &  \Big[ \sum_{r=1}^{N} \widehat{y}_r(t) - \sum_{r=1}^{N} \widehat{y}_r(t\mbox{-}1) \Big]^{+} \\
 & = & \Big[ \sum_{r=1}^{N} \tilde{y}_r(t) - \sum_{r=1}^{N} \tilde{y}_r(t\mbox{-}1) \Big]^{+}  \\
 & \le & \Big[  \sum_{r=1}^{N} [\tilde{y}_r(t) - \tilde{y}_r(t\mbox{-}1)]^{+}  \Big]^{+}
= \sum_{r=1}^{N} [\tilde{y}_r(t) - \tilde{y}_r(t\mbox{-}1)]^{+}  \notag
\end{eqnarray}
\end{proof}

{\rev {\section{Proof of Theorem 6}\label{subsec:MultipleRatio}
\setcounter{thm}{5}

\begin{thm} \label{thm:CHASE-NG}
The competitive ratio of ${\sf CHASE}^{\rm lk(\omega)}$ satisfies
\begin{equation}
{\sf CR}({\sf CHASE}^{\rm lk(\omega)})\le \min (3-2\cdot g(\alpha,\omega), 1/\alpha), \label{eq:CHASE-N_Ratio}
\end{equation}
where {$\alpha\in(0,1]$ is defined in \eqref{eq:alpha_def}} and $g(\alpha,\omega)\in [\alpha, 1]$ is defined in \eqref{eq:g_alpha_omega}.
\end{thm}

\begin{proof}
We denote the outcome of $\mathbf{CHASE^{lk(\omega)}}$ be
\[
\left(y^{on},u^{on},v^{on},s^{on}\right)=\left(y_{n}^{on},u_{n}^{on},v_{n}^{on},s_{n}^{on}\right)_{n=1}^{N}.
\]
 The cost of $\mathbf{CHASE^{lk(\omega)}}$ can be expressed in the
following way:
\begin{eqnarray}
 &  & \mathrm{Cost}\left(y^{on},u^{on},v^{on},s^{on}\right)\nonumber \\
 & = & \sum_{t=1}^{T}\left\{ p(t)\left(\sum_{n=1}^{N}v_{n}^{on}(t)\right)+c_{g}\left(\sum_{n=1}^{N}s_{n}^{on}(t)\right)+\right.\nonumber \\
 &  & \left.\sum_{n=1}^{N}\left(c_{o}\cdot u_{n}^{on}(t)+y_{n}^{on}(t)\cdot c_{m}+\beta\cdot\left[y_{n}^{on}(t)-y_{n}^{on}(t-1)\right]^{+}\right)\right\} \nonumber \\
 & = & \sum_{n=1}^{N}\sum_{t=1}^{T}\left\{ p(t)v_{n}^{on}(t)+c_{g}\cdot s_{n}^{on}(t)+\right.\nonumber \\
 &  & \left.c_{o}\cdot u_{n}^{on}(t)+y_{n}^{on}(t)\cdot c_{m}+\beta\cdot\left[y_{n}^{on}(t)-y_{n}^{on}(t-1)\right]^{+}\right\} \nonumber \\
 & = & \sum_{n=1}^{N}\mathrm{Cost}\left(y_{n}^{on},u_{n}^{on},v_{n}^{on},s_{n}^{on}\right),\label{eq:exchangeSum}
\end{eqnarray}

where $\left(y_{n}^{on},u_{n}^{on},v_{n}^{on},s_{n}^{on}\right)$
is the online solution by $\mathbf{CHASE_{s}^{lk(\omega)}}$for sub-problem
$\mathbf{fMCMP_{s}^{ly-n}}.$

Based on Theorem 5, we know the optimal offline solution, denoted
as $(y^{*},u^{*},v^{*},s^{*})$ can be expressed as:

\[
(y^{*},u^{*},v^{*},s^{*})=(\sum_{n=1}^{N}y_{n}^{*},\sum_{n=1}^{N}u_{n}^{*},\sum_{n=1}^{N}v_{n}^{*},\sum_{n=1}^{N}s_{n}^{*}),
\]

where $(y_{n}^{*},u_{n}^{*},v_{n}^{*},s_{n}^{*})$ is the optimal
offline solution for sub-problem $\mathbf{fMCMP_{s}^{ly-n}}.$ Similar
as (\ref{eq:exchangeSum}), we have
\begin{eqnarray*}
\mathrm{Cost}(y_{n}^{*},u_{n}^{*},v_{n}^{*},s_{n}^{*}) & = & \sum_{n=1}^{N}\mathrm{Cost}(y_{n}^{*},u_{n}^{*},v_{n}^{*},s_{n}^{*}).
\end{eqnarray*}

Next, from Corollary 1, we know
\begin{equation}
\mathrm{Cost}\left(y_{n}^{on},u_{n}^{on},v_{n}^{on},s_{n}^{on}\right)\leq c\cdot\mathrm{Cost}(y_{n}^{*},u_{n}^{*},v_{n}^{*},s_{n}^{*}),\ \forall n,\label{eq:singleraito->wholeratio}
\end{equation}

where $c$ is the competitive ratio we achieved in Corollary 1.

Thus, by summing up $N$ inequalities (\ref{eq:singleraito->wholeratio}),
we get
\[
\mathrm{Cost}\left(y^{on},u^{on},v^{on},s^{on}\right)\leq c\cdot\mathrm{Cost}(y^{*},u^{*},v^{*},s^{*}).
\]

It completes the proof.
\end{proof}

}}

\section{Proof of Theorem 7}\label{subsec:slowratio}
\setcounter{thm}{6}
\begin{thm}
The competitive ratio of $\mathrm{{\sf CHASE}_{gen}^{lk(w)}}$ is upper bounded
by $(3-2g(\alpha,\omega))\cdot\max\big(r_{1},r_{2}\big)$, where
\begin{eqnarray*}
r_{1} & = & 1+\max\left\{ \frac{\big(P_{\max}+c_{g}\cdot\eta-c_{0}\big)}{Lc_{0}+c_{m}}\max\left\{ 0,\big(L-{\sf R}_{{\rm up}}\big)\right\} \right.\\
 &  & \left.\frac{c_{o}}{c_{m}}\max\left\{ 0,\big(L-{\sf R}_{{\rm dw}}\big)\right\} \right\} \\
r_{2} & = & 1+\frac{L\big(P_{\max}+c_{g}\cdot\eta\big)+c_{m}}{\beta}{\sf T}_{{\rm on}} +\frac{L\big(P_{\max}+c_{g}\cdot\eta\big)}{\beta}{\sf T}_{{\rm off}}
\end{eqnarray*}
\end{thm}

\begin{proof}
The competitive ratio can be expressed as follows:

\begin{eqnarray}
{\sf CR}({\sf CHASE}_{gen}) & = & \max_{\sigma}\frac{{\rm Cost}_{{\sf CHASE}_{gen}}(\sigma)}{{\rm Cost}_{\sf Opt}(\sigma)}\nonumber \\
 & \leq & \max_{\sigma}\frac{{\rm Cost}_{{\sf CHASE}_{gen}}(\sigma)}{{\rm Cost}_{\sf iOpt}(\sigma)}\nonumber \\
 & = & \max_{\sigma}\frac{{\rm Cost}_{{\sf CHASE}_{gen}}(\sigma)}{{\rm Cost}_{{\sf CHASE}_{s}}(\sigma)}\cdot\frac{{\rm Cost}_{{\sf CHASE}_{s}}(\sigma)}{{\rm Cost}_{\sf iOpt}(\sigma)}.\label{eq:non-ideal-ratio}
\end{eqnarray}

In Eqn.~(\ref{eq:non-ideal-ratio}), we have the following notations:
\begin{itemize}
\item ${\rm Cost}_{{\sf CHASE}_{gen}}(\sigma)$: The cost of online algorithm $\mathrm{{\sf CHASE}_{gen}}=\big(y_{2},u_{2},v_{2},s_{2}\big)$
for \textbf{MCMP} with input $\sigma$

\item ${\rm Cost}_{\sf Opt}(\sigma)$: The cost of offline optimal algorithm
for \textbf{MCMP} with input $\sigma$

\item ${\rm Cost}_{\sf iOpt}(\sigma)$: The cost of offline optimal algorithm
for \textbf{fMCMP} with input $\sigma$

\item ${\rm Cost}_{{\sf CHASE}_{s}}(\sigma)$: The cost of online algorithm $\mbox{CHASEs}=\big(y_{s},u_{s},v_{s},s_{s}\big)$
for\textbf{ fMCMP} with input $\sigma$
\end{itemize}
Now, we analyze each term in Eqn.~(\ref{eq:non-ideal-ratio}).

First, as we prove in Theorem \ref{thm:CHASE-competitive-ratio},
\begin{eqnarray*}
 &  & \frac{{\rm Cost}_{{\sf CHASE}_{s}}(\sigma)}{{\rm Cost}_{\sf iOpt}(\sigma)}\\
 & \leq & 1+\frac{2\beta L\big(P_{\max}+\eta\cdot c_{g}-c_{o}-\frac{c_{m}}{L}\big)}{\big(P_{\max}+\eta\cdot c_{g}\big)\big(\beta L+w\cdot c_{m}\big(L-\frac{c_{m}}{P_{\max}+\eta\cdot c_{g}-c_{o}}\big)\big)}\\
 &  & =r_{1}(w)
\end{eqnarray*}

Second, upper-bounding the term $\frac{{\rm Cost}_{{\sf CHASE}_{gen}}(\sigma)}{{\rm Cost}_{{\sf CHASE}_{s}}(\sigma)}$
We divide the cost of ${\sf CHASE}_{s}$ into two parts:

\[
{\rm Cost}_{{\sf CHASE}_{s}}(\sigma)={\rm Cost}_{{\sf CHASE}_{s}}^{e}(\sigma)+{\rm Cost}_{{\sf CHASE}_{s}}^{n}(\sigma)
\]

where:
\[
{\rm Cost}_{{\sf CHASE}_{s}}^{e}(\sigma)=\sum_{t\in\mathbf{T_{e}}}c_{o}u_{s}(t)+p(t)v_{s}(t)+c_{g}s_{s}(t)+c_{m}y_{s}(t)
\]

\begin{eqnarray*}
{\rm Cost}_{{\sf CHASE}_{s}}^{n}(\sigma) & = & \sum_{t\in\mathbf{T_{n}}}c_{o}u_{s}(t)+p(t)v_{s}(t)+c_{g}s_{s}(t)+c_{m}y_{s}(t)\\
 &  & +\sum_{t=1}^{T}\beta\big[y_{s}(t)-y_{s}(t-1)\big]^{+}.
\end{eqnarray*}
and $\mathbf{T_{e}}=\left\{ t|y_{s}(t)=y_{2}(t)\right\} ,\ \mathbf{T_{n}}=\left\{ t|y_{s}(t)\neq y_{2}(t)\right\} $

Similarly, we divide the cost of $\mathrm{{\sf CHASE}_{gen}}$ into two parts:
\[
{\rm Cost}_{{\sf CHASE}_{gen}}(\sigma)={\rm Cost}_{{\sf CHASE}_{gen}}^{e}(\sigma)+{\rm Cost}_{{\sf CHASE}_{gen}}^{n}(\sigma)
\]

Therefore,
\begin{eqnarray*}
\frac{{\rm Cost}_{{\sf CHASE}_{gen}}(\sigma)}{{\rm Cost}_{{\sf CHASE}_{s}}(\sigma)} & \leq & \max\big(\frac{{\rm Cost}_{{\sf CHASE}_{gen}}^{e}(\sigma)}{{\rm Cost}_{{\sf CHASE}_{s}}^{e}(\sigma)},\frac{{\rm Cost}_{{\sf CHASE}_{gen}}^{n}(\sigma)}{{\rm Cost}_{{\sf CHASE}_{s}}^{n}(\sigma)}\big)
\end{eqnarray*}

Next we will prove
\[
\frac{{\rm Cost}_{{\sf CHASE}_{gen}}^{e}(\sigma)}{{\rm Cost}_{{\sf CHASE}_{s}}^{e}(\sigma)}\leq r_{1}
\]

Note that
\begin{eqnarray}
 &  & \frac{{\rm Cost}_{{\sf CHASE}_{gen}}^{e}(\sigma)}{{\rm Cost}_{{\sf CHASE}_{s}}^{e}(\sigma)}\nonumber \\
 & = & \frac{\sum_{t\in\mathbf{T_{e}}}\big(c_{o}u_{2}(t)+p(t)v_{2}(t)+c_{m}y_{2}(t)+c_{g}s_{2}(t)\big)}{\sum_{t\in\mathbf{T_{e}}}\big(c_{o}u_{s}(t)+p(t)v_{s}(t)+c_{m}y_{s}(t)+c_{g}s_{s}(t)\big)}\\
 & \leq & \max_{\forall t\in\mathbf{T_{e}}}\frac{c_{o}u_{2}(t)+p(t)v_{2}(t)+c_{m}y_{2}(t)+c_{g}s_{2}(t)}{c_{o}u_{s}(t)+p(t)v_{s}(t)+c_{m}y_{s}(t)+c_{g}s_{s}(t)}\label{eq:one-slot ratio}
\end{eqnarray}

(\ref{eq:one-slot ratio}) says to build a upper bound of ${\rm Cost}_{{\sf CHASE}_{gen}}^{e}(\sigma)/{\rm Cost}_{{\sf CHASE}_{s}}^{e}(\sigma)$
over time slots, we only need to consider the maximum ratio on a single
time slot. When $y_{2}(t)=y_{s}(t)=0$, it is easy to see $u_{2}(t)=u_{s}(t)=0,\ v_{2}(t)=v_{s}(t),\ s_{2}(t)=s_{s}(t)\rightarrow {\rm Cost}_{{\sf CHASE}_{gen}}^{e}(\sigma)/{\rm Cost}_{{\sf CHASE}_{s}}^{e}(\sigma)=1$
Thus, we only consider the situation when $y_{2}(t)=y_{s}(t)=1:$

\textbf{Case 1}: $u_{s}(t)<u_{2}(t-1)$ For $\mathrm{{\sf CHASE}_{gen}}$:
\begin{eqnarray*}
u_{2}(t) & = & u_{2}(t-1)-\min\big({\sf R}_{{\rm dw}},u_{2}(t-1)-u_{s}(t)\big)\\
 & = & \max\left\{ u_{2}(t-1)-R_{dw},u_{s}(t)\right\} \\
 & \leq & \max\left\{ L-R_{dw},u_{s}(t)\right\}
\end{eqnarray*}

Therefore,
\begin{eqnarray}
 &  & \frac{c_{o}u_{2}(t)+p(t)v_{2}(t)+c_{m}y_{2}(t)+c_{g}s_{2}(t)}{c_{o}u_{s}(t)+p(t)v_{s}(t)+c_{m}y_{s}(t)+c_{g}s_{s}(t)}\nonumber \\
 & = & \frac{c_{o}u_{2}(t)+p(t)v_{2}(t)+c_{g}s_{2}(t)+c_{m}}{c_{o}u_{s}(t)+p(t)v_{s}(t)+c_{g}s_{s}(t)+c_{m}}\\
 & \leq & 1+\frac{c_{o}\big(u_{2}(t)-u_{s}(t)\big)}{c_{m}+c_{o}u_{s}(t)}\nonumber \\
 &  & \frac{p(t)\big(v_{2}(t)-v_{s}(t)\big)+c_{g}\big(s_{2}(t)-s_{s}(t)\big)}{c_{m}+c_{o}u_{s}(t)}\\
 & \leq & 1+\frac{c_{o}\max\left\{ L-R_{dw}-u_{s}(t),0\right\} }{c_{m}+c_{o}u_{s}(t)}\label{eq:v2>vs, s2>ss}\\
 & = & 1+\frac{c_{o}}{c_{m}}\max\left\{ L-R_{dw},0\right\}
\end{eqnarray}

Eqn.~(\ref{eq:v2>vs, s2>ss}) is from $v_{2}(t)\leq v_{s}(t),\ s_{2}(t)\leq s_{s}(t)$,
this is because $u_{2}(t)\geq u_{s}(t)$

\textbf{Case 2}: $u_{s}(t)\geq u_{2}(t-1)$ For $\mathrm{{\sf CHASE}_{gen}}$:
\begin{eqnarray*}
u_{2}(t) & = & u_{2}(t-1)+\min\big({\sf R}_{{\rm up}},u_{s}(t)-u_{2}(t-1)\big)\\
 & = & \min\left\{ u_{2}(t-1)+{\sf R}_{{\rm up}},u_{s}(t)\right\}
\end{eqnarray*}

Therefore,
\begin{eqnarray}
 &  & \frac{c_{o}u_{2}(t)+p(t)v_{2}(t)+c_{m}+c_{g}s_{2}(t)}{c_{o}u_{s}(t)+p(t)v_{s}(t)+c_{m}+c_{g}s_{s}(t)}\nonumber \\
 & = & 1+\frac{c_{o}\big(u_{2}(t)-u_{s}(t)\big)}{c_{o}u_{s}(t)+c_{m}}\nonumber \\
 &  & \frac{p(t)\big(v_{2}(t)-v_{s}(t)\big)+c_{g}\big(s_{2}(t)-s_{s}(t)\big)}{c_{o}u_{s}(t)+c_{m}}\\
 & \leq & \frac{\big(p(t)+c_{g}\eta-c_{o}\big)\big(u_{s}(t)-u_{2}(t)\big)}{c_{o}u_{s}(t)+c_{m}}+1\label{eq:v-v<u-u}\\
 & \leq & 1+\big(P_{\mathrm{max}}+c_{g}\eta-c_{o}\big)\frac{\max\big(0,u_{s}(t)-u_{2}(t-1)-{\sf R}_{{\rm up}}\big)}{c_{o}u_{s}(t)+c_{m}}\nonumber \\
 & \leq & 1+\big(P_{\mathrm{max}}+c_{g}\eta-c_{o}\big)\frac{\max\big(0,u_{s}(t)-{\sf R}_{{\rm up}}\big)}{c_{o}u_{s}(t)+c_{m}}\nonumber \\
 & \leq & 1+\big(P_{\mathrm{max}}+c_{g}\eta-c_{o}\big)\frac{\max\big(0,L-{\sf R}_{{\rm up}}\big)}{Lc_{o}+c_{m}}\nonumber
\end{eqnarray}

Eqn.~(\ref{eq:v-v<u-u}) is from $v_{2}(t)-v_{s}(t)\leq u_{s}(t)-u_{2}(t),\ s_{2}(t)-s_{s}(t)\leq\eta\cdot\big(u_{s}(t)-u_{2}(t)\big)$
Note now we have $u_{s}\geq u_{2}:$
\begin{eqnarray*}
v_{2}-v_{s} & = & \big[a-u_{2}\big]^{+}-\big[a-u_{s}\big]^{+}\\
 & = & \max\big(a-u_{2},0\big)+\min\big(u_{s}-a,0\big)\\
 & = & \begin{cases}
a-u_{2} & if\ u_{s}>a>u_{2}\\
0 & if\ u_{s}>u_{2}>a\\
u_{s}-a & if\ a>u_{s}>u_{2}
\end{cases}\\
 & \leq & u_{s}-u_{2}
\end{eqnarray*}

Similarly we prove $s_{2}(t)-s_{s}(t)\leq\eta\cdot\big(u_{s}(t)-u_{2}(t)\big)$

Summarizing the above two cases, we conclude that
\begin{eqnarray*}
 &  & \frac{{\rm Cost}_{{\sf CHASE}_{gen}}^{e}(\sigma)}{{\rm Cost}_{{\sf CHASE}_{s}}^{e}(\sigma)}\\
 & \leq & 1+\max\left\{ \big(P_{\mathrm{max}}+c_{g}\eta-c_{o}\big)\frac{\max\big(0,L-{\sf R}_{{\rm up}}\big)}{Lc_{o}+c_{m}},\right.\\
 &  & \left.\frac{c_{o}}{c_{m}}\max\left\{ L-R_{dw},0\right\} \right\}
\end{eqnarray*}

Now, we still need to upper bound the term $\frac{{\rm Cost}_{{\sf CHASE}_{gen}}^{n}(\sigma)}{{\rm Cost}_{{\sf CHASE}_{s}}^{n}(\sigma)}$
Note the set $\mathbf{T_{n}}$ represents the time durations when
${\sf CHASE}_{gen}$ and ${\sf CHASE}_{s}$ have different on/off status,
and this can only occur on the minimum on/off periods: $\mathrm{T_{on}/T_{off}}$
That is: when ${\sf CHASE}_{s}$ has a startup, ${\sf CHASE}_{gen}$ can
not startup the generator because the generator is during its minimum
off periods ${\sf T}_{{\rm off}}$ Similarly, we know such mismatch
of on/off status also occurs during the minimum on periods $T_{on}$
of ${\sf CHASE}_{gen}$.

As ${\sf CHASE}_{gen}$ always follows ${\sf CHASE}_{s}$ to startup, thus
${\sf CHASE}_{gen}$ has at most the same number of startups as ${\sf CHASE}_{s}$.
If we denote the number of the startup of $\mathrm{{\sf CHASE}_{s}}$
and ${\sf CHASE}_{gen}$ as $k_{s}$ and $k_{2}$, we have%
\footnote{Without loss of generality, we assume $k_{s}\geq1$ Otherwise, $k_{s}=0$,
then $k_{2}=0$, thus ${\rm Cost}_{{\sf CHASE}_{gen}}^{n}(\sigma)/{\rm Cost}_{{\sf CHASE}_{s}}^{n}(\sigma)=1$%
}:
\[
k_{2}\leq k_{s}
\]

Therefore, we can find a lower bound of ${\rm Cost}_{{\sf CHASE}_{s}}^{n}(\sigma):$
\[
{\rm Cost}_{{\sf CHASE}_{s}}^{n}(\sigma)\geq k_{s}\cdot\beta.
\]

We cam also find an upper bound of $\mathrm{{\rm Cost}_{{\sf CHASE}_{2}}^{n}(\sigma)}$,
by assuming on the minimum on/off periods incurring the maximum cost:
\begin{eqnarray*}
 &  & {\rm Cost}_{{\sf CHASE}_{2}}^{n}(\sigma)\\
 & \leq & k_{2}\big(\beta+\mathrm{\textrm{\mbox{\ensuremath{T_{on}}}}}(LP_{\max}+\eta c_{g}L+c_{m})+\mbox{\ensuremath{{\sf T}_{{\rm off}}}}L\big(P_{\max}+\eta c_{g}\big)\big)
\end{eqnarray*}

Thus,
\begin{eqnarray*}
 &  & \frac{{\rm Cost}_{{\sf CHASE}_{2}}^{n}(\sigma)}{{\rm Cost}_{{\sf CHASE}_{s}}^{n}(\sigma)}\\
 & \leq & \frac{k_{2}\big(\beta+\mathrm{\textrm{\mbox{\ensuremath{T_{on}}}}}(LP_{\max}+\eta c_{g}L+c_{m})+\mbox{\ensuremath{{\sf T}_{{\rm off}}}}L\big(P_{\max}+\eta c_{g}\big)\big)}{k_{s}\cdot\beta}\\
 & \leq & 1+\frac{\mathrm{\textrm{\mbox{\ensuremath{T_{on}}}}}(LP_{\max}+\eta c_{g}L+c_{m})+\mbox{\ensuremath{{\sf T}_{{\rm off}}}}L\big(P_{\max}+\eta c_{g}\big)}{\beta}\\
 & = & r_{2}
\end{eqnarray*}

It completes the proof.
\end{proof}

\end{document}